\documentclass[10pt, a4paper, fleqn]{article}
\usepackage{etex}
\usepackage[utf8]{inputenc}
\usepackage[TS1,T1]{fontenc}
\usepackage{xspace}
\usepackage[a4paper,tmargin=3truecm,bmargin=3truecm,hmargin=3.5truecm]{geometry}

\usepackage[english]{babel}

\usepackage{amsfonts}
\usepackage{lmodern}

\usepackage{enumitem}
\newlist{enum-hypothesis}{enumerate}{1}
\setlist[enum-hypothesis]{label=(\arabic*),itemsep=0pt, parsep=0pt}

\usepackage{booktabs, dcolumn}
\usepackage{graphicx, subfig}

\usepackage{pbox}

\usepackage{amsmath,mathtools,slashed,mathdots}
\usepackage[thmmarks,amsmath]{ntheorem}

\newtheorem{theorem}{Theorem}[section]
\newtheorem{proposition}[theorem]{Proposition}
\newtheorem{lemma}[theorem]{Lemma}
\newtheorem{corollary}[theorem]{Corollary}

\newtheorem{definition}[theorem]{Definition}

\theoremsymbol{$\square$}
\theoremstyle{plain}
\theorembodyfont{\upshape}
\newtheorem{remark}[theorem]{Remark}
\theoremsymbol{$\Diamond$}
\newtheorem{example}[theorem]{Example}

\theoremsymbol{}
\theoremstyle{break}
\theorembodyfont{\slshape}

\theoremheaderfont{\scshape}
\theorembodyfont{\upshape} 
\theoremstyle{nonumberplain}
\theoremseparator{} 
\theoremsymbol{\rule{1.3ex}{1.3ex}} 
\newtheorem{proof}{Proof}

\usepackage[square,numbers,compress,merge]{natbib}
\usepackage{hyperref}

\hypersetup{
plainpages=false,
colorlinks=true,
linkcolor=black, 
anchorcolor=black, 
citecolor=black, 
urlcolor=black, 
menucolor=black, 
filecolor=black, 
bookmarksopen=true,
bookmarksnumbered=true}

\hypersetup{
pdfauthor={Thierry Masson, Gaston Nieuviarts},
pdftitle={Derivation-based Noncommutative Field Theories on AF algebras},
pdfsubject={},
pdfcreator={pdflatex},
pdfproducer={pdflatex},
pdfkeywords={}
}

\newcommand{\bbC}{\mathbb{C}}

\newcommand{\bbN}{\mathbb{N}}
\newcommand{\bbR}{\mathbb{R}}

\newcommand\bbZ{\mathbb{Z}}

\newcommand{\bbbone}{{\text{\usefont{U}{bbold}{m}{n}\char49}}} 
\newcommand{\bbbzero}{{\text{\usefont{U}{bbold}{m}{n}\char48}}} 
\newcommand{\hbbbone}{\widehat{\bbbone}}
\newcommand{\boldone}{\mathbf{1}}
\newcommand{\boldzero}{\mathbf{0}}

\newcommand{\calA}{\mathcal{A}}
\newcommand{\calB}{\mathcal{B}}

\newcommand{\calD}{\mathcal{D}}

\newcommand{\calG}{\mathcal{G}}

\newcommand{\calL}{\mathcal{L}}
\newcommand{\calM}{\mathcal{M}}
\newcommand{\calN}{\mathcal{N}}
\newcommand{\calP}{\mathcal{P}}

\newcommand{\calS}{\mathcal{S}}
\newcommand{\calU}{\mathcal{U}}
\newcommand{\calZ}{\mathcal{Z}}

\newcommand{\kX}{\mathfrak{X}}
\newcommand{\kY}{\mathfrak{Y}}

\newcommand{\ksl}{\mathfrak{sl}}
\newcommand{\ksu}{\mathfrak{su}}

\newcommand{\kS}{\mathfrak{S}}

\newcommand{\algA}{\calA}
\newcommand{\algB}{\calB}

\newcommand{\modM}{\calM}
\newcommand{\modN}{\calN}

\newcommand{\hstar}{\star} 
\newcommand{\grast}{\bullet} 
\newcommand{\omi}[1]{\buildrel { \buildrel{#1}\over{\vee} } \over .} 
\newcommand\exter{{\textstyle\bigwedge}} 

\newcommand\cdotaction{\mathord{\cdot}}
\newcommand\halgA{\widehat{\algA}}
\newcommand\halgB{\widehat{\algB}}

\newcommand\tphi{\widetilde{\phi}}

\newcommand{\defeq}{\vcentcolon=} 

\DeclareMathOperator{\tsum}{\textstyle\sum}
\DeclareMathOperator{\toplus}{\textstyle\oplus}

\DeclareMathOperator{\ad}{ad}

\DeclareMathOperator{\card}{card}
\DeclareMathOperator{\Der}{Der}
\DeclareMathOperator{\diag}{diag}

\DeclareMathOperator{\Hom}{Hom} 
\DeclareMathOperator{\Id}{Id}

\DeclareMathOperator{\Int}{Int}

\DeclareMathOperator{\rank}{rank}	 

\DeclareMathOperator{\tr}{tr}	   

\DeclarePairedDelimiter\abs{\lvert}{\rvert}

\newcommand{\dd}{\text{\textup{d}}}
\newcommand{\dddR}{\dd_{\text{dR}}}
\newcommand{\vol}{\text{\textup{vol}}}

\newcommand{\iotaDer}{\iota^{\Der}}
\newcommand{\iotaMod}{\iota^{\text{\textup{Mod}}}}
\newcommand{\piMod}{\pi_{\text{\textup{Mod}}}}
\newcommand{\phiMod}[1][]{\phi_{\text{\textup{Mod}}#1}}

\newcommand{\OmegaDer}{\Omega_{\Der}}
\newcommand{\PsiDer}{\Psi_{\Der}}
\newcommand{\PsiMod}{\Psi_{\text{\textup{Mod}}}}

\newcommand{\mrnabla}{\mathring{\nabla}}
\newcommand{\mromega}{\mathring{\omega}}

\newcounter{mnotecount}[section]
\renewcommand{\themnotecount}{\thesection.\arabic{mnotecount}}
\newcommand{\mnote}[1]%
{\protect{\stepcounter{mnotecount}}${}^{\text{\footnotesize$\bullet$\themnotecount}}$%
\reversemarginpar%
\marginpar{\raggedleft\footnotesize$\bullet$\themnotecount: #1}}

\newlength{\mnotewidth}
\numberwithin{equation}{section}
\allowdisplaybreaks

\begin{document}
\renewcommand\figurename{Fig.}

{
\makeatletter\def\@fnsymbol{\@arabic}\makeatother 
\title{Derivation-based\\
Noncommutative Field Theories\\ 
on $AF$ algebras}
\author{T. Masson, G. Nieuviarts\\
{\small Centre de Physique Théorique}%
\\
\small{Aix Marseille Univ, Université de Toulon, CNRS, CPT, Marseille, France}\\[2ex]
}
\date{}
\maketitle
}

\begin{abstract}
In this paper, we start the investigation of a new natural approach to “unifying” noncommutative gauge field theories (NCGFT) based on approximately finite-dimensional ($AF$) $C^*$-algebras. The defining inductive sequence of an $AF$ $C^*$-algebra is lifted to enable the construction of a sequence of NCGFT of Yang-Mills-Higgs types. The present paper focus on derivation-based noncommutative field theories. A mathematical study of the ingredients involved in the construction of a NCGFT is given in the framework of $AF$ $C^*$-algebras: derivation-based differential calculus, modules, connections, metrics and Hodge $\star$-operators, Lagrangians… Some physical applications concerning mass spectra generated by Spontaneous Symmetry Breaking Mechanisms (SSBM) are proposed using numerical computations for specific situations.
\end{abstract}

\tableofcontents

\section{Introduction}

Gauge field theories are an essential ingredient to model high energy particle physics. From the pioneer work by Yang and Mills to the Standard Model of Particle Physics (SMPP), the gauge principle has shown how insightful it is both technically and conceptually, and the search for the “right gauge group” has stimulated physicists to construct Grand Unified Theories (GUT). Unfortunately, none of these theories has been retained until now as a convincing model beyond the SMPP. 

The way these GUT are constructed relies on the classical mathematics of fiber bundles and connections: the gauge groups (infinite dimensional spaces) are the groups of vertical automorphisms of principal fibers over space-time with some convenient structure groups (finite dimensional Lie groups). The choice of the structure group is then the choice of the gauge group of the theory. GUT rely on finite dimensional Lie groups which are “big enough”, for instance $SU(5)$, to contain the group of the SMPP, $U(1)\times SU(2) \times SU(3)$. This unifying approach (for interactions) is then controlled by the choice of possible “not too large” finite dimensional Lie groups. This group has to be “not too large” because one has to reduce it to the group of the SMPP that we actually see in experiments, and the larger the original group, the more hypothesis it requires to perform this reduction, usually using some successive Spontaneous Symmetry Breaking Mechanisms (SSBM).

Since the 90's, noncommutative geometry (NCG) has shown that one can construct more general gauge field theories in a natural way (see \cite{ConnLott90a, DuboKernMado90a, DuboKernMado90b} for the seminal papers and \cite{Suij15a} for a review and references for more recent developments). NCG has permitted to include in a natural way the scalar fields used in the SMPP to make manifest the SSBM which gives masses to fermions and some gauge particles.\footnote{Notice that NCG is not the only mathematical framework beyond ordinary fiber geometry to make this happen in a natural way: connections on transitive Lie algebroids share some similar features, see \cite{FourLazzMass13a, LazzMass12a}.} In this approach, the gauge group is the group of inner automorphisms of an algebraic structure, which in general is an associative algebra (it could be the group of automorphisms of a module in some cases).

\medskip
In this paper, we start the investigation of a new natural approach to “unifying” noncommutative gauge field theories (NCGFT). This approach is based on approximately finite-dimensional ($AF$) $C^*$-algebras, a very natural class of algebras in NCG (see \cite{Blac06a, Davi96a, RordLarsLaus00a} for instance). By definition, $AF$ $C^*$-algebras are inductive limits of sequence of finite-dimensional $C^*$-algebras. Let us recall the following two important points (see Sect.~\ref{sec AF algebras} for more details):
\begin{enumerate}
\item A finite-dimensional $C^*$-algebra is, up to isomorphism, a finite sum of matrix algebras: $\algA = M_{n_1} \toplus \cdots \toplus M_{n_r}$ where $M_n \defeq M_n(\bbC)$ is the space of $n \times n$ matrices over $\bbC$. For a manifold $M$, NCGFT have been investigated on the algebras $C^\infty(M) \otimes \algA$ (referred to in the literature as “almost commutative” algebras) and these NCGFT are of Yang-Mills-Higgs types.

\item An $AF$ $C^*$-algebra is constructed in such a way that we get a control of the approximation of this algebra by the successive finite-dimensional $C^*$-algebras in its defining inductive sequence. This control (in terms of $C^*$-norms) can be used to approximate some structures defined on the $AF$ $C^*$-algebra. The best example is the one of the $K_0$-group that we briefly recall in Sect.~\ref{sec AF algebras} for sake of illustration.
\end{enumerate}

The motivation for the present approach can be summarized in the following way: point~1 suggests to use the defining inductive sequence of an $AF$ $C^*$-algebra to construct a sequence of NCGFT of Yang-Mills-Higgs types and  point~2 could be used to get some control of this sequence of NCGFT in a meaningful way as successive approximations of a “unifying” NCGFT on the full $AF$ $C^*$-algebra. This would be a way to implement inclusions of (finite dimensional) “gauge groups” as successive approximations of an (infinite dimensional) “unifying gauge group”.

Notice that the NCGFT that we can define on the full (maybe infinite dimensional) $AF$ $C^*$-algebra can be quite unusual from a physical point of view since it can involve an infinite number of degrees of freedom in the gauge sector. But if the control of approximations by “finite dimensional“ NCGFT (as suggested by point~2) is possible, then the content of this NCGFT could be understood as a limit of usual and manageable Yang-Mills-Higgs theories constructed gradually, for instance using empirical data.

In the GUT approach, the “big enough” gauge group must contain all the empirical phenomenology of present particle physics. But, in our approach, thanks to the approximation procedure, we only require any step in the inductive limit of NCGFT to \emph{approximate} the present empirical phenomenology and the future possible discoveries in particle physics. As the probing energy increases, a better approximating NCGFT in the sequence has to be taken at a farther position. In our bottom-to-top approximation, the number of degrees of freedom (in the gauge sector) can increase along the sequence, contrary to the usual SSBM, which is a top-to-bottom procedure that relies on reduction of degrees of freedom.

Notice that one could encounter a stationary sequence starting at some point, so that the full $AF$ $C^*$-algebra would be of finite dimension. In that case, the “final” NCGFT would make appear only a finite number of degrees of freedom in the gauge sector, very much like ordinary GUT. But even in that situation, some interesting features could be gained by the way the successive NCGFT in the sequence (here finite) of NCGFT are connected to each other, in particular concerning the mass spectra of the gauge bosons, see Sect.~\ref{sect direct limit NCGFT}.

At present time, we are not aware of any empirical fact suggesting that such a radical new approach could be relevant. But we hope that our new way to construct unifying gauge field theories beyond the SMPP could reveal new empirical content that would be suggestive to answer open questions in particle physics. It is out of the scope of the present paper to already get positive insight in that direction.

\medskip
In the present paper, we start this research program in the framework of the derivation-based noncommutative geometry. This very algebraic approach to NCG is convenient for a first study of this new approach to NCGFT on $AF$ $C^*$-algebras since the space of derivations, and so the space of differential forms, is canonically associated to the algebra we consider. This permits to lift in a natural way some properties of the inductive sequence defining the $AF$ $C^*$-algebra to many structures that are needed to construct a NCGFT (see \cite{Mass12a} for instance).  We follow the seminal papers \cite{DuboKernMado90a, DuboKernMado90b} to construct NCGFT at each step of this sequence. The only choices that remain to be done then concern the modules and the noncommutative connections on theses modules. We postpone to a forthcoming paper the exploration of our new NCGFT approach on $AF$ $C^*$-algebras using spectral triples.

\medskip
The paper is organized as follows. Sect.~\ref{sec notations and useful results} is dedicated to recalling all the technical ingredients involved in our constructions: $AF$ $C^*$-algebras, derivation-based differential calculus, NCGFT on algebras of matrices and on algebras of matrix-valued smooth functions on a manifold. Since the algebras in the sequence defining an $AF$ $C^*$-algebra are of the form $\algA = \toplus_{i=1}^{r} M_{n_i}$, it is necessary to study the derivation-based differential calculus of an algebra decomposed in this way, as well as modules, metrics and connections. This is the object of Sect.~\ref{sec some properties on sums of algebras}, where general results are obtained for algebras decomposed as $\algA = \toplus_{i=1}^{r} \algA_i$. Sect.~\ref{sec one step in the sequence} is devoted to the study of the relations between the structures involved in the construction of a NCGFT for an injective map $\phi : \algA \to \algB$ where $\algA = \toplus_{i=1}^{r} M_{n_i}$ and $\algB = \toplus_{j=1}^{s} M_{m_j}$. This is where the main definitions are proposed to connect the structures used to defined a NCGFT on $\algB$ with the similar structures defined on $\algA$, in a so-called $\phi$-compatible way which “extends” in a natural way (as much as possible) the inclusion $\phi : \algA \to \algB$. Finally, in Sect.~\ref{sect direct limit NCGFT}, we consider “direct limits” of NCGFT. As usual in the study of $AF$ $C^*$-algebras (see Sect.~\ref{sec AF algebras} for the example of the $K$-theory group), the characterization of the “induced” NCGFT on the full $AF$ $C^*$-algebra relies on the way a NCGFT on $\algB$ is related (via the $\phi$-compatibilities considered in Sect.~\ref{sec one step in the sequence}) to a NCGFT on $\algA$, where $\algA$ and $\algB$ are two algebras in the sequence defining the $AF$ $C^*$-algebra. So, this situation is described and some general conclusions are drawn. Since the SSBM is an essential feature of the NCGFT considered here, some numerical computations are presented on simple situations (for the choices of $\algA$ and $\algB$) to illustrate the way masses (obtained from the SSBM) can be related on $\algA$ and $\algB$ when $\phi$-compatibility is required. Some technical results are given in two appendices: one concerning the decomposition of the Hodge $\hstar$-operator (referred to in Sect.~\ref{sec some properties on sums of algebras}) and the other one on an explicit construction of an “extended” basis for derivations of an algebra $\algB$ for an basis of derivations for an algebra $\algA$ in the situation of Sect.~\ref{sec one step in the sequence}.

\section{Notations and useful results}
\label{sec notations and useful results}

\subsection{\texorpdfstring{$AF$ $C^*$-algebras}{AF C*-algebras}}
\label{sec AF algebras}

In this section, we would like to recall the necessary structures involved in the definition of $AF$ $C^*$-algebras that will be used in the following. We would like also to illustrate, with the $K_0$-group example, the powerful approximation procedure by finite dimensional structures that we inherit in this framework.

\medskip
A $C^*$-algebra $\algA$ is said to be $AF$ (approximately finite-dimensional) if it is the closure of an increasing union of finite dimensional subalgebras $\algA_n$ \textit{i.e.} $\algA = \overline{\cup_{n\geq 0} \algA_n}$. We will always suppose that $\algA$ is unital and that $\algA_0 \simeq \bbC \bbbone_{\algA}$ \cite{Davi96a}.

It is convenient to describe $\algA$ as the direct limit $\algA = \varinjlim \algA_n$ of the inductive sequence of the finite dimensional (sub)algebras $\{ (\algA_n, \phi_{n,m}) \, / \,  0 \leq n < m \}$ where $\phi_{n,m} : \algA_n \to \algA_{m}$ are injective unital $*$-homomorphisms such that $\phi_{m,p} \circ \phi_{n,m} = \phi_{n,p}$ for any $0 \leq n < m < p$. From this composition property, one needs only to describe the homomorphisms $\phi_{n,n+1} : \algA_n \to \algA_{n+1}$. This can be done in two steps.

Firstly, any finite dimensional $C^*$-algebra $\algA$ is $*$-isomorphic to the direct sum of matrix algebras, $\algA \simeq \toplus_{i=1}^{r} M_{n_i}$ \cite[Thm.~III.1.1]{Davi96a}. Secondly, any unital $*$-homomorphism $\phi : \algA = \toplus_{i=1}^{r} M_{n_i} \to \algB = \toplus_{j=1}^{s} M_{m_j}$ is determined up to unitary equivalence in $\algB$ by a $s \times r$ matrix $A = ( \alpha_{ji} )$ where $\alpha_{ji} \in \bbN$ (non-negative integers) is the multiplicity of the inclusion of $M_{n_i}$ into $M_{m_j}$ \cite[Lemma~III.2.1]{Davi96a}. The multiplicity matrix $A$ is such that $\sum_{i=1}^{r} \alpha_{ji} n_i = m_j$. 

The characterization of $\phi$ up to unitary equivalence in $\algB$ permits to take a convenient presentation of the inclusions of the $M_{n_i}$'s into the $M_{m_j}$'s, for instance by increasing order of the $n_i$'s along the diagonal of $M_{m_j}$. 

\medskip
In this paper, we will not be interested in the $C^*$ aspect of $AF$ $C^*$-algebras since we will focus on the differentiable structures compatible with the increasing sequence of $\{ (\algA_n, \phi_{n,n+1}) \, / \,  n \geq 0 \}$. In our point of view, we will consider $\algA_\infty \defeq \cup_{n\geq 0} \algA_n$ as the dense subalgebra of $\algA = \overline{\cup_{n\geq 0} \algA_n}$ of “smooth” elements. $\algA_\infty$, as the direct limit of $\{ (\algA_n, \phi_{n,n+1}) \, / \,  n \geq 0 \}$ in the category of associative (unital) algebras, inherits some algebraic structures of the algebras $\algA_n$. 

\medskip
Let us now illustrate how the defining sequence $\{ (\algA_n, \phi_{n,n+1}) \, / \,  n \geq 0 \}$ can be used to construct “approximations” of the $K_0$ group of $\algA$. 

The definition (see \cite{RordLarsLaus00a,Davi96a} for details) of the $K_0$ group of a unital $C^*$-algebra $\algA$ starts with an equivalence class of projections in $M_\infty(\algA) \defeq \cup_{p\geq 1} M_p(\algA)$ where $M_p(\algA) \defeq M_p \otimes \algA$ is the $C^*$-algebra of $p\times p$ matrices with entries in $\algA$. We denote by $\calP_\infty(\algA)$ the semigroup of projections in $M_\infty(\algA)$. Two projections $P, Q \in \calP_\infty(\algA)$ are $*$-equivalent, $P \sim Q$, if there is a partial isometry $S \in M_\infty(\algA)$ such that $P = S^* S$ and $Q = S S^*$. The space $\calD(\algA) \defeq \calP_\infty(\algA)/\sim$ of equivalence classes $[P]$ of projections in $M_\infty(\algA)$ is an Abelian semigroup for the additive law $[P] + [Q] \defeq [P \oplus Q]$ where $P \oplus Q \defeq \begin{psmallmatrix} P & 0 \\ 0 & Q \end{psmallmatrix} \in M_\infty(\algA)$. Then $K_0(\algA)$ is the Grothendieck group of $\calD(\algA)$.\footnote{The Grothendieck group $G(S)$ of an Abelian semigroup $(S,+)$ is the unique Abelian group $K$ which satisfies the following universal property: there is a morphism of semigroups $\iota : S \to G(S)$ such that for any morphism of semigroups $\varphi : S \to T$ for any Abelian group $T$, there is a morphism of groups $\widetilde{\varphi} : G(S) \to T$ with $\varphi = \widetilde{\varphi} \circ \iota$.} We denote by $\iota : \calD(\algA) \to K_0(\algA)$ the map defining the universal property of $K_0(\algA)$. Then, for any $P, Q \in \calP_\infty(\algA)$, one has \cite[Prop.~3.1.7]{RordLarsLaus00a} $\iota([P \oplus Q]) = \iota([P]) + \iota([Q])$; $\iota([P]) = \iota([Q])$ if and only if there exists $R \in \calP_\infty(\algA)$ such that $P \oplus R \sim Q \oplus R$; and $K_0(\algA) = \{ \iota([P]) - \iota([Q]) \, / \, P, Q \in \calP_\infty(\algA) \}$. So, describing all the $\iota([P])$ is sufficient to get $K_0(\algA)$.

For a matrix algebra $\algA = M_n$, one has $M_p(M_n) = M_p \otimes M_n = M_{pn}$ and,  for any $P, Q \in \calP_\infty(M_n)$, $P \sim Q$ iff $\rank(P) = \rank(Q)$. So $\calD(M_n) \simeq \bbN$, $K_0(M_n) = \bbZ$ and $\iota$ is the inclusion $\iota : \bbN \hookrightarrow \bbZ$ and so it can be omitted in that case. For a finite dimensional algebra $\algA = \toplus_{i=1}^{r} M_{n_i}$, this result generalizes as $P = \toplus_{i=1}^{r} P_i \sim Q = \toplus_{i=1}^{r} Q_i$ iff $\rank(P_i) = \rank(Q_i)$ for any $i$, so that $\calD(\toplus_{i=1}^{r} M_{n_i}) \simeq \bbN^r$ and $K_0(\toplus_{i=1}^{r} M_{n_i}) = \bbZ^r$ \cite[Ex.~IV.2.1]{Davi96a}. Here again $\iota$ is the natural inclusion and it can be omitted.

Any morphism of $C^*$-algebras $\phi : \algA \to \algB$ induces a canonical morphism of groups $\phi_* : K_0(\algA) \to K_0(\algB)$ by $\phi_* \circ \iota_\algA([P]) = \iota_\algB([\phi(P)])$ where $\phi(P) \in \calP_\infty(\algB) \subset M_\infty(\algB)$ is defined by applying $\phi$ to the entries of the matrix $P \in M_p(\algA) \subset M_\infty(\algA)$. So, from the defining inductive sequence $\{ (\algA_n, \phi_{n,m}) \, / \,  0 \leq n < m \}$ of an $AF$ $C^*$-algebra $\algA = \varinjlim \algA_n$, we get an inductive sequence $\{ (K_0(\algA_n), \phi_{n,m\, *}) \, / \,  0 \leq n < m \}$. Then one has $K_0(\algA) = \varinjlim K_0(\algA_n)$ \cite[Thm~6.3.2]{RordLarsLaus00a}.

To get $\varinjlim K_0(\algA_n)$, one has to describe the morphisms $\phi_{n,n+1\, *} : K_0(\algA_n) \to K_0(\algA_{n+1})$. This can be done easily in terms of the multiplicity matrices $A_{n, n+1}$ associated to the morphisms $\phi_{n, n+1} : \algA_n \to \algA_{n+1}$. In order to do that, we switch from projections to finitely generated projective modules. 

Let $P \in M_p(M_n)$ be a projection with $\rank(P) = \alpha \ (\leq pn)$. Then $P$ can be diagonalized as $P = U^* E_\alpha U$ for a unitary $U \in M_{pn}$ and $E_\alpha = \begin{psmallmatrix} \bbbone_\alpha & 0 \\ 0 & 0 \end{psmallmatrix} \in M_{pn}$ where $\bbbone_\alpha$ is the unit $\alpha \times \alpha$ matrix. Then $S \defeq E_\alpha U$ satisfies $S^* S = (E_\alpha U)^* (E_\alpha U) = U^* E_\alpha^* E_\alpha U = U^* E_\alpha U = P$ and $S S^* = (E_\alpha U) (E_\alpha U)^* = E_\alpha U U^* E_\alpha^* = E_\alpha$ so that $[P] = [E_\alpha]$. Consider the free left $M_n$-module $(M_n)^p \simeq M_{pn, n}$ (row of $p$ copies of $M_n$ or rectangular matrices $pn \times n$). Then up to the unitary equivalence by $U$ (acting on the right on rectangular matrices), $P$ defines the submodule $\modM_P \simeq M_{pn, n} E_\alpha \simeq M_{\alpha, n} \simeq \bbC^n \otimes \bbC^\alpha$ (a finitely generated projective module over $M_n$). 

In the same way, a class $[P = \toplus_{i=1}^{r} P_r] \in \calD(\toplus_{i=1}^{r} M_{n_i})$ defines a class (modulo isomorphisms) of left (finitely generated projective) modules $[\modM_P]$ with $\modM_P = \bbC^{n_1} \otimes \bbC^{\alpha_1} \toplus \cdots \toplus \bbC^{n_r} \otimes \bbC^{\alpha_r}$ where $\alpha_i = \rank(P_i)$. Indeed, if $P \in M_p(\toplus_{i=1}^{r} M_{n_i}) = \toplus_{i=1}^{r} M_p \otimes M_{n_i}$ then $P_i \in M_p \otimes M_{n_i} = M_p(M_{n_i})$ and then we are in the previous situation for every $P_i$. The map $\phi_* : \calD(\toplus_{i=1}^{r} M_{n_i}) \to \calD(\toplus_{j=1}^{s} M_{m_i})$ induced by $\phi : \toplus_{i=1}^{r} M_{n_i} \to \toplus_{j=1}^{s} M_{m_j}$ with multiplicity matrix $A = ( \alpha_{ji} )$ sends $[P]$ to $[Q = \toplus_{j=1}^{s} Q_j]$ where every entry along $M_p$ in $Q_j \in M_{p}(M_{m_j})$ contains $\alpha_{ji}$ copies, distributed along the diagonal of $M_{m_j}$, of the entry at the same position along $M_p$ in $P_i$. Since the rank of a matrix projection is its trace, one gets $\beta_j \defeq \rank(Q_j) = \sum_{i=1}^{r} \alpha_{ji} \alpha_i$ and the associated module is then $\modN_Q = \bbC^{m_1} \otimes \bbC^{\beta_1} \toplus \cdots \toplus \bbC^{m_s} \otimes \bbC^{\beta_s}$ by the previous construction. So, in terms of modules, $\phi_*$ sends the class of $\bbC^{n_1} \otimes \bbC^{\alpha_1} \toplus \cdots \toplus \bbC^{n_r} \otimes \bbC^{\alpha_r}$ to the class of $\bbC^{m_1} \otimes \bbC^{\beta_1} \toplus \cdots \toplus \bbC^{m_s} \otimes \bbC^{\beta_s}$, where $\bbC^{n_i} \otimes \bbC^{\alpha_i} = M_{\alpha_i, n_i}$ is repeated $\alpha_{ji}$ times on the diagonal of $\bbC^{m_j} \otimes \bbC^{\beta_j} = M_{\beta_j, m_j}$. The diagonals is filled thanks to the relations $m_j = \sum_{i=1}^{r} \alpha_{ji} n_i$ and $\beta_j = \sum_{i=1}^{r} \alpha_{ji} \alpha_i$. See \cite[Ex.~IV.3.1]{Davi96a} where the identification of $\phi_*$ with $A$ is also presented using projections.

This describes the maps $\phi_{n,n+1\, *} : \calD(\algA_n) \to \calD(\algA_{n+1})$ in terms of (finitely generated projective) modules. For $AF$ $C^*$-algebras,  $\calD(\algA)$ equals the space $K_0^+(\algA)$ of stable equivalence classes of projections in $\calP_\infty(\algA)$,\footnote{$P,Q \in \calP_\infty(\algA)$ are stably equivalent, $P \approx Q$, if there is a projection $R \in \calP_\infty(\algA)$ such that $P \oplus R \sim Q \oplus R$.} and this is a cone in $K_0(\algA)$ such that $K_0(\algA) = K_0^+(\algA) - K_0^+(\algA)$ \cite[Thm~IV.1.6, Thm~IV.2.3, Thm~IV.2.4]{Davi96a}. So, for $AF$ $C^*$-algebras, it can be of practical importance to know what it means to approximate elements of $\calD(\algA) = K_0^+(\algA)$. A class $[P] \in \calD(\algA)$ can be looked at as a sequence of classes $[P_n] \in \calD(\algA_n)$ for $n \geq n_0$, related step-by-step by the maps $\phi_{n,n+1\, *}$. The sequence $\{  [P_n] \}_{n \geq n_0}$ corresponds then to a sequence $\{ [\modM_n] \}_{n \geq n_0}$ of equivalence classes of isomorphisms of finitely generated projective modules on the algebras $\algA_n$. Notice that it is only the sequence in the whole that permits to reconstruct the target element $[P] \in \calD(\algA)$. We could say that, for some $n \geq n_0$, the module $\modM_n$ “approximates” (as a representative element in $[\modM_n]$) the class $[P]$, but some information are encoded in the embedding maps $\phi_{n,n+1\, *} : \modM_n \to \modM_{n+1}$ which then participate to this notion of approximation. As seen before, concretely, the maps $\phi_{n,n+1\, *}$ are written in terms of the multiplicity matrices $A_{n,n+1}$ associated to the $\phi_{n,n+1}$.\footnote{The full sequence of multiplicity matrices $A_{n,n+1}$ is provided by the $AF$ $C^*$-algebra. It can be represented graphically by a Bratteli diagram, and it is known that two $AF$ $C^*$-algebras with the same Bratteli diagram are isomorphic \cite[Prop.~III.2.7]{Davi96a}}

It is well-known (Elliott's Theorem, see for instance \cite[Thm~IV.4.3]{Davi96a}) that the $K_0$-group, supplemented with a structure of scaled dimension group, is sufficient to classify $AF$ $C^*$-algebras. So there is no information outside of the one encoded in the sequence of multiplicity matrices $A_{n,n+1}$ to be expected in the constructions described before since it determines a unique $AF$ $C^*$-algebra and it permits to construct its scaled dimension group.\footnote{Keep in mind that an $AF$ $C^*$-algebra can be obtained from different sequences of multiplicity matrices.} 

In our approach to NCGFT based on a “sequence” of finite dimensional NCGFT on the $\algA_n$'s, we will not suppose that an approximation at a level $n_0$  gives us all the information about the “limiting” NCGFT in the $AF$ algebra. In other words, some new inputs (in addition to the $A_{n,n+1}$'s) could be “added” at every step. This implies (obviously) that many non equivalent NCGFT could be constructed on top of a unique $AF$ algebra. This relies on the fact that there may be physical motivations to construct one sequence rather than another and that the chosen embedding at every step could participate to the phenomenology. This is similar to, but also a departure from, GUT where some information are encoded in the SSBM reducing the large group to the group of the SMPP: in our research program, we can look at our embeddings as being in duality with the SSBM, in a way that will be illustrated in Sect.~\ref{sect direct limit NCGFT}.

\bigskip
Let us notice that one key result for the study of $AF$ algebras is \cite[Lemma~III.2.1]{Davi96a}, which describes the possible unital $*$-homomorphisms $\phi : \algA = \toplus_{i=1}^{r} M_{n_i} \to \algB = \toplus_{j=1}^{s} M_{m_j}$. For reasons that will be explained below (see Sect.~\ref{sect direct limit NCGFT}), we will consider non unital $*$-homomorphisms $\phi : \algA = \toplus_{i=1}^{r} M_{n_i} \to \algB = \toplus_{j=1}^{s} M_{m_j}$. In that case, we can use \cite[Lemma~III.2.2]{Davi96a} to describe $\phi$ up to unitary equivalence in $\algB$ with a matrix $A = ( \alpha_{ji} )$, with $\alpha_{ji} \in \bbN$, such that $\sum_{i=1}^{r} \alpha_{ji} n_i \leq m_j$.

Another important point to notice is that in the mathematical considerations described before, the $*$-homomorphisms $\phi : \algA = \toplus_{i=1}^{r} M_{n_i} \to \algB = \toplus_{j=1}^{s} M_{m_j}$ need only be characterized up to unitary equivalence in $\algB$. This is a consequence of the fact that we need only consider “classes” (modulo isomorphisms for instance) for the purpose of classifying the structures. \textit{A priori}, in physics, we may need to consider two $*$-homomorphisms $\phi$ as different even if they are related by a unitary equivalence. This is related to the fact mentioned above that we consider the algebraic structure $\algA_\infty \defeq \cup_{n\geq 0} \algA_n$ instead of its completion, and that its presentation (the sequence of $*$-homomorphisms $\phi_{n,n+1}$) may contain some phenomenological information. But, as will be shown, see Examples~\ref{example Mn} and \ref{example C(M) otimes Mn},  the action of (unitary) inner automorphism is not relevant from a physical point of view since it consists to a transport of structures. These inner automorphisms are similar to gauge transformations in the sense that one can chose a particular representative in the class of equivalent structures to describe a physical situation. This explains why the analysis in this paper relies on a chosen “standard form” for these $*$-homomorphisms which simplifies the presentation.

\subsection{Derivation-based differential calculus}
\label{sec derivation bases differential calculus}

In this paper, we will consider the derivation-based differential calculus, which was defined in \cite{Dubo88a} and studied for various algebras, see for instance \cite{DuboKernMado90a, DuboKernMado90b, Mass96a, DuboMass98a, Mass99a, DuboMich94a, DuboMich96a, DuboMich97a, CagnMassWall11a}. Some reviews can also be found in \cite{Dubo01a, Mass08c, Mass08b}. The main ingredient is the space of derivations on an associative algebra.

Let $\algA$ be an associative algebra with unit $\bbbone$, and let $\calZ(\algA) = \{ a \in \algA \ / \ ab = ba, \forall b \in \algA \}$ be its center. The space of derivations of $\algA$ is
\begin{equation*}
\Der(\algA) = \{ \kX : \algA \rightarrow \algA \ / \ \kX \text{ linear}, \kX\cdotaction(ab) = (\kX\cdotaction a) b + a (\kX\cdotaction b), \forall a,b\in \algA\}.
\end{equation*}
This vector space is a Lie algebra for the bracket $[\kX, \kY ]a = \kX  \kY a - \kY \kX a$ for all $\kX,\kY \in \Der(\algA)$, and a $\calZ(\algA)$-module for the product $(f \kX )\cdotaction a = f ( \kX \cdotaction a)$ for all $f \in \calZ(\algA)$ and $\kX \in \Der(\algA)$. The subspace 
\begin{equation*}
\Int(\algA) = \{ \ad_a : b \mapsto [a,b]\ / \ a \in \algA\} \subset \Der(\algA)
\end{equation*}
is called the vector space of inner derivations: it is a Lie ideal and a $\calZ(\algA)$-submodule. 
 
Suppose that $\algA$ has an involution $a \mapsto a^*$. Then a real derivation on $\algA$ is a derivation $\kX$ such that $(\kX \cdotaction a)^* = \kX \cdotaction a^*$ for any $a \in \algA$.

\medskip
Let $\OmegaDer^p(\algA)$ be the vector space of $\calZ(\algA)$-multilinear antisymmetric maps from $\Der(\algA)^p$ to $\algA$, with $\OmegaDer^0(\algA) = \algA$. Then the total space
\begin{equation*}
\OmegaDer^\grast(\algA) = \toplus_{p \geq 0} \OmegaDer^p(\algA)
\end{equation*}
gets a structure of $\bbN$-graded differential algebra for the product
\begin{multline}
\label{eq def form product}
(\omega \wedge \eta)(\kX_1, \dots, \kX_{p+q}) 
\defeq
\\
 \frac{1}{p!q!} \sum_{\sigma\in \kS_{p+q}} (-1)^{\abs{\sigma}} \omega(\kX_{\sigma(1)}, \dots, \kX_{\sigma(p)}) \eta(\kX_{\sigma(p+1)}, \dots, \kX_{\sigma(p+q)})
\end{multline}
for any $\omega \in \OmegaDer^p(\algA)$, any $\eta \in \OmegaDer^q(\algA)$ and any $\kX_i \in \Der(\algA)$ where $\kS_{n}$ is the group of permutations of $n$ elements. A differential $\dd$ is defined by the so-called Koszul formula
\begin{multline}
\label{eq def differential}
\dd \omega (\kX_1, \dots , \kX_{p+1}) 
\defeq
\sum_{i=1}^{p+1} (-1)^{i+1} \kX_i \cdotaction \omega( \kX_1, \dots \omi{i} \dots, \kX_{p+1}) 
\\
 + \sum_{1 \leq i < j \leq p+1} (-1)^{i+j} \omega( [\kX_i, \kX_j], \dots \omi{i} \dots \omi{j} \dots, \kX_{p+1}). 
\end{multline}
This makes $(\OmegaDer^\grast(\algA), \dd)$ a graded differential algebra.

\begin{proposition}[Transport of forms by automorphisms]
\label{prop transport forms}
Let $\Psi : \algA \to \algA$ be an algebra automorphism. Then $\Psi$ induces an automorphism on $\calZ(\algA)$.

The map $\PsiDer : \Der(\algA) \to \Der(\algA)$ defined by $\PsiDer(\kX) \cdotaction a \defeq \Psi( \kX \cdotaction \Psi^{-1}(a) )$ for any $\kX \in \Der(\algA)$ and $a \in \algA$, is an automorphism of the Lie algebra $\Der(\algA)$ and $\PsiDer(f \kX) = \Psi(f) \PsiDer(\kX)$ for any $f \in \calZ(\algA)$ (so $\PsiDer$ is not necessary an automorphism for the structure of $\calZ(\algA)$-module). For inner derivations, one has $\PsiDer(\ad_a) = \ad_{\Psi(a)}$.

The maps $\Psi : \OmegaDer^p(\algA) \to \OmegaDer^p(\algA)$ defined by 
\begin{align*}
\Psi(\omega) (\kX_{1}, \dots , \kX_{p}) 
\defeq \Psi\left( \omega( \PsiDer^{-1}(\kX_{1}), \dots, \PsiDer^{-1}(\kX_{p}) ) \right)
\end{align*}
for any $\omega \in \OmegaDer^p(\algA)$ and $\kX_i \in \Der(\algA)$ define an automorphism of the graded differential algebra $(\OmegaDer^\grast(\algA), \dd)$.
\end{proposition}

For $p=0$, $\Psi$ defined on $\OmegaDer^0(\algA) = \algA$ is exactly the original automorphism $\Psi$ of $\algA$, so that the notation is justified.

\begin{proof}
For any $f \in \calZ(\algA)$ and $a \in \algA$, one has $\Psi(f) a = \Psi( f \Psi^{-1}(a) ) = \Psi( \Psi^{-1}(a) f ) = a \Psi(f)$ so that $\Psi(f) \in \calZ(\algA)$.

With obvious notations, one has 
\begin{align*}
\PsiDer(\kX) \cdotaction (a b) 
&= \Psi\left( \kX \cdotaction \Psi^{-1}(a b) \right) 
= \Psi\left( \kX \cdotaction (\Psi^{-1}(a) \Psi^{-1}(b)) \right) 
\\
&= \Psi\left( \kX \cdotaction \Psi^{-1}(a) \right) b + a \Psi\left( \kX \cdotaction \Psi^{-1}(b)\right) 
= (\PsiDer(\kX) \cdotaction a) b + a (\PsiDer(\kX) \cdotaction b)
\end{align*}
so that $\PsiDer(\kX)$ is a derivation. 

In the same way, one has $\PsiDer(f \kX) \cdotaction a = \Psi( f (\kX \cdotaction \Psi^{-1}(a)) ) = \Psi(f) \Psi( \kX \cdotaction \Psi^{-1}(a) ) = \Psi(f) \PsiDer(\kX) \cdotaction a$. 

For $\kX, \kY \in \Der(\algA)$, one has 
\begin{align*}
\PsiDer([\kX, \kY]) \cdotaction a 
&= \Psi( \kX \cdotaction (\kY \cdotaction \Psi^{-1}(a)) ) - \Psi( \kY \cdotaction (\kX \cdotaction \Psi^{-1}(a)) ) 
\\
&= \Psi( \kX \cdotaction \Psi^{-1} ( \PsiDer(\kY) \cdotaction a) ) - \Psi( \kY \cdotaction \Psi^{-1} ( \PsiDer(\kX) \cdotaction a) )
\\
&= \PsiDer(\kX) \cdotaction ( \PsiDer(\kY) \cdotaction a ) - \PsiDer(\kY) \cdotaction ( \PsiDer(\kX) \cdotaction a )
\\
&= [\PsiDer(\kX), \PsiDer(\kY)] \cdotaction a
\end{align*}
so that $\PsiDer([\kX, \kY]) = [\PsiDer(\kX), \PsiDer(\kY)]$. The inverse $\PsiDer^{-1}$ is defined by $\PsiDer^{-1}(\kX) \cdotaction a \defeq \Psi^{-1}( \kX \cdotaction \Psi(a) )$ as can be easily checked. For inner derivations, one has $\PsiDer(\ad_a) \cdotaction b = \Psi( [a, \Psi^{-1}(b)]) = [\Psi(a), b] = \ad_{\Psi(a)} \cdotaction b$.

For any $\omega \in \OmegaDer^p(\algA)$, it is easy to check that $\Psi(\omega)$ is a $\calZ(\algA)$-multilinear antisymmetric maps from $\Der(\algA)^p$ to $\algA$. For any $\omega \in \OmegaDer^p(\algA)$ and any $\eta \in \OmegaDer^q(\algA)$, the relation $\Psi(\omega) \wedge \Psi(\eta) = \Psi( \omega \wedge \eta )$ is a direct consequence of the definition of $\Psi$ on forms. The proof of $\Psi(\dd a) = \dd \Psi(a)$ is a straightforward computation: $\Psi(\dd a)(\kX) = \Psi( \dd a (\PsiDer^{-1}(\kX))) = \Psi( \PsiDer^{-1}(\kX) \cdotaction a) = \kX \cdotaction \Psi(a)$ on the one hand and $(\dd \Psi(a))(\kX) = \kX \cdotaction \Psi(a)$ on the other hand. To prove $\dd \Psi(\omega) = \Psi(\dd \omega)$ for $\omega \in \OmegaDer^p(\algA)$, one has to use a similar computation and the fact that $\PsiDer^{-1}$ is a morphism of Lie algebras.
\end{proof}

\begin{example}[Transport of derivations by inner automorphisms]
\label{ex transport derivations inner automorphisms}
Let $u \in \algA$ be an invertible element (one can take $u$ to be unitary when $\algA$ has an involution). The map $\Psi(a) \defeq u a u^{-1}$ defines an automorphism of $\algA$ and a simple computation shows that $\PsiDer(\kX) = \kX + \ad_{u (\kX \cdotaction u^{-1})}$ for any $\kX \in \Der(\algA)$. In particular, if $\kX = \ad_a$ is an inner derivation, then $\PsiDer(\ad_a) = \ad_{u a u^{-1}} = \ad_{\Psi(a)}$ as expected. Notice also that $\Psi(f) = f$ for any $f \in \calZ(\algA)$ so that $\PsiDer : \Der(\algA) \to \Der(\algA)$ is an automorphism of $\calZ(\algA)$-module in that case.
\end{example}

\subsection{Results on matrix algebras}
\label{sec matrix algebras}

Since the situation $\algA_i = M_{n_i}(\bbC)$ is our main objective for $AF$-algebras, we give here a series of notations and results that will be used below when this specific situation will be considered. This is for instance the case in Sect.~\ref{sec metric hodge}. We refer to \cite{DuboKernMado90b, Mass95a, DuboMass98a, Mass08b, Mass12a} for more details.

\medskip
The center of the algebra $M_n \defeq M_n(\bbC)$ is $\calZ(M_n) = \bbC \bbbone_n$ where $\bbbone_n$ is the unit matrix in $M_n$. Let $\ksl_n$ be the Lie algebra (for the commutator) of traceless matrices in $M_n$. Then the map $\ksl_n \ni a \mapsto \ad_a \in \Int(M_n)$ realizes an isomorphism $\ksl_n \simeq \Der(M_n) = \Int(M_n)$.

Let $\{ E_\alpha \}_{\alpha \in I}$ be a basis of $\ksl_n$, where $I$ is a totally ordered set with $\card(I) = n^2 -1 = \dim \ksl_n$. Choosing an abstract totally ordered set $I$ to label this basis will be convenient when the inductive sequence defining the $AF$-algebra will be considered since then the $\alpha$'s will be constructed as cumulative multi-indexes. The ordering will be used to order basis forms (for instance to define volume forms). Let us introduce the unique multiplet $(\alpha^0_1, \dots, \alpha^0_{n^2-1}) \in I^{n^2-1}$ such that $\alpha^0_1 < \cdots < \alpha^0_{n^2-1}$.  We will use the notation $C(n)_{\alpha\beta}^\gamma = C_{\alpha\beta}^\gamma$ for the structure constants of the Lie algebra $\ksl_n$ in the basis $\{ E_\alpha \}_{\alpha \in I}$: $[E_\alpha, E_\beta] = C_{\alpha\beta}^\gamma E_\gamma$.

 The basis $\{ E_\alpha \}_{\alpha \in I}$ induces a basis $\{ \partial_\alpha \defeq \ad_{E_\alpha} \}_{\alpha \in I}$ of $\Der(M_n) = \Int(M_n)$.  Let $\{ \theta^\alpha \}_{\alpha \in I}$ be its dual basis in $\ksl_n^*$. The derivation $\partial_\alpha$ is real if and only if $E_\alpha$ is anti-Hermitean and one has $[\partial_\alpha, \partial_\beta] = C_{\alpha\beta}^\gamma \partial_\gamma$.

\medskip
The space of noncommutative forms on $M_n$ has a simple structure:
\begin{align*}
\OmegaDer^\grast(M_n)
&= M_n \otimes \exter^\grast \ksl_n^*
\end{align*}
and the differential is the Chevalley-Eilenberg differential for the differential graded algebra associated to the Lie algebra $\ksl_n$ with values in $M_n$ using the adjoint representation. Identifying $\theta^\gamma$ with $\bbbone_n \otimes \theta^\gamma \in \OmegaDer^1(M_n) = M_n \otimes \exter^1 \ksl_n^*$, one has $\dd \theta^\gamma = - \tfrac{1}{2} C_{\alpha\beta}^\gamma \theta^\alpha \wedge \theta^\beta$.

Let us consider the canonical metric $g : \Der(M_n) \times \Der(M_n) \to \calZ(M_n) \simeq \bbC$ defined by $g(\ad_a, \ad_b) \defeq \tr(a b)$ for $a, b \in \ksl_n$. This is not the metric defined in \cite{DuboKernMado90b} where a factor $\tfrac{1}{n}$ was put in front of the trace (to get the \emph{normalized} trace). The reason for this convention will be explained below (see \eqref{eq gB and gA} and comments after).  Once the basis $\{ \partial_\alpha \}_{\alpha \in I}$ is given, one introduces the components $g_{\alpha \beta} \defeq g(\partial_\alpha, \partial_\beta) = \tr(E_\alpha E_\beta)$ of $g$.

\medskip
Let $\abs{g}$ be the determinant of the matrix $( g_{\alpha \beta} )_{\alpha,\beta}$. We define the (noncommutative) integral $\int_{M_n}$ on $\OmegaDer^\grast(M_n)$ by the following rule. For any $\omega \in \OmegaDer^p(M_n)$ with $p<n^2-1$, $\int_{M_n} \omega = 0$. Any $\omega \in \OmegaDer^{n^2-1}(M_n)$ can be written as $\omega = a \sqrt{\abs{g}} \theta^{\alpha^0_1} \wedge \cdots \wedge \theta^{\alpha^0_{n^2-1}}$ for a unique $a \in M_n$ which is independent of the chosen basis $\{ E_\alpha \}_{\alpha \in I}$ and we define 
\begin{align*}
\int_{M_n} \omega 
&= \int_{M_n} a \sqrt{\abs{g}} \theta^{\alpha^0_1} \wedge \cdots \wedge \theta^{\alpha^0_{n^2-1}} 
\defeq \tr(a)
\end{align*}
Once again, this is not the convention used in \cite{DuboKernMado90b} where a factor $\tfrac{1}{n}$ was put in front of the RHS. In our convention, $\omega_{\vol} \defeq \sqrt{\abs{g}} \theta^{\alpha^0_1} \wedge \cdots \wedge \theta^{\alpha^0_{n^2-1}}$ is the volume form whose integral is normalized to $n$.

\medskip
The metric permits to define the Hodge $\hstar$-operator 
\begin{align*}
\hstar : \OmegaDer^p(M_n) \to \OmegaDer^{n^2-1-p}(M_n)
\end{align*}
defined by
\begin{align}
\label{eq def hodge star}
\hstar (\theta^{\alpha_1} \wedge \cdots \wedge \theta^{\alpha_p})
&\defeq
\tfrac{1}{(n^2-1-p) !} \sqrt{\abs{g}} g^{\alpha_1\beta_1} \cdots g^{\alpha_p\beta_p} \epsilon_{\beta_1, \dots, \beta_{n^2-1}} \theta^{\beta_{p+1}} \wedge \cdots \wedge \theta^{\beta_{n^2-1}}
\end{align}
where $\epsilon_{\beta_1, \dots, \beta_{n^2-1}}$ is the completely antisymmetric tensor such that $\epsilon_{\alpha^0_1, \dots, \alpha^0_{n^2-1}} = 1$.

\begin{example}[Transport by inner automorphisms]
\label{ex transport matrix by inner automorphisms}
Let us consider the situation described in Example~\ref{ex transport derivations inner automorphisms} in the context of the matrix algebra. Let $\omega = \frac{1}{p!} \omega_{\alpha_1, \dots, \alpha_p} \theta^{\alpha_1} \wedge \cdots \wedge \theta^{\alpha_p}$ be a $p$-form. To compute $\omega^u \defeq \Psi(\omega)$, let us introduce the matrix $U = (U_\alpha^{\beta})$ defined by $u^{-1} E_\alpha u = U_\alpha^{\beta} E_{\beta}$, so that $\PsiDer^{-1}(\partial_\alpha) = \ad_{u^{-1} E_\alpha u} = U_\alpha^{\beta} \partial_{\beta}$. Since $GL_n$ is unimodular, one has $\det(U) = 1$. Notice also that $u^{-1} [E_{\alpha_1}, E_{\alpha_2}] u = [u^{-1} E_{\alpha_1} u, u^{-1} E_{\alpha_2} u] = U_{\alpha_1}^{\beta_1} U_{\alpha_2}^{\beta_2} [E_{\beta_1}, E_{\beta_2}] = U_{\alpha_1}^{\beta_1} U_{\alpha_2}^{\beta_2} C_{\beta_1 \beta_2}^{\beta_3} E_{\beta_3}$ on the one hand and $u^{-1} [E_{\alpha_1}, E_{\alpha_2}] u = C_{\alpha_1 \alpha_2}^{\alpha_3} u^{-1} E_{\alpha_3} u = C_{\alpha_1 \alpha_2}^{\alpha_3} U_{\alpha_3}^{\beta_3} E_{\beta_3}$ on the other hand, so that $U_{\alpha_1}^{\beta_1} U_{\alpha_2}^{\beta_2} C_{\beta_1 \beta_2}^{\beta_3} = C_{\alpha_1 \alpha_2}^{\alpha_3} U_{\alpha_3}^{\beta_3}$. By definition, $(\theta^\alpha)^u(\partial_{\alpha'}) = u \theta^\alpha(\PsiDer^{-1}(\partial_{\alpha'})) u^{-1} = u U_{\alpha'}^{\beta'} \delta_{\beta'}^{\alpha} u^{-1} = U_{\alpha'}^{\alpha} = U_{\beta}^{\alpha} \theta^{\beta}(\partial_{\alpha'})$ so that $(\theta^\alpha)^u = U_{\beta}^{\alpha} \theta^{\beta}$. In the same way, $\omega^u_{\alpha_1, \dots, \alpha_p} = \omega^u(\partial_{\alpha_1}, \dots, \partial_{\alpha_p}) = u \omega(\PsiDer^{-1}(\partial_{\alpha_1}), \dots, \PsiDer^{-1}(\partial_{\alpha_p})) u^{-1} = U_{\alpha_1}^{\beta_1} \cdots U_{\alpha_p}^{\beta_p} u \omega_{\beta_1, \dots, \beta_p} u^{-1}$, so that $\omega^u = \frac{1}{p!} U_{\alpha_1}^{\beta_1} \cdots U_{\alpha_p}^{\beta_p} u \omega_{\beta_1, \dots, \beta_p} u^{-1} \theta^{\alpha_1} \wedge \cdots \wedge \theta^{\alpha_p} = \frac{1}{p!} u \omega_{\alpha_1, \dots, \alpha_p} u^{-1} (\theta^{\alpha_1})^u \wedge \cdots \wedge (\theta^{\alpha_p})^u$. Now, the metric $g(\ad_a, \ad_b) = \tr(a b)$, for $a, b \in \ksl_n$, is invariant by the transport associated to the inner automorphism $\Psi(a) = u a u^{-1}$, and so one has $g_{\alpha \beta} = U_{\alpha}^{\alpha'} U_{\beta}^{\beta'} g_{\alpha' \beta'}$ and $g^{\alpha \beta} = U_{\alpha'}^{\alpha} U_{\beta'}^{\beta} g^{\alpha' \beta'}$ for the inverse metric. In particular, all the conditions and properties concerning orthonormality associated to $g$ are transported by $\Psi$. For $\omega = a \sqrt{\abs{g}} \theta^{\alpha^0_1} \wedge \cdots \wedge \theta^{\alpha^0_{n^2-1}}$, one has $\omega^u = u a u^{-1} \sqrt{\abs{g}} (\theta^{\alpha^0_1})^u \wedge \cdots \wedge (\theta^{\alpha^0_{n^2-1}})^u = u a u^{-1} \sqrt{\abs{g}} U_{\beta_1}^{\alpha^0_1} \cdots U_{\beta_{n^2-1}}^{\alpha^0_{n^2-1}} \theta^{\beta_1} \wedge \cdots \wedge \theta^{\beta_{n^2-1}} = u a u^{-1} \sqrt{\abs{g}} \theta^{\alpha^0_1} \wedge \cdots \wedge \theta^{\alpha^0_{n^2-1}}$ where we have used $\theta^{\beta_1} \wedge \cdots \wedge \theta^{\beta_{n^2-1}} = \epsilon^{\beta_1, \dots, \beta_{n^2-1}} \theta^{\alpha^0_1} \wedge \cdots \wedge \theta^{\alpha^0_{n^2-1}}$ and $\epsilon^{\beta_1, \dots, \beta_{n^2-1}} U_{\beta_1}^{\alpha^0_1} \cdots U_{\beta_{n^2-1}}^{\alpha^0_{n^2-1}} = \det(U) = 1$. This implies that $\int_{M_n} \omega^u = \tr(u a u^{-1}) = \tr(a) = \int_{M_n} \omega$. Since the metric $g$ is invariant, the Hodge $\hstar$-operator is also invariant according to \eqref{eq def hodge star}, and since the inverse metric is also invariant under the action of $U$, a straightforward computation of $\hstar ( (\theta^{\alpha_1})^u \wedge \cdots \wedge (\theta^{\alpha_p})^u)$ shows that the relation \eqref{eq def hodge star} is also valid when one replaces all the $\theta^\alpha$ by $(\theta^\alpha)^u$ on both sides. Combining all these results and the explicit relation \eqref{eq hodge star omega omega'}, one can show that for any $p$-forms $\omega$ and $\omega'$, one has $\int_{M_n} (\omega \wedge \hstar \omega')^u = \int_{M_n} \omega^u \wedge \hstar \omega'^u = \int_{M_n} \omega \wedge \hstar \omega'$.
\end{example}

\subsection{Noncommutative Gauge Field Theories}
\label{sec NCGFT}

Gauge field theories of Yang-Mills type can be described in terms of fiber bundles and connections. Noncommutative geometry, as a natural extension of ordinary geometry, has been used to develop gauge field theories (hereafter mentioned as NCGFT) in which scalar fields are part of the generalized notion of connections. Then, the naturally constructed Lagrangians produce quadratic potentials for these fields, providing a SSBM in these models. See \cite{DuboKernMado90a, DuboKernMado90b, ConnLott90a} for the initial attempts and \cite{ChamConnMarc07a} for a more elaborated reconstruction of the Standard Model of particles physics. See also \cite{Suij15a} for a review and references. 

The necessary building blocks to constructed noncommutative gauge fields theories are motivated and described in \cite{FranLazzMass14a}. In order to fix notations, we summarize here the main ingredients in the case of the derivation-based differential calculus.

Let $\algA$ be a unital associative algebra equipped with an involution $a \mapsto a^*$ and let $\modM$ be a left $\algA$-module. A (noncommutative) connection $\nabla$ on $\modM$ is a family of linear maps $\nabla_\kX : \modM \to \modM$ defined for any $\kX \in \Der(\algA)$ such that
\begin{enumerate}
\item $\nabla_{f \kX} = f \nabla_{\kX}$ and $\nabla_{\kX + \kY} = \nabla_{\kX} + \nabla_{\kY}$ for any $f \in \calZ(\algA)$ and $\kX, \kY \in \Der(\algA)$.

\item $\nabla_{\kX} (a e) = (\kX \cdotaction a) e + a \nabla_{\kX} e$ for any $a \in \algA$, $e \in \modM$ and $\kX \in \Der(\algA)$.
\end{enumerate}

The curvature of $\nabla$ is the family of maps $R(\kX, \kY) : \modM \to \modM$ defined for any $e \in \modM$ and $\kX, \kY \in \Der(\algA)$ by
\begin{align*}
R(\kX, \kY) e \defeq ( \nabla_{\kX} \nabla_{\kY} - \nabla_{\kY} \nabla_{\kX} - \nabla_{[\kX, \kY]}) e
\end{align*}
It can be easily shown that $R(\kX, \kY) (a e) = a R(\kX, \kY) e$ for any $a \in \algA$ so that $R(\kX, \kY) \in \Hom_{\algA}(\modM, \modM)$ (space of homomorphisms of left modules).

\medskip
A Hermitian structure on $\modM$ is a $\bbR$-linear map $h : \modM \otimes \modM \to \algA$ such that $h(a_1 e_1, a_2 e_2) = a_1 h(e_1, e_2) a_2^*$ for any $a_1, a_2 \in \algA$ and $e_1, e_2 \in \modM$. A connection $\nabla$ is Hermitian if for any real derivation $\kX$ of $\algA$ and any $e_1, e_2 \in \modM$, one has
\begin{align*}
\kX \cdotaction h(e_1, e_2)
&= h(\nabla_\kX e_1, e_2) + h(e_1, \nabla_\kX e_2)
\end{align*}

\medskip
We suppose that $\modM$ is equipped with a Hermitian structure $h$.

\medskip
The gauge group $\calG$ of $\modM$ is the group of automorphisms of $\modM$ as a left module which preserve the Hermitian structure: so, for any $a \in \algA$, $e, e' \in \modM$, $\phi \in \calG$ satisfies 
\begin{align*}
\phi(a e) &= a \phi(e),
\quad
h(\phi(e_1),\phi(e_2)) = h(e_1, e_2).
\end{align*}
The action of $\phi$ on a connection is defined by the compositions $\nabla_{\kX} \mapsto \nabla^{\phi}_{\kX} \defeq \phi \circ \nabla_{\kX} \circ \phi^{-1}$. It is easy to check that $\nabla^{\phi}$ is a connection and that $\nabla^{\phi_2 \circ \phi_1} = (\nabla^{\phi_1})^{\phi_2}$ for any $\phi_1, \phi_2 \in \calG$.

\medskip
A special case of interest is the left module $\modM = \algA$ for the multiplication in $\algA$ equipped with the canonical Hermitian structure $h(a,b) \defeq a b^*$ for any $a,b \in \modM = \algA$. Then, since $\algA$ is unital, the connection $\nabla$ is completely given by its values on the unit $\bbbone \in \algA$: for any $e \in \modM = \algA$ and any $\kX \in \Der(\algA)$,
\begin{align*}
\nabla_{\kX} e 
= \nabla_{\kX} (e \bbbone)
= (\kX \cdotaction e) + e (\nabla_{\kX} \bbbone)
= (\kX \cdotaction e) + e \omega(\kX)
\end{align*}
where we define $\omega(\kX) \defeq \nabla_{\kX} \bbbone$. Then one has $\omega \in \OmegaDer^1(\algA)$ and $\omega$ is called the connection $1$-form of $\nabla$. The compatibility of $\nabla$ with $h$ implies that for any real derivation $\kX$, one has $\omega(\kX) + \omega(\kX)^* = 0$ since $0 = \kX \cdotaction \bbbone = \kX \cdotaction h(\bbbone, \bbbone) = \omega(\kX) + \omega(\kX)^*$.

The curvature can be computed in terms of $\omega$ as $R(\kX, \kY) e = e \Omega(\kX, \kY)$ where $\Omega(\kX, \kY) \defeq (\dd \omega)(\kX, \kY) - [\omega(\kX), \omega(\kY)]$ (here we use the fact that $\modM = \algA$ is also a right $\algA$-module). The $2$-form $\Omega \in \OmegaDer^2(\algA)$ is the curvature $2$-form of $\nabla$.

The gauge group $\calG$ is the space $\calU(\algA)$ of unitary elements in $\algA$ which act on the right on $\modM$. Indeed, any $\phi \in \calG$ is defined by its value $u \defeq \phi(\bbbone) \in \algA$ so that $\phi(e) = \phi(e \bbbone) = e \phi(\bbbone) = e u$. Since $\calG$ is a group, the element $u$ is invertible in $\algA$ and the unitary condition comes from the compatibility with the Hermitian structure: $\bbbone = h(\bbbone, \bbbone) = h(\bbbone u, \bbbone u) = u u^\ast$. The connection $1$-form associated to $\nabla^\phi$ is then $\omega^u \defeq u^{-1} \omega u - u^{-1} (\dd u)$ and its curvature $2$-form is $\Omega^u \defeq u^{-1} \Omega u$.

\medskip
To define a gauge field theory on $\modM$, one considers the “fields” defining a connection $\nabla$ on $\modM$ and a Lagrangian $\calL(\nabla)$ for these fields. This Lagrangian is usually constructed for the left module $\modM = \algA$ using a Hodge star operator $\hstar$ on the space of forms on $\algA$ as $\calL(\nabla) \defeq - \Omega \wedge \hstar \Omega$. Then, using a trace $\int_{\algA}$ (which sends forms to scalars) we can define an action $\calS[\nabla] \defeq \int_{\algA} \calL(\nabla) =  - \int_{\algA} \Omega \wedge \hstar \Omega$ (the sign is necessary for positivity). The matter Lagrangian can be defined in a similar way. One first consider $\nabla$ as a map $\nabla : \modM \to \OmegaDer^1(\algA) \otimes_{\algA} \modM$. Using a natural involution on $\OmegaDer^\grast(\algA)$ which extends the involution on $\algA$ (see \cite{FranLazzMass14a} for instance), one can extend $h$ to $(\OmegaDer^p(\algA) \otimes_{\algA} \modM )  \otimes (\OmegaDer^q(\algA) \otimes_{\algA} \modM ) \to \OmegaDer^{p+q}(\algA)$ by $h(\omega_p \otimes e, \omega_q \otimes e') \defeq \omega_p  h(e, e') \wedge \omega_q^*$. Then $\int_{\algA} h( \nabla e, \hstar \nabla e)$ defines a Klein-Gordon type action for matter fields $e \in \modM$.

Since we restrict our analysis to matrix algebras, we refer to Sect.~\ref{sec matrix algebras} for the construction of an explicit Hodge star operator and a trace.

\begin{proposition}[Transport of connections by automorphisms]
\label{prop transport of connections by automorphisms}
Let us consider the hypothesis of Prop.~\ref{prop transport forms}. 

Let $\PsiMod : \modM \to \modM$ be an invertible linear map such that $\PsiMod(a e) = \Psi(a) \PsiMod(e)$ for any $a \in \algA$ and $e \in \modM$.

Let $\nabla$ be a connection on $\modM$ compatible with a Hermitian structure $h$ on $\modM$. Then, for any $\kX \in \Der(\algA)$ and $e \in \modM$, the maps $\nabla^{\PsiMod}_\kX e \defeq \PsiMod \left( \nabla_{\PsiDer^{-1}(\kX)} \PsiMod^{-1}(e) \right)$ define a connection on $\modM$ which is compatible with the Hermitian structure $h^{\PsiMod}$ defined by $h^{\PsiMod}(e_1, e_2) \defeq \Psi \left( h( \PsiMod^{-1}(e_1), \PsiMod^{-1}(e_2)) \right)$. Its curvature $R^{\PsiMod}$ satisfies $R^{\PsiMod}(\kX, \kY) e = \PsiMod \left( R(\PsiDer^{-1}(\kX), \PsiDer^{-1}(\kY) ) \PsiMod^{-1}(e) \right)$ where $R$ is the curvature of $\nabla$.

Let $\phi \in \calG$ be a gauge transformation on $\modM$. Then $\phi^{\PsiMod} (e) \defeq \PsiMod \circ \phi \circ \PsiMod^{-1}(e)$ belongs to $\calG$. If $\phi$ is compatible with $h$ then $\phi^{\PsiMod}$ is compatible with $h^{\PsiMod}$. One has $(\nabla^\phi)^{\PsiMod} = (\nabla^{\PsiMod})^{\phi^{\PsiMod}}$.

For $\modM = \algA$, let $\PsiMod = \Psi$. Let $\omega$ (resp. $\omega^\Psi$) be the connection $1$-form of $\nabla$ (resp. of $\nabla^{\Psi}$). Then one has $\omega^\Psi = \Psi(\omega)$. Let $u = \phi(\bbbone)$ and $u^\Psi = \phi^\Psi(\bbbone)$, then $u^\Psi = \Psi(u)$.
\end{proposition}

\begin{proof}
For any $a \in \algA$, $f \in \calZ(\algA)$, $\kX \in \Der(\algA)$ and $e \in \modM$, one has 
\begin{align*}
\nabla^{\PsiMod}_\kX (a e)
&= \PsiMod \left( \nabla_{\PsiDer^{-1}(\kX)} \PsiMod^{-1}(a e) \right)
= \PsiMod \left( \nabla_{\PsiDer^{-1}(\kX)} \Psi^{-1}(a) \PsiMod^{-1}(e) \right)
\\
& = \begin{multlined}[t]
\PsiMod \left( (\PsiDer^{-1}(\kX) \cdotaction \Psi^{-1}(a) ) \PsiMod^{-1}(e) \right) 
\\
	+ \PsiMod \left( \Psi^{-1}(a) \nabla_{\PsiDer^{-1}(\kX)} \PsiMod^{-1}(e) \right)
	\end{multlined}
\\
&= \PsiMod \left( \Psi^{-1}(\kX \cdotaction a) \PsiMod^{-1}(e) \right)  + a \PsiMod \left( \nabla_{\PsiDer^{-1}(\kX)} \PsiMod^{-1}(e) \right)
\\
&= (\kX \cdotaction a) e + a \nabla^{\PsiMod}_\kX e
\end{align*}
and
\begin{align*}
\nabla^{\PsiMod}_{f \kX} e
&= \PsiMod \left( \nabla_{\PsiDer^{-1}(f \kX)} \PsiMod^{-1}(e) \right)
= \PsiMod \left( \nabla_{\Psi^{-1}(f )\PsiDer^{-1}(\kX)} \PsiMod^{-1}(e) \right)
\\
&= \PsiMod \left( \Psi^{-1}(f ) \nabla_{\PsiDer^{-1}(\kX)} \PsiMod^{-1}(e) \right)
= f \PsiMod \left( \nabla_{\PsiDer^{-1}(\kX)} \PsiMod^{-1}(e) \right)
\\
&= f \nabla^{\PsiMod}_{\kX} e
\end{align*}
so that $\nabla^{\PsiMod}$ is a connection. The compatibility with $h^{\PsiMod}$ is proved by
\begin{align*}
\kX \cdotaction h^{\PsiMod}(e_1, e_2)
&= \kX \cdotaction \Psi \left( h( \PsiMod^{-1}(e_1), \PsiMod^{-1}(e_2)) \right)
\\
&= \Psi \left( \PsiDer^{-1}(\kX) \cdotaction h( \PsiMod^{-1}(e_1), \PsiMod^{-1}(e_2)) \right)
\\
&= \begin{multlined}[t]
\Psi \left( h( \nabla_{\PsiDer^{-1}(\kX)} \PsiMod^{-1}(e_1), \PsiMod^{-1}(e_2)) \right)
\\
	+ \Psi \left( h( \PsiMod^{-1}(e_1), \nabla_{\PsiDer^{-1}(\kX)} \PsiMod^{-1}(e_2)) \right)
	\end{multlined}
\\
&= \begin{multlined}[t]
\Psi \left( h( \PsiMod^{-1} (\nabla^{\PsiMod}_{\kX} e_1), \PsiMod^{-1}(e_2)) \right)
\\
	+ \Psi \left( h( \PsiMod^{-1}(e_1), \PsiMod^{-1} (\nabla^{\PsiMod}_{\kX} e_2)) \right)
	\end{multlined}
\\
&= h^{\PsiMod}( \nabla^{\PsiMod}_{\kX} e_1, e_2)
	+ h^{\PsiMod}(e_1, \nabla^{\PsiMod}_{\kX} e_2) .
\end{align*}
The relation for the curvature $R^{\PsiMod}$ is a straightforward computation:
\begin{align*}
R^{\PsiMod}(\kX, \kY) e
&= \left(  \nabla^{\PsiMod}_{\kX}  \nabla^{\PsiMod}_{\kY} -  \nabla^{\PsiMod}_{\kY}  \nabla^{\PsiMod}_{\kX} -  \nabla^{\PsiMod}_{[\kX, \kY]} \right) e
\\
&= \begin{multlined}[t]
\PsiMod \circ \nabla_{\PsiDer^{-1}(\kX)} \circ \PsiMod^{-1} \circ \PsiMod \circ \nabla_{\PsiDer^{-1}(\kY)} \circ \PsiMod^{-1}(e)
\\
- \PsiMod \circ \nabla_{\PsiDer^{-1}(\kY)} \circ \PsiMod^{-1} \circ \PsiMod \circ \nabla_{\PsiDer^{-1}(\kX)} \circ \PsiMod^{-1}(e)
\\
- \PsiMod \circ \nabla_{\PsiDer^{-1}([\kX, \kY])} \circ \PsiMod^{-1}(e)
\end{multlined}
\\
&= \begin{multlined}[t]
\PsiMod \circ \nabla_{\PsiDer^{-1}(\kX)} \circ \nabla_{\PsiDer^{-1}(\kY )} \circ \PsiMod^{-1}(e)
\\
- \PsiMod \circ \nabla_{\PsiDer^{-1}(\kY)} \circ \nabla_{\PsiDer^{-1}(\kX)} \circ \PsiMod^{-1}(e)
\\
- \PsiMod \circ \nabla_{[\PsiDer^{-1}(\kX), \PsiDer^{-1}(\kY)]} \circ \PsiMod^{-1}(e)
\end{multlined}
\\
&= \PsiMod \left( R(\PsiDer^{-1}(\kX), \PsiDer^{-1}(\kY)) \PsiMod^{-1}(e) \right)
\end{align*} 

The map $\phi^{\PsiMod}$ is obviously invertible with inverse $(\phi^{\PsiMod})^{-1} = \PsiMod \circ \phi^{-1} \circ \PsiMod^{-1}$. It is a morphism of modules: $\phi^{\PsiMod}(a e) = \PsiMod \circ \phi \circ \PsiMod^{-1} (a e) = \PsiMod \circ \phi \left( \Psi^{-1}(a) \PsiMod^{-1}(e) \right) = \PsiMod \left( \Psi^{-1}(a) \phi \circ \PsiMod^{-1}(e) \right) = a \phi^{\PsiMod}(e)$. One has
\begin{align*}
h^{\PsiMod}&( \phi^{\PsiMod}(e_1), \phi^{\PsiMod}(e_2) )
\\
&= \Psi \left( h( \PsiMod^{-1} \circ \PsiMod \circ \phi \circ \PsiMod^{-1}(e_1), \PsiMod^{-1} \circ \PsiMod \circ \phi \circ \PsiMod^{-1}(e_2) \right)
\\
&= \Psi \left( h( \phi \circ \PsiMod^{-1}(e_1), \phi \circ \PsiMod^{-1}(e_2) \right)
\\
&= \Psi \left( h( \PsiMod^{-1}(e_1), \PsiMod^{-1}(e_2) \right)
\\
&= h^{\PsiMod}(e_1, e_2)
\end{align*}
and
\begin{align*}
(\nabla^\phi)^{\PsiMod}_\kX
&= \PsiMod \circ \nabla^\phi_{\PsiDer^{-1}(\kX)} \circ \PsiMod^{-1}
= \PsiMod \circ \phi \circ \nabla_{\PsiDer^{-1}(\kX)} \circ \phi^{-1} \circ \PsiMod^{-1}
\\
&= (\PsiMod \circ \phi \circ \PsiMod^{-1}) \!\circ\! (\PsiMod \circ \nabla_{\PsiDer^{-1}(\kX)} \circ \PsiMod^{-1}) \!\circ\! (\PsiMod \circ \phi \circ \PsiMod^{-1})^{-1}
\\
&= \phi^{\PsiMod} \circ \nabla^{\PsiMod}_\kX \circ (\phi^{\PsiMod})^{-1}
= (\nabla^{\PsiMod})^{\phi^{\PsiMod}}_\kX .
\end{align*}
Finally, one has 
\begin{align*}
\omega^\Psi(\kX) 
&= \nabla^\Psi_\kX \bbbone 
= \Psi \circ \nabla_{\PsiDer^{-1}(\kX)} \Psi^{-1} \bbbone 
= \Psi \circ \nabla_{\PsiDer^{-1}(\kX)} \bbbone 
\\
&= \Psi \left( \omega(\PsiDer^{-1}(\kX)) \right) 
= \Psi(\omega) (\kX)
\end{align*}
and $u^\Psi = \phi^\Psi(\bbbone) = \Psi \circ \phi \circ \Psi^{-1}(\bbbone) = \Psi \circ \phi(\bbbone) = \Psi(u)$.
\end{proof}

\medskip
Let us describe the degrees of freedom in the gauge sector of a NCGFT defined on $M_n$ and on $C^\infty(M) \otimes M_n$, see \cite{DuboKernMado90a, DuboKernMado90b, DuboMass98a, Mass12a, FranLazzMass14a} for some details.  We use some notations from Sect.~\ref{sec matrix algebras}.

\begin{example}[$\algA = M_n$]
\label{example Mn}
Let us consider $\algA = M_n$. Then $\Der(\algA) = \Der(M_n) \simeq \ksl_n$, and, for $k=1, \ldots, n^2-1$, let $\{ E_k \}$ be a basis of anti-Hermitean traceless matrices in $\ksl_n$ so that $\{ \partial_k \defeq \ad_{E_k} \}$ is a basis of real derivations of $M_n$.\footnote{We depart here from the conventions in many papers where the $E_k$ are chosen to be Hermitean and $\partial_k$ are defined as $\ad_{i E_k}$.} Let us consider the left module $\modM = \algA$ with the canonical Hermitian structure $h(a,b) \defeq a b^*$. There is a canonical connection $\mrnabla$ on $\modM$ defined by $\mrnabla_{\partial_k} a \defeq E_k a$ for any $k=1, \ldots, n^2-1$ and $a \in \algA = \modM$ with connection $1$-form  $\mromega(\partial_k) = E_k$. This canonical connection satisfies two important properties: firstly, its curvature is zero; secondly, it is gauge invariant (see also \cite{CagnMassWall11a} for another occurrence of such a canonical connection). Explicitly, one has  $\mromega = E_k \theta^k$, which makes it look very much like the Maurer-Cartan $1$-form on $\ksl_n$. It is then convenient to compare any connection $1$-form $\omega$ on $\modM$ to this canonical connection, by writing $\omega = \omega_k \theta^k = \mromega - B_k \theta^k = (E_k - B_k) \theta^k$. Then the curvature $2$-form $\Omega = \tfrac{1}{2} \Omega_{k\ell} \theta^k \wedge \theta^\ell$ has components $\Omega_{k\ell} \defeq \Omega(\partial_k, \partial_\ell) = - ([B_k, B_\ell] - C_{k\ell}^m B_m)$. This curvature vanishes iff $E_k \mapsto B_k$ is a representation of the Lie algebra $\ksl_n$ (for instance $B_k = 0$ or $B_k = E_k$). The connection $\omega$ is compatible with $h$ iff $\omega_k + \omega^*_k = 0$ for any $k$. Since the $E_k$'s are anti-Hermitean, this compatibility condition is equivalent to $B_k + B_k^* = 0$ for any $k$ and then $\Omega_{k\ell}^* = - \Omega_{k\ell}$. We can then decompose $B_k = B_k^\ell E_\ell + i B_k^0 \bbbone_n$ with real functions $B_k^\ell$, $\ell = 0, \ldots, n^2-1$, so that the number of degrees of freedom (number of real functions) in $\omega$ is $n^2 (n^2-1)$. The action of a gauge transformation $g \in U(n)$ induces the transformation $B_k \mapsto g^{-1} B_k g$ (the inhomogeneous part of the gauge transformation is absorbed by $\mromega$).

Notice that this approch is only interesting for $n \geq 2$ since for $n=1$, $M_1 = \bbC$ is commutative and so there is no derivation and so no degree of freedom $B_k$'s.

Suppose that the basis $\{ \partial_k \}$ is orthonormal for the metric $g$ defined as in Sect.~\ref{sec matrix algebras}. Since $\Omega \wedge \hstar \Omega = \tfrac{1}{2} \Omega_{k\ell} \Omega^{k\ell} \omega_{\vol}$ for $\Omega^{k\ell} = g^{kk'} g^{\ell \ell'} \Omega_{k'\ell'}$, the action is then $- \tfrac{1}{2} \sum_{k,\ell} \tr (\Omega_{k\ell})^2 = - \tfrac{1}{2} \sum_{k,\ell} \tr ([B_k, B_\ell] - C_{k\ell}^m B_m)^2$. Notice that $- \tr (\Omega_{k\ell})^2 = \tr (\Omega_{k\ell} \Omega_{k\ell}^*) \geq 0$.

From Prop.~\ref{prop transport of connections by automorphisms} and Examples~\ref{ex transport derivations inner automorphisms} and \ref{ex transport matrix by inner automorphisms}, an inner automorphism defined by a unitary element $u$ in $M_n$ produces a transport of all the structures defining the NCGFT on $M_n$. One has $\omega^u = U_k^\ell u \omega_\ell u^{-1} \theta^k$, and since $\mromega^u = U_k^\ell u E_\ell u^{-1} \theta^k = \mromega$, this implies that $B_k$ is mapped to $B^u_k = U_k^\ell u B_\ell u^{-1}$. One then has $[B^u_k, B^u_\ell] - C_{k\ell}^m B^u_m = U_k^{k'} U_\ell^{\ell'} u^{-1} \left( [B_{k'}, B_{\ell'}] - C_{k'\ell'}^{m' }B_{m'} \right) u$ and the Lagrangian in the $B^u_k$ is the same as the Lagrangian in the $B_k$. We conclude that such an action of inner automorphism is not relevant from a physical point of view.
\end{example}

\begin{example}[$\halgA = C^\infty(M) \otimes M_n$]
\label{example C(M) otimes Mn}
Let us consider the algebra $\halgA = C^\infty(M) \otimes M_n$ for a manifold $M$. The space of derivations is $\Der(\halgA) = [\Gamma(M) \otimes \bbbone_n] \oplus [C^\infty(M) \otimes \ksl_n]$ where $\Gamma(M) = \Der(C^\infty(M))$ is the space of vector fields on $M$. For $\mu = 1, \ldots, \dim M$, let $\{ \partial_\mu \}$ (usual partial derivatives in a coordinate system in a chart of $M$) be a basis of real derivations on the geometric part, and let $\{ \dd x^\mu \}$ be the dual basis of $1$-forms. Let us consider as before the left module $\modM = \halgA$ with the canonical Hermitian structure $h(a,b) \defeq a b^*$. Then a connection $1$-form $\omega$ can be written as $\omega = \omega_\mu \dd x^\mu + \omega_k \theta^k = A_\mu \dd x^\mu + (E_k - B_k) \theta^k$ with $A_\mu, B_k \in \halgA$ and this connection is compatible with $h$ when $A_\mu + A_\mu^* = 0$ and $B_k + B_k^* = 0$ (since the $E_k$'s are anti-Hermitean). As before, we can decompose $A_\mu = A_\mu^\ell E_\ell + i A_\mu^0 \bbbone_n$ and $B_k = B_k^\ell E_\ell + i B_k^0 \bbbone_n$ so that the number of degrees of freedom in $\omega$ is $n^2 (\dim M + n^2-1)$. A gauge transformation given by $g \in C^\infty(M) \otimes U(n)$ induces the transformations $A_\mu \mapsto g^{-1} A_\mu g - g^{-1} \dddR g$ and $B_k \mapsto g^{-1} B_k g$ where $\dddR$ is the ordinary de~Rham differential on $M$ (to simplify, we used the notation $\dd x^\mu$ instead of $\dddR x^\mu$). So $A_\mu \dd x^\mu$ can be identified with an ordinary $U(n)$-connection.

The curvature of $\omega$ can be decomposed into three parts: $\Omega = \tfrac{1}{2} \Omega_{\mu\nu} \dd x^\mu \wedge \dd x^\nu + \Omega_{\mu k} \dd x^\mu \wedge \theta^k + \tfrac{1}{2}  \Omega_{k\ell} \theta^k \wedge \theta^\ell$ with
\begin{align*}
\Omega_{\mu\nu}
&= \partial_\mu A_\nu - \partial_\nu A_\mu - [A_\mu, A_\nu],
\\
\Omega_{\mu k}
&= - ( \partial_\mu B_k - [A_\mu, B_k] ),
\\
\Omega_{k\ell}
&= - ( [B_k, B_\ell] - C_{k\ell}^m B_m ).
\end{align*}
The term $\Omega_{\mu\nu}$ is the usual field strength of $A_\mu$, $\Omega_{\mu k}$ is (up to a sign) the covariant derivative of $B_k$ along the connection $A_\mu$ and $\Omega_{k\ell}$ is the expression obtained for the algebra $\algA = M_n$. Using natural notions of metric and Hodge $\hstar$-operator in this context, a natural Lagrangian is the sum of 3 (positive) terms $- \tfrac{1}{2} \tr(\Omega_{\mu\nu} \Omega^{\mu\nu}) - \tr(\Omega_{\mu k} \Omega^{\mu k}) - \tfrac{1}{2} \tr(\Omega_{k\ell} \Omega^{k\ell})$. Finding a minimal configuration for such a Lagrangian is equivalent to minimizing independently these 3 terms. The last one vanishes if and only if $E_k \mapsto B_k$ is a representation of $\ksl_n$. One possibility is the take $B_k = 0$ for all $k$ (referred to as the “null-configuration” in the following), which cancels also the second term. Then one reduces the theory to massless gauge fields $A_\mu$. Another more stimulating configuration is to consider $B_k = E_k$ (referred to as the “basis-configuration” in the following), and then the second term reduces to $- \tr ([A_\mu, E_k] [A^\mu, E^k])$, which, after developing, produces mass terms for the $A_\mu$ fields, see Lemma~\ref{lem mass rep-config}. This is similar to the SSBM implemented in the SMPP to give masses to some gauge fields.

Notice that for an ordinary Yang-Mills theory in the framework of fiber bundles and connections, with structure group $U(n)$, we have only the fields $A_\mu^\ell$, $\ell = 0, \ldots, n^2-1$ since there is no “algebraic part” which produces the $B_k$'s. With the structure group $SU(n)$, there is no field $A_\mu^0$ (the matrices $E_\ell$, $\ell = 1, \ldots, n^2-1$, generate the real Lie algebra $\ksu(n)$).

Contrary to Example~\ref{example Mn}, this case is also interesting for $n=1$. In that case, the degrees of freedom are only in the spatial direction (the $A_\mu$'s) and they can be used to construct an ordinary $U(1)$ gauge field theory.

Once again, one can ask about the action of an inner automorphism defined by a unitary element in $\halgA$. The action of such an automorphism on a “spatial” derivation $X \in \Gamma(M)$ is given by $\PsiDer(X) = X + \ad_{u(X\cdotaction u^{-1})}$ (see Example~\ref{ex transport derivations inner automorphisms}). If $u$ is a unitary in $M_n$ (not depending on $M$), then one gets $\PsiDer(X) = X$. This implies that the spatial directions (the $\partial_\mu$'s) are only affected by $u$ through the action of $\Psi$, $\omega^u_\mu = u \omega_\mu u^{-1}$, while the “inner” directions (the $\partial_k$'s) change according to the rules given in Example~\ref{ex transport matrix by inner automorphisms}. This implies that the Lagrangian in the new fields is the same as the one in the original fields and so such an action of inner automorphism is not relevant from a physical point of view. When $u$ is depending on $M$, the second term in $\PsiDer(X)$ does not vanish and it produces mixing between spatial directions and inner directions: some degrees of freedom in the $B_k$'s are sent in the spatial part $\omega^u_\mu$. This situation will not be considered in the following.
\end{example}

\begin{lemma}
\label{lem mass rep-config}
Let us consider a NCGFT as in Example~\ref{example C(M) otimes Mn}. In the basis-configuration for the $B_k$'s, the masses induced on the fields $A_\mu^\ell$ for $\ell = 1, \dots, n^2-1$, in the decomposition $A_\mu = A_\mu^\ell E_\ell + i A_\mu^0 \bbbone_n$, are all the same and equal to $m_{\text{basis-config}} = \sqrt{2n}$, while the field $A_\mu^0$ is mass-less.
\end{lemma}

\begin{proof}
Using the metric $g$ defined as $g(E_k, E_\ell) = \tr(E_k E_\ell)$ (see Sect.~\ref{sec matrix algebras}), the masses for the fields $A_\mu^\ell$, $\ell = 1, \dots, n^2-1$, are given by the term 
\begin{align*}
M^2_{\ell_1 \ell_2} A_\mu^{\ell_1} A^{\mu, \ell_2}
&= - g^{k_1 k_2} \tr ([A_\mu, E_{k_1}] [A^\mu, E_{k_2}]) 
\\
&= - A_\mu^{\ell_1} A^{\mu, \ell_2} g^{k_1 k_2} \tr( [E_{\ell_1}, E_{k_1}] [E_{\ell_2}, E_{k_2}] )
\\
&= - A_\mu^{\ell_1} A^{\mu, \ell_2} g^{k_1 k_2} C_{\ell_1 k_1}^{m_1} C_{\ell_2 k_2}^{m_2} \tr( E_{m_1} E_{m_2})
\\
&= - A_\mu^{\ell_1} A^{\mu, \ell_2} g^{k_1 k_2} g_{m_1 m_2} C_{\ell_1 k_1}^{m_1} C_{\ell_2 k_2}^{m_2}
\end{align*}
where $A_\mu = A_\mu^\ell E_\ell + i A_\mu^0 \bbbone_n$. Since the field $A_\mu^0$ disappears, its mass is $0$.

For any $X, Y \in \ksu(n)$, the Killing form $K(X, Y) = \tr( \ad_{X} \circ \ad_{Y} )$ satisfies $K(X,Y) = 2n \tr(XY)$ so that, on the one hand, $K_{k\ell} \defeq K(E_k, E_\ell) = 2 n \, g_{k\ell}$ and on the other hand, $K_{k\ell} = C_{k m}^n C_{\ell n}^m$. Let us define $C_{k \ell m} \defeq g_{m n} C_{k \ell}^n$, so that $C_{k \ell}^n = g^{m n} C_{k \ell m}$ and $C_{k \ell m}$ is completely antisymmetric in $(k, \ell, m)$. This leads to $g^{k_1 k_2} g_{m_1 m_2} C_{\ell_1 k_1}^{m_1} C_{\ell_2 k_2}^{m_2} = g^{k_1 k_2} C_{\ell_1 k_1 m} C_{\ell_2 k_2}^{m} = - g^{k_1 k_2} C_{\ell_1 m k_1} C_{\ell_2 k_2}^{m} = - C_{\ell_1 m}^k C_{\ell_2 k}^{m}$, so that $M^2_{\ell_1 \ell_2} = C_{\ell_1 m}^k C_{\ell_2 k}^{m} = K_{\ell_1 \ell_2} = 2 n \, g_{\ell_1 \ell_2}$. This proves that the diagonalization of $(M^2_{\ell_1 \ell_2})$ gives a unique eigenvalue $2 n$ so that there is a unique mass $m_{\text{basis-config}} = \sqrt{2n}$.
\end{proof}

\section{Some properties on sums of algebras}
\label{sec some properties on sums of algebras}

In this section we consider the derivation-based differential calculus on algebras decomposed as
\begin{align*}
\algA = \algA_1 \toplus \cdots \toplus \algA_r = \toplus_{i=1}^{r} \algA_i
\end{align*}
where $\algA_i$ are unital algebras, not necessary of finite dimension. We define respectively
\begin{align*}
&\pi^i : \algA \to \algA_i,
\text{ and }
\iota_i : \algA_i \to \algA
\end{align*}
as the natural projection on the $i$-th term and the natural inclusion of the $i$-th term.

Some results are presented using this full generality but others will require $\algA_i = M_n$, the unital algebra of $n \times n$ matrices over $\bbC$. As far as we know, the results presented here have never been exposed elsewhere in a systematic way.

\subsection{Center and derivations}

\begin{lemma}[Center of $\algA$]
The center of $\algA$ is $\calZ(\algA) = \toplus_{i=1}^{r} \calZ(\algA_i)$.
\end{lemma}

\begin{proof}
Every $a = \toplus_{i=1}^{r} a_i \in \calZ(\algA)$ must commute with any $b = 0 \toplus \cdots 0 \toplus b_j \toplus 0 \toplus 0$ for any $j =1, \dots, r$ and any $b_j \in \algA_j$. This implies that $a_j \in \calZ(\algA_j)$ for any $j$. The result follows since $\toplus_{i=1}^{r} \calZ(\algA_i) \subset \calZ(\algA)$.
\end{proof}

Let us introduce the convenient notation for the elements 
\begin{align*}
\hbbbone_i \defeq \iota_i(\bbbone) = 0 \toplus \cdots \toplus 0 \toplus \bbbone \toplus 0 \toplus \cdots \toplus 0 \in \calZ(\algA) \subset \algA. 
\end{align*}
Notice that we use the fact that the $\algA_i$'s are unital.

\begin{proposition}[Decomposition of derivations]
\label{prop decomposition derivations}
One has
\begin{align}
\label{eq Der(A)}
\Der(\algA) = \toplus_{i=1}^{r} \Der(\algA_i),
\end{align}
\textsl{i.e.} for any $a = \toplus_{i=1}^{r} a_i \in \algA$ and $\kX = \toplus_{i=1}^{r} \kX_i \in \Der(\algA)$, one has $\kX (a) = \toplus_{i=1}^{r} \kX_i (a_i)$.

This decomposition holds true as Lie algebras and modules over $\calZ(\algA)$ on the left and over $\toplus_{i=1}^{r} \calZ(\algA_i)$ on the right.

If $\Der(\algA_i) = \Int(\algA_i)$ for any $i=1, \dots, r$, then 
\begin{align}
\label{eq Der(A) = Int(A)}
\Der(\algA) = \Int(\algA) = \toplus_{i=1}^{r} \Int(\algA_i)
\end{align}
\end{proposition}

\begin{proof}
The vector space decomposition \eqref{eq Der(A)} can be established using the maps $\kX_i^j \defeq \pi^j \circ \kX \circ \iota_i : \algA_i \to \algA_j$. For any $a = \toplus_{i=1}^{r} a_i$ and $b = \toplus_{i=1}^{r} b_i$, one has
\begin{align*}
\kX(a) 
&= \toplus_{j=1}^{r} \left( \tsum_{i=1}^{r} \kX_i^j (a_i) \right)
\end{align*}
and the Leibniz rule $\kX(a b) = \kX(a) b + a \kX(b)$ can then be written as
\begin{align*}
\toplus_{j=1}^{r} \left( \tsum_{i=1}^{r} \kX_i^j (a_i b_i) \right)
&= \toplus_{j=1}^{r} \left( \tsum_{i=1}^{r} \kX_i^j (a_i) \right)  b_j
+ \toplus_{j=1}^{r} a_j \left( \tsum_{i=1}^{r} \kX_i^j (b_i) \right).
\end{align*}
For a fixed $k$, take $a_i = b_i = 0$ for $i \neq k$. Then this relation reduces to
\begin{align*}
\toplus_{j=1}^{r} \kX_k^j (a_k b_k)
&= 0 \toplus \cdots \toplus \kX_k^k(a_k) b_k \toplus 0 \toplus \cdots \toplus 0
+
 0 \toplus \cdots \toplus a_k \kX_k^k(b_k) \toplus 0 \toplus \cdots \toplus 0.
\end{align*}
The $k$-th term shows that $\kX_k^k \in \Der(\algA_k)$, which implies in particular that $\kX_k^k(\bbbone) = 0$. Then, with $b_k = \bbbone$, one gets $\kX_k^j = 0$ for $j \neq k$. This shows that $\kX(a) = \toplus_{i=1}^{r} \kX_i (a_i)$ with $\kX_i \defeq \kX_i^i \in \Der(\algA_i)$ and we write
\begin{align*}
\kX = \toplus_{i=1}^{r} \kX_i 
\end{align*}

For any $f = \toplus_{i=1}^{r} f_i \in \calZ(\algA) = \toplus_{i=1}^{r} \calZ(\algA_i)$, one has
\begin{align*}
(f \kX) (a)
&= f (\toplus_{i=1}^{r} \kX_i (a_i))
= \toplus_{i=1}^{r} (f_i \kX_i) (a_i)
\end{align*}
so that \eqref{eq Der(A)} holds true as modules over $\calZ(\algA)$ on the left and over $\toplus_{i=1}^{r} \calZ(\algA_i)$ on the right.

For $\kY = \toplus_{i=1}^{r} \kY_i \in \Der(\algA)$, one has
\begin{align*}
[\kX, \kY] (a)
&= \kX \circ \kY(a) - \kY \circ \kX(a)
= \kX \left( \toplus_{i=1}^{r} \kY_i(a_i) \right) - \kY \left( \toplus_{i=1}^{r} \kX_i(a_i) \right)
\\
&= \toplus_{i=1}^{r} \kX_i \circ \kY_i(a_i) - \toplus_{i=1}^{r} \kY_i \circ \kX_i(a_i)
\\
&= \toplus_{i=1}^{r} [\kX_i, \kY_i](a_i)
\end{align*}
so that 
\begin{align*}
[\kX, \kY] = \toplus_{i=1}^{r} [\kX_i, \kY_i]
\end{align*}
This shows that \eqref{eq Der(A)} holds true as Lie algebras.

The proof of \eqref{eq Der(A) = Int(A)} is then a direct consequence: assuming $\Der(\algA_i) = \Int(\algA_i)$, one has $\Int(\algA) \subset \Der(\algA) = \toplus_{i=1}^{r} \Int(\algA_i) \subset \Int(\algA)$ where the last inclusion follows from $\ad_{a_1} \toplus \cdots \toplus \ad_{a_r} = \ad_{a_1 \toplus \cdots \toplus a_r}$.
\end{proof}

We define $\iotaDer_i : \Der(\algA_i) \to \Der(\algA)$ as the inclusion $\kX_i \mapsto 0 \toplus \cdots \toplus 0 \toplus \kX_i \toplus 0 \cdots \toplus 0$ in the $i$-th term. This is a morphism of Lie algebras and for any $a_i \in \algA_i$ and $f_i \in \calZ(\algA_i)$, one has $\iota_i(\kX_i(a_i)) = \iotaDer_i( \kX_i)(\iota_i(a_i))$ and $\iotaDer_i( f_i \kX_i) = \iota_i(f_i) \iotaDer_i( \kX_i)$. Notice also that for any $\kX = \toplus_{i=1}^{r} \kX_i$, one has $\hbbbone_i  \kX = \iotaDer_i( \kX_i) = \hbbbone_i \iotaDer_i( \kX_i)$.

\bigskip
Using results given in Sect.~\ref{sec matrix algebras}, let us conclude this subsection with this direct Corollary of Prop.~\ref{prop decomposition derivations}:
\begin{corollary}
For $\algA = M_{n_1} \toplus \cdots \toplus M_{n_r}$, one has $\Der(\algA) = \Int(\algA) \simeq \ksl_{n_1} \toplus \cdots \toplus \ksl_{n_r}$.
\end{corollary}

\subsection{Derivation-based differential calculus}

Here is a useful result concerning the structure of the derivation-based differential calculus associated to $\algA$.
\begin{proposition}[Decomposition of forms]
\label{prop decomposition forms}
For any $p \geq 0$, one has
\begin{align*}
\OmegaDer^p(\algA) &= \toplus_{i=1}^{r} \OmegaDer^p(\algA_i),
\end{align*}
that is, any $\omega \in \OmegaDer^p(\algA)$ decomposes as $\omega = \toplus_{i=1}^{r} \omega_i$ with $\omega_i \in \OmegaDer^p(\algA_i)$ and 
\begin{align*}
\omega(\kX_1, \dots, \kX_p)
&= \toplus_{i=1}^{r} \omega_i( \kX_{1, i}, \dots, \kX_{p, i} )
\quad\text{for any $\kX_k = \toplus_{i=1}^{r} \kX_{k, i} \in \Der(\algA)$.}
\end{align*}

This decomposition is compatible with the $\calZ(\algA)$-linearity on the left and the $\toplus_{i=1}^{r} \calZ(\algA_i)$-linearity on the right, and it is compatible with  the products in $\OmegaDer^\grast(\algA)$ and $\OmegaDer^\grast(\algA_i)$.

The differential $\dd$ on $\OmegaDer^\grast(\algA)$ decomposes along the differentials $\dd_i$ on $\OmegaDer^\grast(\algA_i)$ as
\begin{align*}
\dd \omega 
&= \toplus_{i=1}^{r} \dd_i \omega_i
\end{align*}
\end{proposition}

We will extend the projection map $\pi^i : \OmegaDer^\grast(\algA) \to \OmegaDer^\grast(\algA_i)$ with the same notation.

\begin{proof}
Let $\omega \in \OmegaDer^p(\algA)$ and define, for any $i_1, \dots, i_p, j = 1, \dots, r$,
\begin{align*}
\omega^j_{i_1, \dots, i_p} : \Der(\algA_{i_1}) \times \cdots \times \Der(\algA_{i_p}) \to \algA_j
\end{align*}
by
\begin{align*}
\omega^j_{i_1, \dots, i_p}( \kX_{1, i_1}, \dots, \kX_{p, i_p} ) = \pi^j \circ \omega ( \iotaDer_{i_1}(\kX_{1, i_1}), \dots, \iotaDer_{i_p}(\kX_{p, i_p}) )
\end{align*}
for any $\kX_{k, i_k} \in \Der(\algA_{i_k})$ ($k=1, \dots, p$). Then one has
\begin{align*}
\omega(\kX_1, \dots, \kX_p)
&= \toplus_{j=1}^{r} \left( \tsum_{i_1, \dots, i_p = 1}^{r} \omega^j_{i_1, \dots, i_p}( \kX_{1, i_1}, \dots, \kX_{p, i_p} ) \right)
\end{align*}
with $\kX_k = \toplus_{i_k=1}^{r} \kX_{k, i_k}$. Let $f_k = \toplus_{i_k=1}^{r} f_{k, i_k} \in \calZ(\algA)$, then applying $\omega$ on the $p$ derivations $f_k \kX_k$, the $\calZ(\algA)$-linearity of $\omega$ gives
\begin{align}
\toplus_{j=1}^{r} & \left( \tsum_{i_1, \dots, i_p = 1}^{r} \omega^j_{i_1, \dots, i_p}( f_{1,i_1} \kX_{1, i_1}, \dots, f_{p,i_p} \kX_{p, i_p} ) \right)
\nonumber
\\
\label{eq Z(A) linearity omega general}
&=
\toplus_{j=1}^{r} f_{1,j} \cdots f_{p,j} \left( \tsum_{i_1, \dots, i_p = 1}^{r} \omega^j_{i_1, \dots, i_p}( \kX_{1, i_1}, \dots, \kX_{p, i_p} ) \right)
\end{align}
The arbitrariness on the $f_k$'s permits to simplify this relation in successive steps. Let us fix $j_1$ and take $f_{1} = \hbbbone_{j_1}$.
Then on the LHS, the term at $j$ in $\toplus_{j=1}^{r}$ reduces to
\begin{align*}
\tsum_{i_2, \dots, i_p = 1}^{r} \omega^j_{j_1, i_2, \dots, i_p}( \kX_{1, j_1}, f_{2,i_2} \kX_{p, i_p}, \dots, f_{p,i_p} \kX_{p, i_p} ). 
\end{align*}
On the RHS, the only non zero term along $\toplus_{j=1}^{r}$ occurs at $j = j_1$ and gives 
\begin{align*}
f_{2,j_1} \cdots f_{p,j_1} \left( \tsum_{i_1, \dots, i_p = 1}^{r} \omega^{j_1}_{i_1, \dots, i_p}( \kX_{1, i_1}, \dots, \kX_{p, i_p} ) \right).
\end{align*}
By arbitrariness on the $\kX_k$'s, this implies that for $j \neq j_1$ one has $\omega^j_{i_1, \dots, i_p} = 0$ and the remaining non trivial relation becomes (substituting $j_1$ to $j$)
\begin{align*}
\tsum_{i_2, \dots, i_p = 1}^{r} & \omega^j_{j, i_2, \dots, i_p}( \kX_{1, j}, f_{2,i_2} \kX_{2, i_2}, \dots, f_{p,i_p} \kX_{p, i_p} )
\\
&=
f_{2,j} \cdots f_{p,j} \left( \tsum_{i_2, \dots, i_p = 1}^{r} \omega^j_{j, i_2, \dots, i_p}( \kX_{1, j}, \kX_{2, i_2}, \dots, \kX_{p, i_p} ) \right).
\end{align*}
Let us now fix $j_2$ and take $f_{2} = \hbbbone_{j_2}$. Then the relation first gives
\begin{align*}
\tsum_{i_3, \dots, i_p = 1}^{r} \omega^j_{j, j_2, i_3, \dots, i_p} & ( \kX_{1, j}, \kX_{2, j_2}, f_{3,i_3} \kX_{3, i_3}, \dots, f_{p,i_p} \kX_{p, i_p} ) 
\\
&= 0  \text{ for $j \neq j_2$},
\end{align*}
which implies $\omega^j_{j, i_2, i_3, \dots, i_p} = 0$ for $i_2 \neq j$. Then, with this relation, the non vanishing term (at $j=j_2$) simplifies to (substituting $j_2$ to $j$)
\begin{align*}
\tsum_{i_3, \dots, i_p = 1}^{r} & \omega^{j}_{j, j, i_3, \dots, i_p}( \kX_{1, j}, \kX_{2, j}, f_{3,i_3} \kX_{3, i_3}, \dots, f_{p,i_p} \kX_{p, i_p} )
\\
&=
f_{3,j} \cdots f_{p,j} \left( \tsum_{i_3, \dots, i_p = 1}^{r} \omega^{j}_{j, j, i_3, \dots, i_p}( \kX_{1, j}, \kX_{2, j}, \kX_{3, i_3},\dots, \kX_{p, i_p} ) \right).
\end{align*}
We can repeat this argument for $i_k$ up to $k=p$ and conclude that the only non zero maps $\omega^j_{i_1, \dots, i_p}$ are $\omega^j_{j, \dots, j}$ for $j=1, \dots, r$. Defining $\omega_i \defeq \omega^i_{i, \dots, i}$, one then gets
\begin{align*}
\omega(\kX_1, \dots, \kX_p)
&= \toplus_{i=1}^{r} \omega_i( \kX_{1, i}, \dots, \kX_{p, i} )
\end{align*}
and \eqref{eq Z(A) linearity omega general} reduces to
\begin{align*}
\toplus_{i=1}^{r}  \omega_i( f_{1,i} \kX_{1, i}, \dots, f_{p,i} \kX_{p, i} )
&= \toplus_{i=1}^{r} f_{1,i} \cdots f_{p,i} \omega_i( \kX_{1, i}, \dots, \kX_{p, i} ).
\end{align*}
This shows that for any $i$, $\omega_i$ is $\calZ(\algA_i)$-linear. Finally, the antisymmetry of $\omega$ implies antisymmetry of the $\omega_i$'s. This proves that $\omega_i \in \OmegaDer^p(\algA_i)$.

To prove the compatibility of this decomposition with the products, consider $\omega = \toplus_{i=1}^{r}  \omega_i \in \OmegaDer^p(\algA)$ and $\eta = \toplus_{i=1}^{r}  \eta_i \in \OmegaDer^q(\algA)$, and $p+q$ derivations $\kX_{k} = \toplus_{i_k=1}^{r} \kX_{k, i_k} \in \Der(\algA)$. Then by definition
\begin{align*}
(\omega &\wedge \eta) (\kX_{1}, \dots, \kX_{p+q})
\\
&= \tfrac{1}{p! q!} \sum_{\sigma \in \kS_{p+q}} (-1)^{\abs{\sigma}} \omega (\kX_{\sigma(1)}, \dots, \kX_{\sigma(p)}) \eta (\kX_{\sigma(p+1)}, \dots, \kX_{\sigma(p+q)})
\\
&= \begin{multlined}[t]
\tfrac{1}{p! q!} \sum_{\sigma \in \kS_{p+q}} (-1)^{\abs{\sigma}} 
\left( \toplus_{i=1}^{r} \omega_i (\kX_{\sigma(1), i}, \dots, \kX_{\sigma(p), i}) \right)
\\
\left( \toplus_{j=1}^{r} \eta_j (\kX_{\sigma(p+1), j}, \dots, \kX_{\sigma(p+q), j}) \right)
\end{multlined}
\\
&= \tfrac{1}{p! q!} \sum_{\sigma \in \kS_{p+q}} (-1)^{\abs{\sigma}} 
\left( \toplus_{i=1}^{r} 
\omega_i (\kX_{\sigma(1), i}, \dots, \kX_{\sigma(p), i})
\eta_i (\kX_{\sigma(p+1), i}, \dots, \kX_{\sigma(p+q), i}) 
\right)
\\
&= \toplus_{i=1}^{r} \left( 
\tfrac{1}{p! q!} \sum_{\sigma \in \kS_{p+q}} (-1)^{\abs{\sigma}} 
\omega_i (\kX_{\sigma(1), i}, \dots, \kX_{\sigma(p), i})
\eta_i (\kX_{\sigma(p+1), i}, \dots, \kX_{\sigma(p+q), i}) 
\right)
\\
&= \toplus_{i=1}^{r}
(\omega_i \wedge \eta_i) (\kX_{1, i}, \dots, \kX_{p+q, i})
\end{align*}
so that $\omega \wedge \eta =  \toplus_{i=1}^{r} \omega_i \wedge \eta_i$.

Using similar notations, one has
\begin{align*}
(\dd \omega) & (\kX_{1}, \dots, \kX_{p+1})
\\
&= \begin{multlined}[t]
\sum_{k=1}^{p+1} (-1)^{k+1} \kX_{k} \cdotaction \omega (\kX_{1}, \dots \omi{k} \dots, \kX_{p+1})
\\
+ \sum_{1 \leq k < \ell \leq p+1} (-1)^{k+\ell} \omega ([\kX_{k}, \kX_{\ell}], \dots \omi{k} \dots \omi{\ell} \dots, \kX_{p+1})
\end{multlined}
\\
&= \begin{multlined}[t]
\sum_{k=1}^{p+1} (-1)^{k+1} \toplus_{i=1}^{r} \kX_{k, i} \cdotaction \omega_i (\kX_{1, i}, \dots \omi{k} \dots, \kX_{p+1, i})
\\
+ \sum_{1 \leq k < \ell \leq p+1} (-1)^{k+\ell} \toplus_{i=1}^{r} \omega_i ([\kX_{k, i}, \kX_{\ell, i}], \dots \omi{k} \dots \omi{\ell} \dots, \kX_{p+1, i})
\end{multlined}
\\
&= \toplus_{i=1}^{r} \Big( \begin{multlined}[t]
\sum_{k=1}^{p+1} (-1)^{k+1} \kX_{k, i} \cdotaction \omega_i (\kX_{1, i}, \dots \omi{k} \dots, \kX_{p+1, i})
\\
+ \sum_{1 \leq k < \ell \leq p+1} (-1)^{k+\ell} \omega_i ([\kX_{k, i}, \kX_{\ell, i}], \dots \omi{k} \dots \omi{\ell} \dots, \kX_{p+1, i})
\Big)
\end{multlined}
\\
&= \toplus_{i=1}^{r} (\dd_i \omega_i) (\kX_{1, i}, \dots, \kX_{p+1, i})
\end{align*}
so that $\dd \omega = \toplus_{i=1}^{r} \dd_i \omega_i$.
\end{proof}

\subsection{Modules and connections}
\label{sec modules and connections}

We consider left modules on $\algA$ of the form $\modM = \toplus_{i=1}^{r} \modM_i$ where $\modM_i$ is a left module on $\algA_i$. This requirement is sufficient for the particular situation $\algA = M_{n_1} \toplus \cdots \toplus M_{n_r}$ since, according to \cite[Cor.~III.1.2]{Davi96a}, the modules of this algebra are of the form $\bbC^{n_1} \otimes \bbC^{\alpha_1} \toplus \cdots \toplus \bbC^{n_r} \otimes \bbC^{\alpha_r} = M_{n_1 \times \alpha_1} \toplus \cdots \toplus M_{n_r \times \alpha_r}$ for some integers $\alpha_i$, where $M_{n_i \times \alpha_i}$ is the vector space of $n_i \times \alpha_i$ matrices over $\bbC$.

\medskip
Define $\piMod^i : \modM \to \modM_i$ as the projection on the $i$-th term and $\iotaMod_i : \modM_i \to \modM$ as the natural inclusion. Then, $\piMod^i \circ \iotaMod_i = \Id_{\modM_i}$ and for any $a =\toplus_{i=1}^{r} a_i$ and $e = \toplus_{i=1}^{r} e_i$, one has $\piMod^i (a e) = \pi^i(a) \piMod^i (e)$ and $\piMod^i (\iota_i(a_i) e) = a_i \piMod^i (e)$.

\begin{proposition}[Decomposition of connections]
\label{prop decomposition connections}
A connection $\nabla$ on the left $\algA$ module $\modM$ defines a unique family of connections $\nabla^i$ on the left $\algA_i$ modules $\modM_i$ such that for any $e = \toplus_{i=1}^{r} e_i$ and any $\kX = \toplus_{i=1}^{r} \kX_i$, one has
\begin{align*}
\nabla_\kX e
&= \toplus_{i=1}^{r} \nabla^i_{\kX_i} e_i.
\end{align*}
Denote by $R_i$ the curvature associated to $\nabla^i$, then, for any $\kX = \toplus_{i=1}^{r} \kX_i$, any $\kY = \toplus_{i=1}^{r} \kY_i$, and any $e = \toplus_{i=1}^{r} e_i$, one has
\begin{align*}
R(\kX, \kY) e 
&=  \toplus_{i=1}^{r} R_i(\kX_i, \kY_i) e_i
\end{align*}
\end{proposition}

\begin{proof}
Since $\kX = \toplus_{i=1}^{r} \kX_i = \tsum_{i=1}^{r} \iotaDer_i( \kX_i)$, one has $\nabla_{\kX} e = \tsum_{i=1}^{r} \nabla_{\iotaDer_i( \kX_i)} e$ for any $e \in \modM$. This implies that $\nabla_{\kX}$ is completely given by the $r$ maps $\nabla_{\iotaDer_i(\kX_i)} : \modM \to \modM$.

So, for a fixed $i$ and for any $\kX_i \in \Der(\algA_i)$, let us study the map $\nabla_{\iotaDer_i(\kX_i)} : \modM \to \modM$. Since $\iotaDer_i( \kX_i) = \hbbbone_i \iotaDer_i( \kX_i)$, one has, for any $e \in \modM$, $\nabla_{\iotaDer_i(\kX_i)} e = \nabla_{\hbbbone_i  \iotaDer_i(\kX_i)} e = \hbbbone_i  \nabla_{\iotaDer_i(\kX_i)} e$ so that $\piMod^j \circ \nabla_{\iotaDer_i(\kX_i)} = 0$ for $j \neq i$. In other words, $\nabla_{\iotaDer_i(\kX_i)} e$ takes its values in $\modM_i$.

For a fixed $j$, take now $e = \iotaMod_j(e_j)$ for some $e_j \in \modM_j$. Since $\hbbbone_j \iotaMod_j(e_j) = \iotaMod_j(e_j)$ one has $\nabla_{\iotaDer_i(\kX_i)} \iotaMod_j(e_j) = \nabla_{\iotaDer_i(\kX_i)} (\hbbbone_j \iotaMod_j(e_j)) = (\iotaDer_i(\kX_i) \cdotaction \hbbbone_j) \iotaMod_j(e_j) + \hbbbone_j \nabla_{\iotaDer_i(\kX_i)} e_j = \hbbbone_j \nabla_{\iotaDer_i(\kX_i)} e_j$ since $\iotaDer_i(\kX_i) \cdotaction \hbbbone_j = 0$ whatever $i$ and $j$. If $j \neq i$, then $\nabla_{\iotaDer_i(\kX_i)} \iotaMod_j(e_j) = \hbbbone_j \nabla_{\iotaDer_i(\kX_i)} e_j = 0$ since $\nabla_{\iotaDer_i(\kX_i)} e_j$ has only components in $\modM_i$. This implies that  $\nabla_{\iotaDer_i(\kX_i)}$ is only non zero on components in $\modM_i$, and so defines a map
\begin{equation*}
\nabla^i_{\kX_i} \defeq \piMod^i \circ \nabla_{\iotaDer_i(\kX_i)} \circ \iotaMod_i :  \modM_i \to \modM_i.
\end{equation*}
Then, by construction, one has, for $e = \toplus_{i=1}^{r} e_i$ and $\kX = \toplus_{i=1}^{r} \kX_i$, $\nabla_\kX e = \toplus_{i=1}^{r} \nabla^i_{\kX_i} e_i$.

Now, let $f_i \in \calZ(\algA_i)$, then 
\begin{align*}
\nabla^i_{f_i \kX_i} e_i 
&= \piMod^i \left( \nabla_{\iotaDer_i(f_i \kX_i)} \circ \iotaMod_i(e_i) \right) 
= \piMod^i \left( \nabla_{\iota_i(f_i) \iotaDer_i(\kX_i)} \circ \iotaMod_i(e_i)\right) 
\\
&= \piMod^i \left( \iota_i(f_i) \nabla_{\iotaDer_i(\kX_i)} \circ \iotaMod_i(e_i) \right) 
= f_i \piMod^i \circ \nabla_{\iotaDer_i(\kX_i)} \circ \iotaMod_i(e_i)
\\
&= f_i \nabla^i_{\kX_i} e_i 
\end{align*}
Let $a_i \in \algA_i$, then
\begin{align*}
\nabla^i_{\kX_i} a_i e_i 
&= \piMod^i \left( \nabla_{\iotaDer_i(\kX_i)} \circ \iotaMod_i(a_i e_i) \right)
= \piMod^i \left( \nabla_{\iotaDer_i(\kX_i)} \circ \iota_i(a_i) \iotaMod_i(e_i) \right)
\\
&= \piMod^i \left( (\iotaDer_i(\kX_i) \cdotaction \iota_i(a_i)) \iotaMod_i(e_i) \right)
+ \piMod^i \left( \iota_i(a_i) \nabla_{\iotaDer_i(\kX_i)} \circ \iotaMod_i(e_i) \right)
\\
&= (\kX_i \cdotaction a_i) \piMod^i \circ \iotaMod_i(e_i)
+ a_i \piMod^i \circ \nabla_{\iotaDer_i(\kX_i)} \circ \iotaMod_i(e_i)
\\
&= (\kX_i \cdotaction a_i) e_i + a_i \nabla^i_{\kX_i} e_i
\end{align*}
These two relations show that $\nabla^i$ defines a connection on the left $\algA_i$ module $\modM_i$.

Concerning the curvature, one has $\nabla_\kX \nabla_\kY e = \nabla_\kX ( \toplus_{i=1}^{r}  \nabla^i_{\kY_i} e_i) = \toplus_{i=1}^{r} \nabla^i_{\kX_i} \nabla^i_{\kY_i} e_i$ and $\nabla_{[\kX, \kY]} e = \toplus_{i=1}^{r} \nabla^i_{[\kX_i, \kY_i]} e_i$ so that $R(\kX, \kY) e = \toplus_{i=1}^{r} ( [\nabla^i_{\kX_i}, \nabla^i_{\kY_i}] - \nabla^i_{[\kX_i, \kY_i]} ) e_i = \toplus_{i=1}^{r} R_i(\kX_i, \kY_i) e_i$.
\end{proof}

Let us now consider the special case $\modM = \algA$ with the natural left module structure. In that situation, we can characterize $\nabla$ by its connection $1$-form $\omega \in \OmegaDer^1(\algA)$ defined by $\omega(\kX) \defeq \nabla_{\kX} \bbbone$ and its curvature takes the form of the multiplication on the right by the curvature $2$-form $\Omega \in \OmegaDer^2(\algA)$ defined by $\Omega(\kX, \kY) \defeq (\dd \omega)(\kX, \kY) - [\omega(\kX), \omega(\kY)]$.

\begin{proposition}
In the previous situation, the decomposition of the connection $\nabla_\kX = \toplus_{i=1}^{r} \nabla^i_{\kX_i}$ in Prop~\ref{prop decomposition connections} is related to the decomposition of the connection $1$-form $\omega =  \toplus_{i=1}^{r} \omega_i$ in Prop.~\ref{prop decomposition forms} where $\omega_i \in \OmegaDer^1(\algA_i)$ is the connection $1$-form associated to the connection $\nabla^i$.

In the same way, the connection $2$-form $\Omega$ of $\nabla$ decomposes along the connection $2$-forms $\Omega_i$ of $\nabla^i$: $\Omega = \toplus_{i=1}^{r} \Omega_i$.
\end{proposition}

\begin{proof}
One has $\bbbone = \toplus_{i=1}^{r} \bbbone_i$ where $\bbbone_i$ is the unit in $\algA_i$. With $\kX = \toplus_{i=1}^{r} \kX_i$, one then has $\omega(\kX) = \nabla_{\kX} \bbbone = \toplus_{i=1}^{r} \nabla^i_{\kX_i} \bbbone_i = \toplus_{i=1}^{r} \omega_i(\kX_i)$.

The curvature $2$-forms are defined in term of differentials and Lie brackets (commutators in the respective algebras) from the connection $1$-forms. We have shown that these operations respect the decomposition of forms. This proves the relation of the curvature $2$-form of $\nabla$.
\end{proof}

\subsection{\texorpdfstring{Metric and Hodge $\hstar$-operator}{Metric and Hodge *-operator}}
\label{sec metric hodge}

In this subsection, we consider only the situation $\algA_i = M_{n_i}$. This permits to limit the study of metrics and Hodge $\hstar$-operators to the structures defined in Sect.~\ref{sec matrix algebras} (see also \cite{DuboKernMado90b, Mass95a, DuboMass98a, Mass08b, Mass12a}).

For every $\algA_i = M_{n_i}$, one introduces a basis $\{ E^{i}_{\alpha} \}_{\alpha \in I_{i}}$ of $\ksl_{n_i}$ where $I_{i}$ is a totally ordered set of cardinal $n_i^2-1$. Let $\{ \partial^{i}_{\alpha} \defeq \ad_{E^{i}_{\alpha}} \}_{\alpha \in I_{i}}$ be the induced basis of $\Der(M_{n_i})$. The dual basis is denoted by $\{ \theta_{i}^{\alpha} \}_{\alpha \in I_{i}}$.

We consider the metric $g$ on $\algA = \toplus_{i=1}^{r} M_{n_i}$ defined by $g( \partial^{i}_{\alpha}, \partial^{i'}_{\alpha'} ) = 0$ for $i \neq i'$ and $g^{i}_{\alpha \alpha'} \defeq g( \partial^{i}_{\alpha}, \partial^{i}_{\alpha'}) = \tr ( E^{i}_{\alpha} E^{i}_{\alpha'} )$ as in Sect.~\ref{sec matrix algebras}. Then, by construction, $\Der(\algA_i)$ is orthogonal to $\Der(\algA_{i'})$ when $i \neq i'$. 

A natural way to define a (noncommutative) integral of forms on $\algA$ is to decompose it along the $\algA_i$ as
\begin{align*}
\int_{\algA} \omega
& \defeq \tsum_{i=1}^{r} \int_{i} \omega_i
\end{align*}
for any $\omega = \toplus_{i=1}^{r} \omega_i \in \OmegaDer^\grast(\algA)$. Here $\int_{i} = \int_{\algA_i}$ is defined as in Sect.~\ref{sec matrix algebras} using the volume form $\omega_{\vol, i} \defeq \sqrt{\abs{g^i}} \theta_i^{\alpha^0_1} \wedge \cdots \wedge \theta_i^{\alpha^0_{n_i^2-1}}$ where $(\alpha^0_1, \dots, \alpha^0_{n^2-1}) \in I_i^{n_i^2-1}$ is such that $\alpha^0_1 < \cdots < \alpha^0_{n_i^2-1}$.

In order to compute $\int_{\algA} \omega$, one has to find the unique element $a = \toplus_{i=1}^{r} a_i \in \algA$ such that $\toplus_{i=1}^{r} a_i \omega_{\vol, i}$ captures the highest degrees in every $\OmegaDer^\grast(\algA_i)$, and then one has $\int_{\algA} \omega = \tsum_{i=1}^{r} \tr(a_i)$. In particular, with $\omega_{\vol}
 \defeq \toplus_{i=1}^{r} \omega_{\vol, i}$, one has $\int_{\algA} \omega_{\vol} = \tsum_{i=1}^{r} n_i$.

The metric $g$ on $\algA$ gives rise to a well defined Hodge $\hstar$-operator, which, according to the proof of Lemma~\ref{lemma *omega decomposed}, can be written using the Hodge $\hstar$-operators on each $\algA_i$ for the metric $g^i$. For any $\omega = \toplus_{i=1}^{r} \omega_i, \omega' = \toplus_{i=1}^{r} \omega'_i \in \OmegaDer^\grast(\algA)$, one has 
\begin{align*}
\omega \wedge \hstar \omega'
&= \tsum_{i=1}^{r} \omega_i \wedge \hstar_i \omega'_i
\end{align*}
where $\hstar_i$ is defined on $\OmegaDer^\grast(M_{n_i})$ as in Sect.~\ref{sec matrix algebras}, and we define the noncommutative scalar product of forms on $\algA = \toplus_{i=1}^{r} M_{n_i}$ by
\begin{align}
\label{eq def scalar product forms}
(\omega, \omega') \defeq
\int_\algA \omega \wedge \hstar \omega'
= \tsum_{i=1}^{r} \int_i \omega_i \wedge \hstar_i \omega'_i
\end{align}
This expression will be used to define action functional out of a connection $1$-form.

\section{Lifting one step of the defining inductive sequence}
\label{sec one step in the sequence}

In this section we study the lifting of an inclusion $\phi : \algA \to \algB$ regarding some of the structures defined on $\algA$ and $\algB$. Contrary to the previous section, here we consider the special case of sums of matrix algebras, $\algA = \toplus_{i=1}^{r} M_{n_i}$ and $\algB = \toplus_{j=1}^{s} M_{m_j}$. For reasons that will be explained in Sect.~\ref{sect direct limit NCGFT}, $\phi$ is not necessarily unital. We define the corresponding projection and injection maps $\pi_\algA^i$, $\pi_\algB^j$, $\iota^\algA_i$ and $\iota^\algB_j$.

The inclusion $\phi$ is taken in its simplest form, and we normalize it such that, for any $a = \toplus_{i=1}^{r} a_i$,
\begin{align*}
\phi^j (a) \defeq \pi_\algB^j \circ \phi (a)
&= 
\begin{pmatrix}
a_1 \otimes \bbbone_{\alpha_{j1}} & 0 & \cdots & 0 & 0\\
0 & a_2 \otimes \bbbone_{\alpha_{j2}} & \cdots & 0 & 0\\
\vdots & \vdots & \ddots & \vdots & \vdots\\
0 & 0 & \cdots & a_r \otimes \bbbone_{\alpha_{jr}} & 0 \\
0 & 0 & \cdots & 0 & \bbbzero_{n_0}
\end{pmatrix}
\end{align*}
where the integer $\alpha_{ji} \geq 0$ is the multiplicity of the inclusion of $M_{n_i}$ into $M_{m_j}$,  $\bbbzero_{n_0}$ is the $n_0 \times n_0$ zero matrix such that $n_0 \geq 0$ satisfies $m_j = n_0 + \tsum_{i=1}^{r} \alpha_{ji} n_i$, and 
\begin{align*}
a_i \otimes \bbbone_{\alpha_{ji}}
&= \left.
\begin{pmatrix}
a_i & 0 & 0 & 0 \\
0 & a_i & 0 & 0 \\
\vdots & \vdots & \ddots & \vdots \\
0 & 0 & \cdots & a_i \\
\end{pmatrix}
\right\} \text{$\alpha_{ji}$ times.}
\end{align*}
We define the maps $\phi_{i}^{j} \defeq \phi^j \circ \iota^\algA_i : M_{n_i} \to M_{m_j}$, which takes the explicit form
\begin{align*}
\phi_{i}^{j}(a_i)
&= \begin{pmatrix}
0  & \cdots & 0 & 0 & 0 & \cdots & 0 \\
\vdots  & \ddots & \vdots & \vdots  & \vdots & \cdots & \vdots \\
0  & \cdots & 0 & 0 & 0 & \cdots & 0 \\
0  & \cdots & 0 & a_i \otimes \bbbone_{\alpha_{ji}} & 0 & \cdots & 0 \\
0  & \cdots & 0 & 0 & 0 & \cdots & 0 \\
\vdots  & \vdots & \vdots & \vdots  & \vdots & \ddots & \vdots \\
0  & \cdots & 0 & 0 & 0 & \cdots & 0
\end{pmatrix}
\end{align*}
When $\alpha_{ji} > 0$, for $1 \leq \ell \leq \alpha_{ji}$ we define the maps $\phi_{i}^{j, \ell} : M_{n_i} \to M_{m_j}$ which insert $a_i$ at the $\ell$-th entry on the diagonal of $\bbbone_{\alpha_{ji}}$ in the previous expression, so that $a_i$ appears only once on the RHS. The maps $\phi^{j}$, $\phi_{i}^{j}$, and $\phi_{i}^{j, \ell}$ are morphisms of algebras and one has
\begin{align}
\phi &= \toplus_{j=1}^{s} \phi^j : \toplus_{i=1}^{r} M_{n_i} \to \toplus_{j=1}^{s} M_{m_j},
\nonumber
\\
\label{eq decompositions phi}
\phi^j &= \tsum_{i=1}^{r} \phi_{i}^{j} \circ \pi_\algA^i : \toplus_{i=1}^{r} M_{n_i} \to M_{m_j},
\\
\phi_{i}^{j} &= \tsum_{\ell=1}^{\alpha_{ji}} \phi_{i}^{j, \ell} : M_{n_i} \to M_{m_j}.
\nonumber
\end{align}
Notice then that $\tsum_{i=1}^{r} \tsum_{\ell=1}^{\alpha_{ji}} \phi_{i}^{j, \ell}(\bbbone_{\algA_i})$ fills the diagonal of $M_{m_j}$ with $\tsum_{i=1}^{r} \alpha_{ji} n_i$ copies of $1$ except for the last $n_0$ entries. When $n_0 = 0$, one gets $\tsum_{i=1}^{r} \tsum_{\ell=1}^{\alpha_{ji}} \phi_{i}^{j, \ell}(\bbbone_{\algA_i}) = \bbbone_{\algB_i}$.

We will make use of the following result.
\begin{lemma}
\label{lem product phiijell}
For any $a,b \in \algA$, any $i_1, i_2 \in \{1, \dots, r\}$, any $\ell_1 \in \{1, \dots, \alpha_{j i_1} \}$ and any $\ell_2 \in \{1, \dots, \alpha_{j i_2} \}$,
\begin{align}
\label{eq product phiijell}
\phi_{i_1}^{j, \ell_1}(a)\phi_{i_2}^{j, \ell_2}(b) = \delta_{i_1, i_2} \delta_{\ell_1, \ell_2} \phi_{i_1}^{j, \ell_1}(a b).
\end{align}
\end{lemma}

\begin{proof}
This is just a direct consequence of the definition of $\phi_{i}^{j, \ell}$ and the multiplications of block diagonal matrices.
\end{proof}

\medskip
\begin{definition}
An injective map $\phiMod : \modM \to \modN$ between a left $\algA$-module $\modM$ and a left $\algB$-module $\modN$ is $\phi$-compatible if $\phiMod(a e) = \phi(a) \phiMod(e)$ for any $a \in \algA$ and $e \in \modM$.
\end{definition}

In the following, since we want to construct a direct limit of modules accompanying a direct limit of algebras with injective maps, we always suppose that $\phiMod$ is also injective. As before we define the corresponding projection and injection maps $\pi_\modM^i$, $\pi_\modN^j$, $\iota^\modM_i$ and $\iota^\modN_j$.

\medskip
From \cite[Cor.~III.1.2]{Davi96a}, we know that all the left modules on $\algA$ and $\algB$ are of the form $\modM = \bbC^{n_1} \otimes \bbC^{\alpha_1} \toplus \cdots \toplus \bbC^{n_r} \otimes \bbC^{\alpha_r}$ and $\modN = \bbC^{m_1} \otimes \bbC^{\beta_1} \toplus \cdots \toplus \bbC^{m_s} \otimes \bbC^{\beta_s}$ for some integers $\alpha_i$ and $\beta_j$. Two situations are easily handled to construct an injective $\phi$-compatible map $\phiMod : \modM \to \modN$.
\begin{enumerate}
\item The case $\alpha_i = \beta_j = 1$ for any $i$ and $j$. In that situation, $\phiMod$ can be constructed in a natural way using the multiplicities $\alpha_{ji}$ of $\phi$. Denote by $\boldone_n \in \bbC^n$ (resp. $\boldzero_{n} \in \bbC^{n}$) the vector with all the entries equal to $1$ (resp. $0$). Then one can define $\phiMod$ by
\begin{align*}
\pi_\modN^j \circ \phiMod (e)
&=
\begin{pmatrix}
e_1 \otimes \boldone_{\alpha_{j1}} \\
e_2 \otimes \boldone_{\alpha_{j2}}\\
\vdots \\
e_r \otimes \boldone_{\alpha_{jr}}\\
\boldzero_{n_0}
\end{pmatrix}
\in \modN_j = \bbC^{m_j}
\end{align*}

\item The case $\alpha_i = n_i$ and $\beta_j = m_j$ for any $i$ and $j$. This situation corresponds to $\modM = \algA$ and $\modN = \algB$. The canonical map $\phiMod$ is taken to be $\phi$ itself.
\end{enumerate}
By construction, these two maps are $\phi$-compatible. In the more general situation, we have to inject $\alpha_{ji}$ times (as rows) $\bbC^{n_i} \otimes \bbC^{\alpha_i}$ into $\bbC^{m_j} \otimes \bbC^{\beta_j}$ ($\bbC^{n_i} \otimes \bbC^{\alpha_i}$ is injected only  when $\alpha_{ji}>0$). A necessary condition is that $\beta_j$ is large enough to accept the largest $\alpha_i$. This necessary condition leaves open the possibility of constructing many modules and many maps $\phiMod$ which are $\phi$-compatible.

Similarly to $\phi$, we decompose $\phiMod$ as $\phiMod[,i]^j \defeq \pi_\modN^j \circ \phiMod \circ \iota^\modM_i : \modM_i \to \modN_j$ and for any $1 \leq \ell \leq \alpha_{ji}$, $\phiMod[,i]^{j,\ell} : \modM_i \to \modN_j$ which insert $e_i \in \modM_i$ at the $\ell$-th row. 

\medskip
As expected, $\phi$ does not relate the centers of $\algA$ and $\algB$. This implies in particular that we can't expect to find or to construct a “general” map to inject $\Der(\algA)$ into $\Der(\algB)$ as modules over the centers, or, with less ambition, to inject a sub module and sub Lie algebra of $\Der(\algA)$ into a sub module and sub Lie algebra of $\Der(\algB)$. This strategy may indeed require very specific situations. 

Since it is convenient to consider all the derivations of $\algA$ and $\algB$, our approach is to keep track of the derivations in $\Der(\algB)$ which “come from” (to be defined below) derivations in $\Der(\algA)$. These derivations will propagate along the sequence of the direct limit, while new derivations will be introduced at each step of the limit. 

For any $i$, let us chose an \emph{orthogonal basis} $\{ \partial^{i}_{\!\!\algA, \alpha} \defeq \ad_{E^{i}_{\!\!\algA, \alpha}} \}_{\alpha \in I_{i}}$ of $\Der(\algA_i) = \Int(M_{n_i})$ where $E^{i}_{\!\!\algA, \alpha} \in \ksl_{n_i}$ and $I_{i}$ is a totally ordered set of cardinal $n_i^2-1$. For any $j$, we can introduce a basis of $\Der(\algB_j) = \Int(M_{m_j})$ in two steps. Let us define the set
\begin{align*}
J^\phi_{j} \defeq 
\left\{
(i, \ell, \alpha) \, / \, i \in \{ 1, \dots, r\}, \, \ell \in \{1, \dots, \alpha_{ji}\}, \, \alpha \in I_{i}  
\right\}
\end{align*}
and for any $\beta = (i, \ell, \alpha) \in J^\phi_{j}$, define
\begin{align*}
E^{j}_{\algB, \beta} \defeq \phi_{i}^{j, \ell}(E^{i}_{\!\!\algA, \alpha}) \in \ksl_{m_j}
\text{ and } \partial^{j}_{\algB, \beta} \defeq \ad_{E^{j}_{\algB, \beta}} \in \Der(\algB_j).
\end{align*}
The set $J^\phi_{j}$ is totally ordered for $\beta = (i,\ell,\alpha) < \beta' = (i',\ell',\alpha')$ iff $i<i'$ or [$i=i'$ and $\ell < \ell'$] or [$i=i'$ and $\ell=\ell'$ and $\alpha < \alpha' \in I_{i}$]. 

Denote by $g_{\algA}$ and $g_{\algB}$ the metrics on $\algA$ and $\algB$, defined as in Sect.~\ref{sec metric hodge}. We know that $g_{\algB}( \ad_{E^{j}_{\algB, \beta}}, \ad_{E^{j'}_{\algB, \beta'}} ) = 0$ for $j \neq j'$. So, let us consider a fixed value $j$. With the previous notations, one then has $g_{\algB}( \ad_{E^{j}_{\algB, (i,\ell,\alpha)}}, \ad_{E^{j}_{\algB, (i',\ell',\alpha')}} ) = 0$ when $i \neq i'$ or [$i = i'$ and $\ell \neq \ell'$] since then the product of matrices $E^{j}_{\algB, (i,\ell,\alpha)} E^{j}_{\algB, (i',\ell',\alpha')}$ is zero. Then one has 
\begin{align}
g^{j}_{\algB, \beta\beta'}
& \defeq
g_{\algB}( \ad_{E^{j}_{\algB, (i,\ell,\alpha)}}, \ad_{E^{j}_{\algB, (i',\ell',\alpha')}} ) 
= \tr( E^{j}_{\algB, (i,\ell,\alpha)} E^{j}_{\algB, (i',\ell',\alpha')} ) 
\nonumber
\\
&= \delta_{i i'} \delta_{\ell \ell'} \tr( E^{i}_{\!\!\algA, \alpha} E^{i}_{\!\!\algA, \alpha'} ) 
= \delta_{i i'} \delta_{\ell \ell'} g_{\algA}( \ad_{E^{i}_{\!\!\algA, \alpha}}, \ad_{E^{i}_{\!\!\algA, \alpha'}} )
\nonumber
\\
&= \delta_{i i'} \delta_{\ell \ell'} g^{i}_{\algA, \alpha\alpha'}
\label{eq gB and gA}
\end{align}
In the following, we will use the fact that the metric $( g^{j}_{\algB, \beta\beta'} )_{\beta, \beta' \in J^\phi_{j}}$ (and also its inverse $( g_{\algB, j}^{\beta\beta'} )_{\beta, \beta' \in J^\phi_{j}}$) is diagonal by blocks along the divisions induced by the choice of a couple $(i, \ell)$ for which $\beta = (i, \ell, \alpha)$. Notice also that if the $\partial^{i}_{\!\!\algA, \alpha}$'s are orthogonal (resp. orthonormal) for $g_{\algA}$, so are the $\partial^{j}_{\algB, \beta}$'s for the metric $g_{\algB}$ for any $\beta \in J^\phi_{j}$. This is the reason we chose to remove the “$\tfrac{1}{n}$” factor in front of the definition of the metrics (see Sect.~\ref{sec metric hodge}).

\medskip
We can now complete the family $\{ \partial^{j}_{\algB, \beta} \}_{\beta \in J^\phi_{j}}$ into a full basis of $\Der(\algB_j)$ with the same notation, $\beta \in J_{j} = J^\phi_{j}  \cup J^c_{j}$ where $J^c_{j}$ is a complementary set to get $\card(J_{j}) = m_j^2 -1$, in such a way that 
\begin{align}
\label{eq block orthogonality basis}
g_{\algB} ( \partial^{j}_{\algB, \beta}, \partial^{j}_{\algB, \beta'} ) = 0 
\text{ for any $\beta \in J^\phi_{j}$ and $\beta' \in J^c_{j}$.}
\end{align}
In other words, the metric $g_{\algB}$ is block diagonal and decomposes $\Der(\algB_j)$ into two orthogonal summands (See Sect.~\ref{sec construction block orthogonal basis} for an explicit way to construct such a basis adapted to $\phi$). We choose any total order on $J_{j}$ which extends the one on $J^\phi_{j}$. 

Notice that \eqref{eq block orthogonality basis} implies that the inverse of the matrix $( g^{j}_{\algB, \beta\beta'} )_{\beta,\beta' \in J_{j}}$, denoted by $( g_{\algB, j}^{\beta\beta'} )_{\beta,\beta' \in J_{j}}$, is also block diagonal  and is such that $( g_{\algB, j}^{\beta\beta'} )_{\beta,\beta' \in J^\phi_{j}}$ is the inverse of $( g^{j}_{\algB, \beta\beta'} )_{\beta,\beta' \in J^\phi_{j}}$ with $g_{\algB, j}^{(i,\ell,\alpha)(i',\ell',\alpha')} = \delta^{i i'} \delta^{\ell \ell'} g_{\algA, i}^{\alpha \alpha'}$.

\medskip
The derivations $\partial^{j}_{\algB, \beta}$ for $\beta \in J^\phi_{j}$ are the one “inherited” from the derivations on $\algA$. We will use the convenient notation $\partial^{j}_{\algB, \beta} = \phi_{i}^{j, \ell}(\partial^{i}_{\!\!\algA, \alpha})$ for $\beta = (i,\ell,\alpha)$. 

\begin{lemma}
\label{lem inherited derivations relations}
For any $1 \leq j \leq s$, $1 \leq i, i' \leq r$, $1 \leq \ell \leq \alpha_{ji}$, $1 \leq \ell' \leq \alpha_{ji'}$, $\alpha \in I_{i}$, $\alpha' \in I_{i'}$, $a_{i'} \in \algA_{i'}$, one has
\begin{align*}
\partial^{j}_{\algB, (i,\ell,\alpha)} \cdotaction \phi_{i'}^{j, \ell'}(a_{i'})
&=
\phi_{i}^{j, \ell}(\partial^{i}_{\!\!\algA, \alpha}) \cdotaction \phi_{i'}^{j, \ell'}(a_{i'})
= \delta_{i, i'} \delta_{\ell, \ell'} \phi_{i}^{j, \ell}(\partial^{i}_{\!\!\algA, \alpha} \cdotaction a_{i'})
\end{align*}
and
\begin{align*}
[\partial^{j}_{\algB, (i,\ell,\alpha)}, \partial^{j}_{\algB, (i',\ell',\alpha')}  ]
&= [ \phi_{i}^{j, \ell}(\partial^{i}_{\!\!\algA, \alpha}), \phi_{i'}^{j, \ell'}(\partial^{i'}_{\!\!\algA, \alpha'}) ]
= \delta_{i, i'} \delta_{\ell, \ell'} \phi_{i}^{j, \ell}( [\partial^{i}_{\!\!\algA, \alpha}, \partial^{i}_{\!\!\algA, \alpha'}])
\end{align*}
\end{lemma}

\begin{proof}
$\phi_{i}^{j, \ell}(\partial^{i}_{\!\!\algA, \alpha})$ is an inner derivation for the matrix $E^{j}_{\algB, (i,\ell,\alpha)}$ in which the only non zero part $E^{i}_{\!\!\algA, \alpha}$ is located on the diagonal of $M_{m_j}$ at a position depending on $i$ and $\ell$, see above. In the same way, the non zero part of $\phi_{i'}^{j, \ell'}(a_{i'})$ is $a_{i'}$ on the diagonal of $M_{m_j}$. When $i \neq i'$ or $\ell \neq \ell'$, the commutator of these two matrices is zero. When $i = i'$ and $\ell = \ell'$, the commutator is $[E^{i}_{\!\!\algA, \alpha}, a_i] = \partial^{i}_{\!\!\algA, \alpha} \cdotaction a_i$ on the diagonal of $M_{m_j}$ at the position designated by $i$ and $\ell$. This matrix is obviously $\phi_{i}^{j, \ell}(\partial^{i}_{\!\!\algA, \alpha} \cdotaction a_i)$. This proves the first relation.

The proof of the second relation relies on the same kind of argument since the $\partial^{i}_{\!\!\algA, \alpha}$ are inner derivations.
\end{proof}

\begin{definition}[$\phi$-compatible forms]
A form $\omega = \toplus_{i=1}^{r} \omega_i \in \OmegaDer^\grast(\algA)$ is $\phi$-compatible with a form $\eta = \toplus_{j=1}^{s} \eta_j \in \OmegaDer^\grast(\algB)$ if and only if for any $1 \leq i \leq r$, $1 \leq j \leq s$, $1 \leq \ell \leq \alpha_{ji}$, $\omega_i$ and $\eta_j$ have the same degree $p$ and for any $\partial^{i}_{\!\!\algA, \alpha_1}, \dots, \partial^{i}_{\!\!\algA, \alpha_p} \in \Der(\algA_i)$ ($\alpha_k \in I_{i}$) , one has
\begin{align}
\label{eq forms compatibility}
\phi_{i}^{j, \ell}\left( \omega_i( \partial^{i}_{\!\!\algA, \alpha_1}, \dots, \partial^{i}_{\!\!\algA, \alpha_p}) \right)
&=
\eta_j \left( \phi_{i}^{j, \ell}(\partial^{i}_{\!\!\algA, \alpha_1}), \dots,  \phi_{i}^{j, \ell}(\partial^{i}_{\!\!\algA, \alpha_p}) \right)
\end{align}
\end{definition}

Notice that the LHS of \eqref{eq forms compatibility} has only non zero values in block matrices on the diagonal of $M_{m_j}$ at a position which depends on $i$ and $\ell$ (see above). This implies that the RHS has the same structure.

\begin{proposition}
\label{prop compatible forms product differential}
Let $\omega = \toplus_{i=1}^{r} \omega_i, \omega' = \toplus_{i=1}^{r} \omega'_i \in \OmegaDer^\grast(\algA)$ and $\eta = \toplus_{j=1}^{s} \eta_j, \eta' = \toplus_{j=1}^{s} \eta'_j \in \OmegaDer^\grast(\algB)$ such that $\omega$ is $\phi$-compatible with $\eta$ and $\omega'$ is $\phi$-compatible with $\eta'$. Then $\omega \wedge \omega'$ is $\phi$-compatible with $\eta \wedge \eta'$ and $\dd \omega$ is $\phi$-compatible with $\dd \eta$.
\end{proposition}

\begin{proof}
Since the product of forms decompose along the indices $i$ and $j$, we fix these indices in the proof and we suppose that the degres of $\omega_i$ and $\omega'_i$ are $p$ and $p'$ respectively. Inserting the RHS of \eqref{eq forms compatibility} into \eqref{eq def form product} one has
\begin{align*}
(\eta_j \wedge \eta'_j) & \left( \phi_{i}^{j, \ell}(\partial^{i}_{\!\!\algA, \alpha_1}), \dots,  \phi_{i}^{j, \ell}(\partial^{i}_{\!\!\algA, \alpha_{p+p'}}) \right)
\\
&= \begin{multlined}[t]
\frac{1}{p!p'!} \!\! \sum_{\sigma\in \kS_{p+p'}} (-1)^{\abs{\sigma}} 
\eta_j \left( \! \phi_{i}^{j, \ell}(\partial^{i}_{\!\!\algA, \alpha_{\sigma(1)}}), \dots,  \phi_{i}^{j, \ell}(\partial^{i}_{\!\!\algA, \alpha_{\sigma(p)}}) \! \right)
\\
\eta'_j \left( \! \phi_{i}^{j, \ell}(\partial^{i}_{\!\!\algA, \alpha_{\sigma(p+1)}}), \dots,  \phi_{i}^{j, \ell}(\partial^{i}_{\!\!\algA, \alpha_{\sigma(p+p')}}) \! \right)
\end{multlined}
\\
&= \begin{multlined}[t]
\frac{1}{p!p'!} \sum_{\sigma\in \kS_{p+p'}} (-1)^{\abs{\sigma}} 
\phi_{i}^{j, \ell} \Big(
\omega_i ( \partial^{i}_{\!\!\algA, \alpha_{\sigma(1)}}, \dots, \partial^{i}_{\!\!\algA, \alpha_{\sigma(p)}} )
\\
\omega'_i ( \partial^{i}_{\!\!\algA, \alpha_{\sigma(p+1)}}, \dots, \partial^{i}_{\!\!\algA, \alpha_{\sigma(p+p')}} )
\Big)
\end{multlined}
\\
&= \phi_{i}^{j, \ell} \left( (\omega_i \wedge \omega'_i) ( \partial^{i}_{\!\!\algA, \alpha_1}, \dots, \partial^{i}_{\!\!\algA, \alpha_{p+p'}} ) \right)
\end{align*}
In the same way, inserting the RHS of \eqref{eq forms compatibility} into \eqref{eq def differential} one has
\begin{align*}
(\dd_j \eta_j) &\left( \phi_{i}^{j, \ell}(\partial^{i}_{\!\!\algA, \alpha_1}), \dots,  \phi_{i}^{j, \ell}(\partial^{i}_{\!\!\algA, \alpha_{p+1}}) \right)
\\
&= 
\begin{multlined}[t]
\sum_{k=1}^{p+1} (-1)^{k+1}
\phi_{i}^{j, \ell}(\partial^{i}_{\!\!\algA, \alpha_{k}}) \cdotaction \eta_j \left(
\phi_{i}^{j, \ell}(\partial^{i}_{\!\!\algA, \alpha_1}), \dots \omi{k} \dots, \phi_{i}^{j, \ell}(\partial^{i}_{\!\!\algA, \alpha_{p+1}})
\right)
\\
+ \sum_{1 \leq k < k' \leq p+1} (-1)^{k+k'}
\eta_j \Big(
\left[ \phi_{i}^{j, \ell}(\partial^{i}_{\!\!\algA, \alpha_{k}}), \phi_{i}^{j, \ell}(\partial^{i}_{\!\!\algA, \alpha_{k'}}) \right], 
\\
\dots \omi{i} \dots \omi{j} \dots,
\phi_{i}^{j, \ell}(\partial^{i}_{\!\!\algA, \alpha_{p+1}})
\Big)
\end{multlined}
\\
&= \begin{multlined}[t]
\sum_{k=1}^{p+1} (-1)^{k+1}
\phi_{i}^{j, \ell} \left(
\partial^{i}_{\!\!\algA, \alpha_{k}} \cdotaction \omega_i( \partial^{i}_{\!\!\algA, \alpha_1}, \dots \omi{k} \dots, \partial^{i}_{\!\!\algA, \alpha_{p+1}} )
\right)
\\
+ \sum_{1 \leq k < k' \leq p+1} (-1)^{k+k'}
\phi_{i}^{j, \ell} \left(
\omega_i( [\partial^{i}_{\!\!\algA, \alpha_{k}}, \partial^{i}_{\!\!\algA, \alpha_{k'}}], \dots \omi{k} \dots, \partial^{i}_{\!\!\algA, \alpha_{p+1}} )
\right)
\end{multlined}
\\
&= \phi_{i}^{j, \ell} \left(
(\dd_i \omega_i) ( \partial^{i}_{\!\!\algA, \alpha_1}, \dots, \partial^{i}_{\!\!\algA, \alpha_{p+1}} )
\right)
\end{align*}
\end{proof}

Let $\{ \theta^\alpha_{\!\!\algA, i} \}_{\alpha \in I_i}$ be the dual basis of $\{ \partial^{i}_{\!\!\algA, \alpha} \}_{\alpha \in I_{i}}$. Then one has $\theta^\alpha_{\!\!\algA, i} ( \partial^{i'}_{\!\!\algA, \alpha'} ) = \delta_{i}^{i'} \delta_{\alpha'}^{\alpha'}$. In the same way, denote by $\{ \theta^\beta_{\algB, j} \}_{\beta \in J_{j}}$ the dual basis of $\{ \partial^{j}_{\algB, \beta} \}_{\beta \in J_{j}}$.

\begin{remark}[$\phi$-compatibility and components of forms]
\label{ex one forms compatibility}
Let us first illustrate $\phi$-compatibility for $1$-forms. One has $\omega = \toplus_{i=1}^{r} \omega^{i}_{\alpha} \otimes \theta^\alpha_{\!\!\algA, i}$ for $\omega^{i}_{\alpha} \in \algA_i$ and $\eta = \toplus_{j=1}^{s} \eta^{j}_{\beta} \otimes \theta^\beta_{\algB, j}$ for $\eta^{j}_{\beta} \in \algB_j$. Then \eqref{eq forms compatibility} reduces to $\phi_{i}^{j, \ell} ( \omega^{i}_{\alpha} ) = \eta^{j}_{(i,\ell, \alpha)}$. This means that the components $\eta^{j}_{\beta}$ of $\eta$ in the “inherited directions” $\phi_{i}^{j, \ell}(\partial^{i}_{\!\!\algA, \alpha})$'s are inherited from $\omega$.

In the same way, for $p$-forms, the LHS of \eqref{eq forms compatibility} is non zero only for components along the $\beta$'s of the form $(i,\ell, \alpha)$ \emph{with the same couple} $(i, \ell)$, and these components are given by the RHS. So, all the components $\eta^{j}_{\beta_1 \dots \beta_p}$ in the “inherited directions” $\beta_1, \dots, \beta_p$ for $\beta_k = (i, \ell, \alpha_k)$ (same $i$ and $\ell$) are constrained by the $\phi$-compatibility condition.
\end{remark}

\medskip
Let $\nabla^{\modM}$ and $\nabla^{\modN}$ be two connections on a $\algA$-module $\modM$ and a $\algB$-module $\modN$ with an injective $\phi$-compatible map $\phiMod : \modM \to \modN$. We will used (as defined before) the maps $\phiMod[,i]^{j,\ell} : \modM_i \to \modN_j$. These connections define connections $\nabla^{\modM,i}$ on $\modM_i$ and $\nabla^{\modN,j}$ on $\modN_j$.

\begin{definition}
The two connections $\nabla^{\modM}$ and $\nabla^{\modN}$ are said to be $\phi$-compatible if and only if, for any $1 \leq i \leq r$, $1 \leq j \leq s$, $1 \leq \ell \leq \alpha_{ji}$, $\alpha \in I_{i}$, one has
\begin{align}
\label{eq def compatibility connections}
\phiMod[,i]^{j,\ell} \left(
\nabla^{\modM,i}_{\partial^{i}_{\!\!\algA, \alpha}} e_i
\right)
&= \nabla^{\modN,j}_{\phi_{i}^{j,\ell}(\partial^{i}_{\!\!\algA, \alpha})} \phiMod[,i]^{j,\ell} (e_i).
\end{align}
\end{definition}

When $\modM = \algA$ and $\modN = \algB$, one can introduce the connection $1$-forms $\omega_\modM$ and $\omega_\modN$ for $\nabla^{\modM}$ and $\nabla^{\modN}$. Then one has
\begin{lemma}
If $\omega_\modM$ and $\omega_\modN$ are $\phi$-compatible, then $\nabla^{\modM}$ and $\nabla^{\modN}$ are $\phi$-compatible. 
\end{lemma}

\begin{proof}
Here we have $\phiMod = \phi$. Using $\algA_i \ni e_i = e_i \bbbone_{\algA_i}$ and the definitions of the connection $1$-forms, one has $\nabla^{\modM,i}_{\partial^{i}_{\!\!\algA, \alpha}} e_i = (\partial^{i}_{\!\!\algA, \alpha} \cdotaction e_i) \bbbone_{\algA_i} + e_i \omega_{\modM,i}( \partial^{i}_{\!\!\algA, \alpha} )$, so that
\begin{align*}
\phi_{i}^{j,\ell} \left(
\nabla^{\modM,i}_{\partial^{i}_{\!\!\algA, \alpha}} e_i
\right)
&= \phi_{i}^{j,\ell} \left( (\partial^{i}_{\!\!\algA, \alpha} \cdotaction e_i) \bbbone_{\algA_i} \right)
+ \phi_{i}^{j,\ell} \left( e_i \omega_{\modM,i}( \partial^{i}_{\!\!\algA, \alpha} ) \right)
\\
&= \left( \phi_{i}^{j,\ell}(\partial^{i}_{\!\!\algA, \alpha}) \cdotaction \phi_{i}^{j,\ell}(e_i) \right) \phi_{i}^{j,\ell}( \bbbone_{\algA_i} )
+ \phi_{i}^{j,\ell}(e_i) 
\omega_{\modN,j} \left( \phi_{i}^{j,\ell}(\partial^{i}_{\!\!\algA, \alpha}) \right)
\\
&= \left( \phi_{i}^{j,\ell}(\partial^{i}_{\!\!\algA, \alpha}) \cdotaction \phi_{i}^{j,\ell}(e_i) \right) \phi_{i}^{j,\ell}( \bbbone_{\algA_i} )
+ \phi_{i}^{j,\ell}(e_i) 
\nabla^{\modN,j}_{\phi_{i}^{j,\ell}(\partial^{i}_{\!\!\algA, \alpha})} \bbbone_{\algB_j}
\\
&= \phi_{i}^{j,\ell}(\partial^{i}_{\!\!\algA, \alpha}) \cdotaction \phi_{i}^{j,\ell}(e_i)
+ \phi_{i}^{j,\ell}(e_i) \nabla^{\modN,j}_{\phi_{i}^{j,\ell}(\partial^{i}_{\!\!\algA, \alpha})} \bbbone_{\algB_j}
\\
&= \nabla^{\modN,j}_{\phi_{i}^{j,\ell}(\partial^{i}_{\!\!\algA, \alpha})} \phi_{i}^{j,\ell}(e_i)
\end{align*}
where we have used $\phi_{i}^{j,\ell}(e_i) = \phi_{i}^{j,\ell}(e_i) \phi_{i}^{j,\ell}( \bbbone_{\algA_i} )$ and $\phi_{i}^{j,\ell}(\partial^{i}_{\!\!\algA, \alpha}) \cdotaction \phi_{i}^{j,\ell}( \bbbone_{\algA_i} ) = \phi_{i}^{j,\ell}( \partial^{i}_{\!\!\algA, \alpha} \cdotaction \bbbone_{\algA_i} ) = 0$ so that, using the Leibniz rule, $\left( \phi_{i}^{j,\ell}(\partial^{i}_{\!\!\algA, \alpha}) \cdotaction \phi_{i}^{j,\ell}(e_i) \right) \phi_{i}^{j,\ell}( \bbbone_{\algA_i} ) = \phi_{i}^{j,\ell}(\partial^{i}_{\!\!\algA, \alpha}) \cdotaction \phi_{i}^{j,\ell}(e_i)$.
\end{proof}

\begin{remark}
Let us stress that the reverse of this Lemma is not true: $\phi$-compatibility between connections is weaker than $\phi$-compatibility between their connection $1$-forms. Indeed, let us assume that the two connections are $\phi$-compatible, that is that the first and last expressions are equal in the computation in the above proof. Then one can extract an equality at the second line between the connection $1$-forms: Take $e_i = \bbbone_{\algA_i}$ and since then the first terms in both sides are zero, one gets
\begin{align}
\label{eq weak form compatibility}
\phi_{i}^{j,\ell} \left( \omega_{\modM,i}( \partial^{i}_{\!\!\algA, \alpha} ) \right)
&= \phi_{i}^{j,\ell}(\bbbone_{\algA_i}) 
\omega_{\modN,j} \left( \phi_{i}^{j,\ell}(\partial^{i}_{\!\!\algA, \alpha}) \right)
\end{align}
This a weaker relation than \eqref{eq forms compatibility}. In Example~\ref{ex one forms compatibility}, we noticed that the $\phi$-compatibility \eqref{eq forms compatibility} between forms implies that all the values $\omega_{\modN,j} \left( \phi_{i}^{j,\ell}(\partial^{i}_{\!\!\algA, \alpha}) \right)$ come exactly from the values $\omega_{\modM,i}( \partial^{i}_{\!\!\algA, \alpha} )$. What \eqref{eq weak form compatibility} says is that $\phi_{i}^{j,\ell}(\bbbone_{\algA_i})$, as a projector, select only a part of the matrix $\omega_{\modN,j} \left( \phi_{i}^{j,\ell}(\partial^{i}_{\!\!\algA, \alpha}) \right)$ to be compared with the matrix $\omega_{\modM,i}( \partial^{i}_{\!\!\algA, \alpha} )$. So some parts of the matrix $\omega_{\modN,j} \left( \phi_{i}^{j,\ell}(\partial^{i}_{\!\!\algA, \alpha}) \right)$ may not be inherited. 

For practical reasons, we prefer to deal with the stronger $\phi$-compatibility condition, since it permits to “trace” (to “follow”) the degrees of freedom of the $\omega_{\modM,i}$'s “inside” the $\omega_{\modN,j}$'s. The weaker $\phi$-compatibility condition mix up these degrees of freedom into the matrices $\omega_{\modN,j} \left( \phi_{i}^{j,\ell}(\partial^{i}_{\!\!\algA, \alpha}) \right)$.
\end{remark}

\medskip
Let us consider the following specific situation for $\phi : \algA \to \algB$, where $\algA = \toplus_{i=1}^{r} M_{n_i}$ is embedded into $\algB = M_m$ with $m \geq \tsum_{i=1}^{r} n_i$ in such a way that each $M_{n_i}$ appears once and only once on the diagonal of $M_m$. We consider on $\algA$ and $\algB$ the integral defined in Sect.~\ref{sec metric hodge}.
\begin{proposition}
\label{prop eta start eta omega star omega}
Let $\omega = \toplus_{i=1}^{r} \omega_i \in \OmegaDer^\grast(\algA)$ and $\eta \in \OmegaDer^\grast(\algB)$ be such that $\eta$ is $\phi$-compatible with $\omega$ and $\eta$ vanishes on every derivation $\partial_{\algB, \beta}$ with $\beta \in J^c$ (here we omit the $j=1$ index). Then
\begin{align*}
\int_{\algB} \eta \wedge \hstar_\algB \eta
&=
\int_{\algA} \omega \wedge \hstar_\algA \omega
\end{align*}
\end{proposition}

\begin{proof}
Since $j=1$ and $\ell =1$ are the only possible values, to simplify the notations we write $\phi_i$ for $\phi_i^{j, \ell}$. We will  refer to $\phi_i(M_{n_i}) \subset M_m$ as the $i$-th block on the diagonal of $M_m$. The metric $g_\algB$ induces an orthogonal decomposition of the derivations of $\algB$ which is compatible with these blocks.  Two derivations $\ad_E$ and $\ad_{E'}$ such that $E \in \phi_i(M_{n_i})$ and $E'$ has no non vanishing entries in $\phi_i(M_{n_i})$ are orthogonal. So, $\Der(\algB) = \toplus_{i=1}^{r} \Der( \phi_i(M_{n_i}) ) \oplus \Der(\algB)^{c}$ where the decomposition is orthogonal for $g_\algB$ and $\Der(\algB)^{c}$ are derivations “outside” of the $\phi_i(M_{n_i})$'s.

For any $p$ one has $\eta( \phi_i(\partial^{i}_{\!\!\algA, \alpha_1}), \dots, \phi_i(\partial^{i}_{\!\!\algA, \alpha_{p}}) ) = \phi_i \left( \omega_i( \partial^{i}_{\!\!\algA, \alpha_1}, \dots, \partial^{i}_{\!\!\algA, \alpha_{p}} ) \right)$.\footnote{Since the $\omega_i$ and $\eta$ are not supposed to be homogeneous in $\OmegaDer^\grast(\algA_i)$ and $\OmegaDer^\grast(\algB)$, one has to evaluate them against any number of derivations.} This are the only possible values for $\eta$. In particular $\eta$ vanishes on any derivation in $\Der(\algB)^{c}$. Notice that $\phi_i \left( \omega_i( \partial^{i}_{\!\!\algA, \alpha_1}, \dots, \partial^{i}_{\!\!\algA, \alpha_{p}} ) \right)$ is $\omega_i( \partial^{i}_{\!\!\algA, \alpha_1}, \dots, \partial^{i}_{\!\!\algA, \alpha_{p}} ) \in M_{n_i}$ at the $i$-th block on the diagonal of $M_m$. We can then write $\eta = \tsum_{i=1}^{r} \eta_i$ where $\eta_i \in \phi_i(M_{n_i}) \otimes \exter^\grast \Der( \phi_i(M_{n_i}) )^* \subset M_m \otimes \exter^\grast \Der( \phi_i(M_{n_i}) )^*$. 

Let us define the metric $g_{\phi_i(M_{n_i})}$ as the restriction of $g_\algB$ to $\Der( \phi_i(M_{n_i}) )$. From this metric we construct its Hodge $\hstar$-operator $\hstar_{\phi_i(M_{n_i})}$ and its noncommutative integral $\int_{\phi_i(M_{n_i})}$ along the derivations of the matrix block $\phi_i(M_{n_i})$. Then we can apply Lemma~\ref{lemma *omega decomposed}, where the linear form $\tau$ is the ordinary trace on $M_m$, to get
\begin{align*}
\int_{\algB} \eta \wedge \hstar_\algB \eta
&= \tsum_{i=1}^{r} \int_{\phi_i(M_{n_i})} \eta_i \wedge \hstar_{\phi_i(M_{n_i})} \eta_i
\end{align*}
One has $\Der( \phi_i(M_{n_i}) ) \simeq \Der(M_{n_i})$ and, using \eqref{eq gB and gA}, $g_{\phi_i(M_{n_i})}$ becomes $g_{\algA, i}$ in this identification, so that  $\hstar_{\phi_i(M_{n_i})}$ corresponds to $\hstar_{M_{n_i}}$. In this identification, the form $\eta_i$ is then $\omega_i$, so that $\int_{\phi_i(M_{n_i})} \eta_i \wedge \hstar_{\phi_i(M_{n_i})} \eta_i = \int_{M_{n_i}} \omega_i \wedge \hstar_{M_{n_i}} \omega_i$. This shows that $\int_{\algB} \eta \wedge \hstar_\algB \eta = \tsum_{i=1}^{r} \int_{M_{n_i}} \omega_i \wedge \hstar_{M_{n_i}} \omega_i = \int_{\algA} \omega \wedge \hstar_\algA \omega$ by \eqref{eq def scalar product forms}.
\end{proof}

A slight adjustment of the proof of Prop.~\ref{prop eta start eta omega star omega} gives (with the same notations):
\begin{corollary}
Suppose that $\phi : \algA \to \algB$ includes $\alpha_i$ times $M_{n_i}$ on the diagonal of $M_m$. Then, with the same assumptions, one has
\begin{align*}
\int_{\algB} \eta \wedge \hstar_\algB \eta
&=
\tsum_{i=1}^{r} \alpha_{i} \int_{i} \omega_i \wedge \hstar_i \omega_i
\end{align*}
\end{corollary}

\section{Direct limit of NC gauge field theories}
\label{sect direct limit NCGFT}

\subsection{\texorpdfstring{$\phi$-compatibility of NC gauge field theories}{phi-compatibility of NC gauge field theories}}
\label{sec phi compatibility of NCGFT}

As mentioned in Sect.~\ref{sec AF algebras}, we will consider non unital $*$-homomorphisms $\phi : \algA = \toplus_{i=1}^{r} M_{n_i} \to \algB = \toplus_{j=1}^{s} M_{m_j}$. The reason for this choice is that we would like to cover physical situations where the gauge group are enlarged at each step of the defining inductive sequence $\{ (\algA_n, \phi_{n,m}) \, / \,  0 \leq n < m \}$. For instance, one may ask for the inclusion of $U(2)$ into $U(3)$, which can be performed in our framework by considering a natural inclusion $\phi : M_2 \to M_3$. This inclusion cannot be unital. A unital morphism would require for instance to consider the inclusion $\phi : M_1 \oplus M_2 \to M_3$, which may not correspond to a phenomenological requirement.

\smallskip
We first consider a NCGFT on the algebra $\algA = \toplus_{i=1}^{r} M_{n_i}$. Let us use the notations of Sect.~\ref{sec one step in the sequence}: for any $i$, let $\{ \partial^{i}_{\!\!\algA, \alpha} \defeq \ad_{E^{i}_{\!\!\algA, \alpha}} \}_{\alpha \in I_{i}}$ be an orthogonal basis  of $\Der(\algA_i) = \Int(M_{n_i})$ where $E^{i}_{\!\!\algA, \alpha} \in \ksl_{n_i}$ and $I_{i}$ is a totally ordered set of cardinal $n_i^2-1$, and let $\{ \theta^\alpha_{\!\!\algA, i} \}_{\alpha \in I_i}$ be the dual basis of $\{ \partial^{i}_{\!\!\algA, \alpha} \}_{\alpha \in I_{i}}$.

Since we are interested in the manipulation of connections as $1$-forms, we restrict our analysis to the left module $\modM = \algA$. From Example~\ref{example Mn} and the results in Sect.~\ref{sec modules and connections}, with obvious notations, a connection $1$-form can be written as $\omega =  \toplus_{i=1}^{r} \omega_i$ and its curvature $2$-form as $\Omega = \toplus_{i=1}^{r} \Omega_i$, with $\omega_i = \omega^{i}_{\alpha} \theta^\alpha_{\!\!\algA, i} = \mromega_{i} - B^{i}_{\!\!\algA, \alpha} \theta^\alpha_{\!\!\algA, i} = (E^{i}_{\!\!\algA, \alpha} - B^{i}_{\!\!\algA, \alpha}) \theta^\alpha_{\!\!\algA, i}$ and $\Omega_i = \tfrac{1}{2} \Omega^{i}_{\alpha_1\alpha_2} \theta^{\alpha_1}_{\!\!\algA, i} \wedge \theta^{\alpha_2}_{\!\!\algA, i}$ with $\Omega^{i}_{\alpha_1\alpha_2} = - ([B^{i}_{\!\!\algA, \alpha_1}, B^{i}_{\!\!\algA, \alpha_2}] - C(n_i)_{\alpha_1\alpha_2}^{\alpha_3} B^{i}_{\!\!\algA, \alpha_3})$ where $C(n_i)_{\alpha_1\alpha_2}^{\alpha_3}$ are the structure constants for the basis $\{ E^{i}_{\!\!\algA, \alpha} \}$ of $\ksl_{n_i}$.

The natural action for this NCGFT is then 
\begin{align*}
S &=
- \tsum_{i=1}^{r} \int_i \Omega_i \wedge \hstar_i \Omega_i
= - \tsum_{i=1}^{r} \tfrac{1}{2} \tr ( \Omega^{i}_{\alpha_1\alpha_2} \Omega^{i, \alpha_1\alpha_2})
\\
&= - \tsum_{i=1}^{r} \tfrac{1}{2} \tsum_{\alpha_1\alpha_2 \in I_i} \tr ( \Omega^{i}_{\alpha_1\alpha_2} )^2
\\
&= - \tsum_{i=1}^{r} \tfrac{1}{2} \tsum_{\alpha_1\alpha_2 \in I_i} \tr ( [B^{i}_{\!\!\algA, \alpha_1}, B^{i}_{\!\!\algA, \alpha_2}] - C(n_i)_{\alpha_1\alpha_2}^{\alpha_3} B^{i}_{\!\!\algA, \alpha_3} )^2
\end{align*}
where on the last line we have used the fact that the metric is diagonal.

As in Example~\ref{example C(M) otimes Mn}, one can also consider a NCGFT on the algebra $\halgA \defeq C^\infty(M) \otimes \algA = \toplus_{i=1}^{r} C^\infty(M) \otimes M_{n_i}$ for a manifold $M$. Then $\omega_i = A^{i}_{\!\!\algA,\mu} \dd x^\mu + (E^{i}_{\!\!\algA, \alpha} - B^{i}_{\!\!\algA, \alpha}) \theta^\alpha_{\!\!\algA, i}$ and $\Omega_i = \tfrac{1}{2} \Omega^{i}_{\mu_1\mu_2} \dd x^{\mu_1} \wedge \dd x^{\mu_2} + \Omega^{i}_{\mu \alpha} \dd x^\mu \wedge \theta^\alpha_{\!\!\algA, i} + \tfrac{1}{2}  \Omega^{i}_{\alpha_1\alpha_2} \theta^{\alpha_1}_{\!\!\algA, i} \wedge \theta^{\alpha_2}_{\!\!\algA, i}$ with
\begin{align*}
\Omega^{i}_{\mu_1\mu_2}
&= \partial_{\mu_1} A^{i}_{\!\!\algA, \mu_2} - \partial_{\mu_2} A^{i}_{\!\!\algA, \mu_1} - [A^{i}_{\!\!\algA, \mu_1}, A^{i}_{\!\!\algA, \mu_2}],
\\
\Omega^{i}_{\mu \alpha}
&= - ( \partial_\mu B^{i}_{\!\!\algA, \alpha} - [A^{i}_{\!\!\algA, \mu}, B^{i}_{\!\!\algA, \alpha}] ),
\\
\Omega^{i}_{\alpha_1\alpha_2}
&= - ( [B^{i}_{\!\!\algA, \alpha_1}, B^{i}_{\!\!\algA, \alpha_2}] - C(n_i)_{\alpha_1\alpha_2}^{\alpha_3} B^{i}_{\!\!\algA, \alpha_3} ).
\end{align*}
In that case, the natural action is
\begin{align*}
S &= \begin{multlined}[t]
- \tsum_{i=1}^r \int_M \Big(
\tfrac{1}{2} \tr( \Omega^{i}_{\mu_1\mu_2} \Omega^{i, \mu_1\mu_2} ) 
+ \tr( \Omega^{i}_{\mu \alpha} \Omega^{i, \mu \alpha} ) 
\\
+ \tfrac{1}{2} \tr( \Omega^{i}_{\alpha_1\alpha_2} \Omega^{i, \alpha_1\alpha_2} ) 
\Big) \sqrt{\abs{g_M}} \dd x
\end{multlined}
\end{align*}
where $g_M$ is a metric on $M$.

This action shares the same main feature as mentioned in Example~\ref{example C(M) otimes Mn}: it makes appear a SSBM thanks to the presence of the scalar fields $B^{i}_{\!\!\algA, \alpha} = B^{i, \alpha'}_{\!\!\algA, \alpha} E^{i}_{\!\!\algA, \alpha'} + i B^{i, 0}_{\!\!\algA, \alpha} \bbbone_{n_i}$ which can be non zero for a minimal configuration of the Higgs potential $- \tfrac{1}{2} \tsum_{i=1}^r \tr( \Omega^{i}_{\alpha_1\alpha_2} \Omega^{i, \alpha_1\alpha_2} )$. Then the couplings in $- \tsum_{i=1}^r \tr( \Omega^{i}_{\mu \alpha} \Omega^{i, \mu \alpha} )$ induce mass terms for the (gauge bosons) fields $A^{i}_{\!\!\algA,\mu} = A^{i, \alpha}_{\!\!\algA,\mu} E^{i}_{\!\!\algA, \alpha} + i A^{i, 0}_{\!\!\algA,\mu} \bbbone_{n_i}$. We will concentrate of this feature in the following.

\medskip
In order to simplify the analysis of the relation between NCGFT defined at each step of an inductive sequence of finite dimensional algebras $\{ (\algA_n, \phi_{n,m}) \, / \,  0 \leq n < m \}$, we will consider a unique inclusion $\phi : \algA \to \algB$ with $\algA = \toplus_{i=1}^{r} M_{n_i}$ and $\algB = \toplus_{j=1}^{s} M_{m_j}$ as in Sect.~\ref{sec one step in the sequence}.

Let $\eta = \toplus_{j=1}^{s} \eta_j$ be a connection $1$-form on $\algB$ for the module $\modN = \algB$. Denote by $\Theta = \toplus_{j=1}^{s} \Theta_j$ its curvature $2$-form. We use the notation $\eta_j = \eta^{j}_{\beta} \theta^\beta_{\!\algB, j} = (E^{j}_{\!\algB, \beta} - B^{j}_{\!\algB, \beta}) \theta^\beta_{\!\algB, j}$ and $\Theta_j = \tfrac{1}{2} \Theta^{j}_{\beta_1\beta_2} \theta^{\beta_1}_{\!\algB, j} \wedge \theta^{\beta_2}_{\!\algB, j}$ with $\Theta^{j}_{\beta_1\beta_2} = - ([B^{j}_{\!\algB, \beta_1}, B^{j}_{\!\algB, \beta_2}] - C(m_j)_{\beta_1\beta_2}^{\beta_3} B^{j}_{\!\algB, \beta_3})$ where $C(m_j)_{\beta_1\beta_2}^{\beta_3}$ are the structure constants for the basis $\{ E^{j}_{\!\algB, \beta} \}$ of $\ksl_{m_j}$.

We suppose that $\eta$ is $\phi$-compatible with $\omega$. This implies that, for all $\beta = (i, \ell, \alpha) \in  J^\phi_{j}$, $\eta^{j}_{\beta} = \phi_{i}^{j, \ell}(\omega^{i}_{\alpha})$ and that $B^{j}_{\!\algB, \beta} = \phi_{i}^{j, \ell}(B^{i}_{\!\!\algA, \alpha})$.

\begin{lemma}
\label{lem curvature inherited indices}
For any $\beta_1 = (i_1, \ell_1, \alpha_1), \beta_2 = (i_2, \ell_2, \alpha_2) \in J^\phi_{j}$, $\Theta^{j}_{\beta_1\beta_2} = 0$ for $i_1 \neq i_2$ or $\ell_1 \neq \ell_2$ and $\Theta^{j}_{\beta_1\beta_2} =\phi_{i}^{j, \ell}( \Omega^{i}_{\alpha_1\alpha_2})$ for $i = i_1 = i_2$ and $\ell = \ell_1 = \ell_2$.
\end{lemma}

\begin{proof}
One has to evaluate 
\begin{align*}
\Theta_{j} & (\phi_{i_1}^{j, \ell_1}(\partial^{i_1}_{\!\!\algA, \alpha_1}), \phi_{i_2}^{j, \ell_2}(\partial^{i_2}_{\!\!\algA, \alpha_2})) 
\\
&= \begin{multlined}[t]
\phi_{i_1}^{j, \ell_1}(\partial^{i_1}_{\!\!\algA, \alpha_1}) \cdotaction \eta_{j}(\phi_{i_2}^{j, \ell_2}(\partial^{i_2}_{\!\!\algA, \alpha_2}))
- \phi_{i_2}^{j, \ell_2}(\partial^{i_2}_{\!\!\algA, \alpha_2}) \cdotaction \eta_{j}(\phi_{i_1}^{j, \ell_1}(\partial^{i_1}_{\!\!\algA, \alpha_1}))
\\
- \eta_{j}( [\phi_{i_1}^{j, \ell_1}(\partial^{i_1}_{\!\!\algA, \alpha_1}), \phi_{i_2}^{j, \ell_2}(\partial^{i_2}_{\!\!\algA, \alpha_2})] )
- [ \eta_{j}( \phi_{i_1}^{j, \ell_1}(\partial^{i_1}_{\!\!\algA, \alpha_1}) ), \eta_{j}( \phi_{i_2}^{j, \ell_2}(\partial^{i_2}_{\!\!\algA, \alpha_2}) ) ]
\end{multlined}
\\
&= \begin{multlined}[t]
\phi_{i_1}^{j, \ell_1}(\partial^{i_1}_{\!\!\algA, \alpha_1}) \cdotaction \phi_{i_2}^{j, \ell_2}(\omega_{i_2}(\partial^{i_2}_{\!\!\algA, \alpha_2}))
- \phi_{i_2}^{j, \ell_2}(\partial^{i_2}_{\!\!\algA, \alpha_2}) \cdotaction \phi_{i_1}^{j, \ell_1}(\omega_{i_1}(\partial^{i_1}_{\!\!\algA, \alpha_1}))
\\
- \eta_{j}( [\phi_{i_1}^{j, \ell_1}(\partial^{i_1}_{\!\!\algA, \alpha_1}), \phi_{i_2}^{j, \ell_2}(\partial^{i_2}_{\!\!\algA, \alpha_2})] )
- [ \phi_{i_1}^{j, \ell_1}(\omega_{i_1}( \partial^{i_1}_{\!\!\algA, \alpha_1}) ), \phi_{i_2}^{j, \ell_2}(\omega_{i_2}( \partial^{i_2}_{\!\!\algA, \alpha_2}) ) ]
\end{multlined}
\end{align*}
For $i_1 \neq i_2$ or $\ell_1 \neq \ell_2$, from Lemma~\ref{lem inherited derivations relations}, all the terms in the first line vanish while the last commutator is zero since the two matrices involved do not occupy the same position on the diagonal of $M_{m_j}$. For $i_1 = i_2 = i$ and $\ell_1 = \ell_2 = \ell$, the expression reduces to $\phi_{i}^{j, \ell} ( \Omega_{i}( \partial^{i}_{\!\!\algA, \alpha_1}, \partial^{i}_{\!\!\algA, \alpha_2}) )$ (which is also a consequence of Prop.~\ref{prop compatible forms product differential}).

Notice that this result is also a direct consequence of the expression of $\Theta^{j}_{\beta_1\beta_2}$ in terms of the $B^{j}_{\!\algB, \beta}$'s.
\end{proof}

The curvature components $\Theta^{j}_{\beta_1\beta_2}$ for $\beta_1, \beta_2 \in J_{j}$ can be separated according to the 3 possibilities: (1) $(\beta_1, \beta_2)$ in $J^\phi_{j} \times J^\phi_{j}$; (2) $(\beta_1, \beta_2)$ or $(\beta_2, \beta_1)$ in $J^\phi_{j} \times J^c_{j}$; (3) $(\beta_1, \beta_2)$ in $J^c_{j} \times J^c_{j}$. From \eqref{eq block orthogonality basis}, the metric $g_\algB$ (and its inverse) is block diagonal for these subsets of indices. The natural action on $\algB$ can then be decomposed as
\begin{align*}
S &=
- \tsum_{j=1}^{s} \tfrac{1}{2} \tsum_{\beta_1\beta_2 \in J_{j}} \tr ( \Theta^{j}_{\beta_1\beta_2} \Theta^{j, \beta_1\beta_2} )
\\
&= \begin{multlined}[t]
- \tsum_{j=1}^{s} \Big(
\tfrac{1}{2} \tsum_{\beta_1\beta_2 \in J^\phi_{j}} \tr ( \Theta^{j}_{\beta_1\beta_2} \Theta^{j, \beta_1\beta_2} )
\\
+ \tsum_{\beta_1 \in J^\phi_{j}, \beta_2 \in J^c_{j}} \tr ( \Theta^{j}_{\beta_1\beta_2} \Theta^{j, \beta_1\beta_2} )
+ \tfrac{1}{2} \tsum_{\beta_1\beta_2 \in J^c_{j}} \tr ( \Theta^{j}_{\beta_1\beta_2} \Theta^{j, \beta_1\beta_2} )
\Big)
\end{multlined}
\end{align*}
For fixed $j$, let us consider the first summation on $J^\phi_{j}$. By Lemma~\ref{lem curvature inherited indices}, the two indices $\beta_1 = (i_1, \ell_1, \alpha_1), \beta_2 = (i_2, \ell_2, \alpha_2) \in J^\phi_{j}$ must satisfy $i_1 = i_2$ and $\ell_1 = \ell_2$ to get a non zero contribution, so that 
\begin{align*}
\tfrac{1}{2} \tsum_{\beta_1\beta_2 \in J^\phi_{j}} 
& \tr ( \Theta^{j}_{\beta_1\beta_2} \Theta^{j, \beta_1\beta_2} ) 
\\
& = 
\tfrac{1}{2} \tsum_{i=1}^{r} \tsum_{\ell=1}^{\alpha_{ji}} \tsum_{\alpha_1, \alpha_2 \in I_i} 
\tr ( \Theta^{j}_{(i, \ell, \alpha_1)(i, \ell, \alpha_2)} \Theta^{j, (i, \ell, \alpha_1)(i, \ell, \alpha_2)} ) 
\\
& =
\tfrac{1}{2} \tsum_{i=1}^{r} \tsum_{\ell=1}^{\alpha_{ji}} \tsum_{\alpha_1, \alpha_2 \in I_i} 
\tr ( \Omega^{i}_{\alpha_1\alpha_2} \Omega^{i, \alpha_1\alpha_2} )
\\
& =
\tfrac{1}{2} \tsum_{i=1}^{r} \alpha_{ji} \tsum_{\alpha_1, \alpha_2 \in I_i} 
\tr ( \Omega^{i}_{\alpha_1\alpha_2} \Omega^{i, \alpha_1\alpha_2} )
\end{align*}
where we have used \eqref{eq gB and gA} (which holds true also for the inverse metrics) to make an equivalence between raising the indices $\beta$ and raising the indices $\alpha$.

This relation tells us that the \emph{action on $\algB$ contains copies of terms from the action on $\algA$}. As expected, these terms involve the degrees of freedom which are inherited on $\algB$ from those on $\algA$. They involve also the multiplicities of the inclusions of $M_{n_i}$ into $M_{m_j}$. This implies that the relative weights of these terms are not the same on $\algB$ as they are on $\algA$.

\medskip
Let us now consider the algebra $\halgB \defeq C^\infty(M) \otimes \algB = \toplus_{j=1}^{s} C^\infty(M) \otimes M_{m_j}$. The connection $1$-form is parametrized as $\eta_j = A^{j}_{\algB,\mu} \dd x^\mu + (E^{j}_{\!\algB, \beta} - B^{j}_{\!\algB, \beta}) \theta^\beta_{\!\algB, j}$. We extend the $\phi$-compatibility between $\eta$ and $\omega$ on the geometrical part by the condition that $A^{j}_{\algB,\mu} = \tsum_{i=1}^{r} \phi_{i}^{j}(A^{i}_{\!\!\algA,\mu}) + A^{j, c}_{\algB,\mu}$ where $A^{j, c}_{\algB,\mu} \in M_{m_j}$ has zero entries in the image of $\phi^{j}$ (which is concentred as blocks on the diagonal). In other words, all the degrees of freedom in $A^{i}_{\!\!\algA,\mu}$ are copied into $A^{j}_{\algB,\mu}$ according to the map $\phi$.

From this decomposition, the components of the curvature $2$-forms $\Theta_j = \tfrac{1}{2} \Theta^{j}_{\mu_1\mu_2} \dd x^{\mu_1} \wedge \dd x^{\mu_2} + \Theta^{j}_{\mu \beta} \dd x^\mu \wedge \theta^\beta_{\!\algB, j} + \tfrac{1}{2}  \Theta^{j}_{\beta_1\beta_2} \theta^{\beta_1}_{\!\algB, j} \wedge \theta^{\beta_2}_{\!\algB, i}$ can be separated into inherited components from the curvature $2$-forms $\Omega_{i}$ on $\halgA$, interactions terms between inherited components of the $\omega_i$ $1$-forms with new components of the $\eta_j$ $1$-forms, and completely new terms from new components of the $\eta_j$. Explicitly, one has
\begin{align*}
\Theta^{j}_{\mu_1\mu_2}
&= \partial_{\mu_1} A^{j}_{\algB, \mu_2} - \partial_{\mu_2} A^{j}_{\algB, \mu_1} - [A^{j}_{\algB, \mu_1}, A^{j}_{\algB, \mu_2}]
\\
&= \begin{multlined}[t]
\tsum_{i=1}^{r} \phi_{i}^{j}\left(
\partial_{\mu_1} A^{i}_{\!\!\algA,\mu_2} - \partial_{\mu_2} A^{i}_{\!\!\algA,\mu_1} - [ A^{i}_{\!\!\algA,\mu_1}, A^{i}_{\!\!\algA,\mu_2}]
\right)
\\
- \tsum_{i=1}^{r} \left(
[ \phi_{i}^{j}( A^{i}_{\!\!\algA,\mu_1} ), A^{j, c}_{\algB,\mu_2} ]
+ [ A^{j, c}_{\algB,\mu_1}, \phi_{i}^{j}( A^{i}_{\!\!\algA,\mu_2} ) ]
\right)
\\
+ \partial_{\mu_1} A^{j, c}_{\algB,\mu_2} - \partial_{\mu_2} A^{j, c}_{\algB,\mu_1} - [ A^{j, c}_{\algB,\mu_1}, A^{j, c}_{\algB,\mu_2} ]
\end{multlined}
\end{align*}
In the same way, for any $\beta = (i, \ell, \alpha) \in J^\phi_{j}$, one has
\begin{align*}
\Theta^{j}_{\mu \beta}
&= - ( \partial_\mu B^{j}_{\!\algB, \beta} - [A^{j}_{\algB, \mu}, B^{j}_{\!\algB, \beta}] )
\\
&
= - \phi_{i}^{j, \ell}\left( \partial_\mu B^{i}_{\!\!\algA, \alpha} - [ A^{i}_{\!\!\algA,\mu} , B^{i}_{\!\!\algA, \alpha}]  \right)
+ [A^{j, c}_{\algB,\mu}, \phi_{i}^{j, \ell}(B^{i}_{\!\!\algA, \alpha})]
\end{align*}
where the last term is off diagonal. For $\beta \in J^c_{j}$, $\Theta^{j}_{\mu \beta}$ depends on the fields $A^{j, c}_{\algB,\mu}$ which couple with $B^{j}_{\!\algB, \beta}$. Finally, $\Theta^{j}_{\beta_1\beta_2}$ has been explored before.

\medskip
Denote by $\algA^\times$ and $\algB^\times$ the groups of invertible elements in $\algA$ and $\algB$ and let us define $\tphi : \algA^\times \to \algB^\times$, for any $a = \toplus_{i=1}^{r} a_i \in \algA^\times$, as
\begin{align*}
\tphi^j (a) \defeq \pi_\algB^j \circ \tphi (a)
&= 
\begin{pmatrix}
a_1 \otimes \bbbone_{\alpha_{j1}} & 0 & \cdots & 0 & 0\\
0 & a_2 \otimes \bbbone_{\alpha_{j2}} & \cdots & 0 & 0\\
\vdots & \vdots & \ddots & \vdots & \vdots\\
0 & 0 & \cdots & a_r \otimes \bbbone_{\alpha_{jr}} & 0 \\
0 & 0 & \cdots & 0 & \bbbone_{n_0}
\end{pmatrix}
\end{align*}
It is easy to check that $\tphi$ is a morphism of groups: for $a \in \algA^\times$, one has $\tphi(a) \in \algB^\times$ and $\tphi(a)^{-1} = \tphi(a^{-1})$. One has also $\tphi(a)^* = \tphi(a^*)$ so that if $u \in \calU(\algA)$ is a unitary element in $\algA$, so is $\tphi(u)$ in $\algB$.

\begin{lemma}
\label{lem hphi product phi}
For any $a \in \algA$ and any $u \in \algA^\times$, one has $\tphi(u)\phi(a) = \phi(ua)$ and $\phi(a)\tphi(u) = \phi(au)$. For any $j \in \{1, \dots, s\}$, any $i \in \{1, \dots, r\}$, and any $\ell \in \{1, \dots, \alpha_{ji} \}$, one has $\tphi(u)\phi_{i}^{j, \ell}(a) = \phi_{i}^{j, \ell}(ua)$ and $\phi_{i}^{j, \ell}(a)\tphi(u) = \phi_{i}^{j, \ell}(au)$.
\end{lemma}

\begin{proof}
From \eqref{eq decompositions phi}, it is sufficient to prove the last relations involving $\phi_{i}^{j, \ell}$. Since $\tphi(u)$ differs only from $\phi(u)$ at the last $n_0 \times n_0$ block entry on the diagonal where $\phi_{i}^{j, \ell}(a)$ is zero, one has $\tphi(u)\phi_{i}^{j, \ell}(a) = \phi(u)\phi_{i}^{j, \ell}(a) = \phi_{i}^{j, \ell}(ua)$ by \eqref{eq decompositions phi} and \eqref{eq product phiijell}. The same proof applies for the right multiplication by $u$.
\end{proof}

\begin{proposition}
\label{prop gauge transformations phi A and B}
Let $\omega$ be a connection $1$-form on $\algA$ and let $\eta$ be a $\phi$-compatible connection $1$-form on $\algB$. Let $u \in \calU(\algA)$ and $v \defeq \tphi(u) \in \calU(\algB)$. Then $\omega^u$ and $\eta^v$ are $\phi$-compatible.
\end{proposition}

\begin{proof}
Recall that $\omega^u = u^{-1} \omega u - u^{-1} (\dd u)$ and $\eta^v = v^{-1} \eta v - v^{-1} (\dd v)$. Notice that $\pi_\algA^i(\omega^u) = u_i^{-1} \omega_i u_i - u_i^{-1} (\dd u_i) \in \OmegaDer^1(\algA_i)$ and similarly $\pi_\algB^j(\eta^v) = v_j^{-1} \eta_j v_j - v_j^{-1} (\dd v_j) \in \OmegaDer^1(\algB_j)$. For any $j$ and any $\beta = (i, \ell, \alpha) \in J^\phi_{j}$, one has
\begin{align*}
\pi_\algB^j(\eta^v) ( \phi_{i}^{j, \ell}(\partial^{i}_{\!\!\algA, \alpha}) )
&=
 v_j^{-1} \eta_j( \phi_{i}^{j, \ell}(\partial^{i}_{\!\!\algA, \alpha}) ) v_j - v_j^{-1} [\phi_{i}^{j, \ell}(E^{i}_{\!\!\algA, \alpha}), v_j]
 \\
&= v_j^{-1} \phi_{i}^{j, \ell} \left( \omega_i(\partial^{i}_{\!\!\algA, \alpha}) \right) v_j  - v_j^{-1} \phi_{i}^{j, \ell}([E^{i}_{\!\!\algA, \alpha}, u_i])
 \\
& =
\phi_{i}^{j, \ell} \left( u_i^{-1} \omega_i(\partial^{i}_{\!\!\algA, \alpha}) u_i - u_i^{-1} (\dd u_i)(\partial^{i}_{\!\!\algA, \alpha}) \right)
= \phi_{i}^{j, \ell} \left( \pi_\algA^i(\omega^u)(\partial^{i}_{\!\!\algA, \alpha}) \right)
\end{align*}
where we have used Lemma~\ref{lem hphi product phi}.
\end{proof}

We have a similar result for connections on $\halgA$ and $\halgB$:
\begin{proposition}
\label{prop gauge transformations phi hA and hB}
Let $\omega$ be a connection $1$-form on $\halgA$ and let $\eta$ be a $\phi$-compatible (in the extended version) connection $1$-form on $\halgB$. Let $u \in \calU(\halgA)$ and $v \defeq \tphi(u) \in \calU(\halgB)$. Then $\omega^u$ and $\eta^v$ are $\phi$-compatible (in the extended version).
\end{proposition}

\begin{proof}
Concerning the algebraic parts of $\omega_i$ and $\eta_j$, the proof is the same as for Prop.~\ref{prop gauge transformations phi A and B}. It remains to show that the $A^{i, u_i}_{\!\!\algA,\mu} \defeq u_i^{-1} A^{i}_{\!\!\algA,\mu} u_i - u_i^{-1} \partial_\mu u_i$ are copied into $A^{j, v_j}_{\algB,\mu} \defeq v_j^{-1} A^{j}_{\algB,\mu} v_j - v_j^{-1} \partial_\mu v_j$ according to the map $\phi$. Using the fact that $v_j$ is block diagonal, and that these blocks are $\phi_{i}^{j, \ell}(u_i)$ or $\bbbone_{n_0}$, the diagonal part of the first term $v_j^{-1} A^{j}_{\algB,\mu} v_j$ is exactly $\tsum_{i=1}^{r} \phi_{i}^{j}(u_i^{-1} A^{i}_{\!\!\algA,\mu}  u_i)$. The second term $v_j^{-1} \partial_\mu v_j$ contains only blocks on the diagonal: the zero block from the block $\bbbone_{n_0}$ and blocks $\phi_{i}^{j, \ell}(u_i)^{-1} \partial_\mu \phi_{i}^{j, \ell}(u_i) = \phi_{i}^{j, \ell}(u_i^{-1} \partial_\mu u_i)$ otherwise. This proves that the blocks on the diagonal of $A^{j, v_j}_{\algB,\mu}$ are copies of the $A^{i, u_i}_{\!\!\algA,\mu}$ according to the map $\phi$. Obviously, the off diagonal part of $A^{j, v_j}_{\algB,\mu}$ mixes the degrees of freedom from the $A^{i}_{\!\!\algA,\mu}$'s and $u_i$'s.
\end{proof}

Prop.~\ref{prop gauge transformations phi A and B} and \ref{prop gauge transformations phi hA and hB} show that $\phi$-compatibility of connections is compatible with gauge transformations.

\medskip
We have now at hand all the technical ingredients to discuss NCGFT on the $AF$ $C^*$-algebra defined by a sequence $\{ (\algA_n, \phi_{n,m}) \, / \,  0 \leq n < m \}$. This NCGFT uses the derivation-based differential calculus constructed on the dense “smooth” subalgebra $\algA_\infty \defeq \cup_{n\geq 0} \algA_n$ as the inductive limit of the differential calculi $(\OmegaDer^\grast(\algA_n), \dd)$, and a natural module is the algebra $\algA_\infty$ itself. With obvious notations, the same holds for $\halgA_\infty$. All these constructions are canonical. A connection is constructed as a limit of connections on each $\algA_n$ (with module the algebra itself). If we insist the connection $1$-form on $\algA_{n+1}$ to be $\phi_{n,n+1}$-compatible with the connection $1$-form on $\algA_{n}$, then some degrees of freedom in this connection are inherited from those of the connections on $\algA_{n}$, and new degrees of freedom are added. This limiting procedure is compatible with a good notion of gauge transformations (see Prop.~\ref{prop gauge transformations phi A and B} and Prop.~\ref{prop gauge transformations phi hA and hB}). 

Concerning the dynamics, the terms in the action functional on $\algA_{n}$ can be found (with possible different weights) as terms in the action functional on $\algA_{n+1}$. If a solution for the gauge field degrees of freedom has been found on $\algA_{n}$, then these degrees of freedom appear as \emph{fixed fields} in the action on $\algA_{n+1}$, and so as constrains when one solves the field equations on $\algA_{n+1}$ for the new fields (non inherited degrees of freedom). The same applies to an inductive sequence $\{ (\halgA_n, \phi_{n,m}) \, / \,  0 \leq n < m \}$. The Lagrangian on $\algA_\infty$ (or $\halgA_\infty$) should be constructed as a limiting procedure by adding new terms at each step in order to take into account the new degrees of freedom. But then this Lagrangian could contain an infinite number of terms. From a physical point of view, we do not expect to reach that point: only some finite dimensional “approximations” (at some levels $n$) can be considered and tested in experiments. \emph{In other word, the purpose of our construction is not to define a “target” NCGFT (which could be quite singular) but to construct a direct sequence of finite dimensional NCFGT.} We expect all the empirical data to be encoded into this sequence which \emph{formally defines} a NCFGT on $\algA_\infty$ (or $\halgA_\infty$).

As already mentioned at the end of Sect.~\ref{sec AF algebras}, the $*$-homomorphisms $\phi_{n,n+1} : \algA_n \to \algA_{n+1}$ are only characterized up to unitary equivalence in $\algA_{n+1}$. We have shown that the action of such an unitary equivalence, which takes the form of an inner automorphism on $\algA_{n+1}$, is a transport of structures that does not change the physics. This is why it is convenient to work with the standard form used in this paper for these $*$-homomorphisms.

\subsection{Numerical exploration of the SSBM}
\label{sec numerical exploration of the SSBM}

We would like now to concentrate on the SSBM in our framework. Using previous notations, when the Higgs potential $- \tfrac{1}{2} \tsum_{i=1}^r \tr( \Omega^{i}_{\alpha_1\alpha_2} \Omega^{i, \alpha_1\alpha_2} )$ for $\halgA$ is minimized, the degrees of freedom in the $B^{i}_{\!\!\algA, \alpha} = B^{i, \alpha'}_{\!\!\algA, \alpha} E^{i}_{\!\!\algA, \alpha'} + i B^{i, 0}_{\!\!\algA, \alpha} \bbbone_{n_i}$ are fixed (possibly with a choice in many possible configurations) and the $\phi$-compatibility transports these values into the $B^{j}_{\!\algB, \beta} = B^{j, \beta'}_{\!\algB, \beta} E^{j}_{\!\algB, \beta'} + i B^{j, 0}_{\!\algB, \beta} \bbbone_{m_i}$'s. Then, using these fixed values, minimizing the Higgs potential for $\halgB$ only concern a subset of all the $B^{j, \beta'}_{\!\algB, \beta}$'s. The configuration they define is not necessarily the minimum of the Higgs potential for $\halgB$ if it were computed along all the $B^{j, \beta'}_{\!\algB, \beta}$'s. 

The configuration of the fields $B^{i, \alpha'}_{\!\!\algA, \alpha}$'s on $\halgA$ (resp. the fields $B^{j, \beta'}_{\!\algB, \beta}$'s on $\halgB$) induces a mass spectrum for the gauge fields $A^{i, \alpha}_{\!\!\algA,\mu}$'s (resp. the gauge fields $A^{j, \beta}_{\algB, \mu}$'s). In order to illustrate the way the masses of the $A^{j, \beta}_{\algB, \mu}$'s are related to the masses of the $A^{i, \alpha}_{\!\!\algA,\mu}$'s by the constrains induced by $\phi$, we have produced numerical computations of the mass spectra for simple situations $\phi : \algA \to \algB$. These computations have been performed using the software \textit{Mathematica} \cite{Mathematica}. The following situations have been considered:
\begin{enumerate}
\item $\algA = M_2$ and $\algB = M_3$. This is the minimal non trivial situation one can consider. It illustrates many features of the other situations concerning the masses of the fields $A^{j, \beta}_{\algB, \mu}$'s.

\item $\algA = M_2 \oplus M_2$ and $\algB = M_4$. This situation is used to illustrate how two different configurations for the fields $B^{i, \alpha'}_{\!\!\algA, \alpha}$'s (one for each $M_2$) can conflict to produce a rich typology for the masses of the fields $A^{j, \beta}_{\algB, \mu}$'s.

\item $\algA = M_2 \oplus M_2$ and $\algB = M_5$. This situation is used to show, by comparison with the preceding one, how the target algebra influences the mass spectrum.

\item $\algA = M_2 \oplus M_3$ and $\algB = M_5$. This situation is used to show, by comparison with the preceding one, how the source algebra influences the mass spectrum.
\end{enumerate}

Due to the large number of parameters involved in the mathematical expressions, the numerical computations cannot explore the full space of configurations for the fields $B^{i, \alpha'}_{\!\!\algA, \alpha}$'s. This is why we have chosen to work with a very simplified situation: for every $i$, the fields $B^{i, \alpha'}_{\!\!\algA, \alpha}$'s are parametrized by a single real parameter $\lambda_i$ which interpolates, on the interval $[0,1]$, between the null-configuration and the basis-configuration (see Example~\ref{example C(M) otimes Mn}) as $B^{i}_{\!\!\algA, \alpha} = \lambda_i E^{i}_{\!\!\algA, \alpha}$. For each value of the $\lambda_i$'s, the minimum of the Higgs potential on $\halgB$ along the fields $B^{j, \beta'}_{\!\algB, \beta}$'s which are not inherited (via $\phi$) is computed. Then this minimum configuration is inserted into the couplings with the fields $A^{j, \beta}_{\algB, \mu}$'s to compute the mass spectrum. A comparison with the mass spectrum of the fields $A^{i, \alpha}_{\!\!\algA,\mu}$'s is easily done, since this spectrum is fully degenerate: all these fields have the same mass $\lambda_i = \sqrt{2 n_i}$ according to Lemma~\ref{lem mass rep-config}.

Let us first make some general remarks on the expected results. When $\lambda_i = 0$ for all $i=1, \dots, r$ (in our examples $r=2$ at most), the configuration on $\halgA$ is the null-configuration. So, we expect the constraints exerted by the values of the fields $B^{i, \alpha'}_{\!\!\algA, \alpha} = 0$ when computing the minimum of the Higgs potential on $\halgB$ to produce the null-configuration for the fields $B^{j, \beta'}_{\!\algB, \beta}$'s. In the same way, when $\lambda_i = 1$ for all $i=1, \dots, r$, the configuration on $\halgA$ is the basis-configuration, and since $\phi$ preserves Lie brackets, the minimum of the Higgs potential on $\halgB$ is expected to be the basis-configuration for the fields $B^{j, \beta'}_{\!\algB, \beta}$'s. All the numerical computations presented below are consistent with these expected results.

When $\lambda_i$ is neither $0$ nor $1$, the configuration for the fields $B^{i, \alpha'}_{\!\!\algA, \alpha}$'s is not a minimum of the Higgs potential on $\halgA$ (with value $0$ in that case). Nevertheless, we consider these configuration as “possible” since $\algA$ is not necessarily the first algebra in the sequence $\{ (\algA_n, \phi_{n,m}) \, / \,  0 \leq n < m \}$. Indeed, as the results will show, and as already mentioned, the minimum of the Higgs potential on $\halgB$ is not the minimum along all the $B^{j, \beta'}_{\!\algB, \beta}$'s, and its value can be non zero. So, we are not reduced to considering only zero minima on $\halgA$ and it is legitimate to explore other configurations for the fields $B^{i, \alpha'}_{\!\!\algA, \alpha}$'s on $\halgA$.

Before describing the four cases, let us consider the situation in Fig.~\ref{fig-5Vd-5Md} which concerns the algebra $M_2$ only. The plot in Fig.~\ref{fig-5Vd} is the Higgs potential for the fields $B^{1}_{\!\!\algA, \alpha} = \lambda_1 E^{1}_{\!\!\algA, \alpha}$ depending only on $\lambda_1$. It is a quadratic polynomials in $\lambda_1$ and it looks very much like the Higgs potential of the SMPP in this approximation (reduction to a $1$-parameter dependency). The plot in Fig.~\ref{fig-5Md} is the mass spectrum for the $A^{1, \alpha}_{\!\!\algA,\mu}$ fields. As proved in Lemma~\ref{lem mass rep-config}, it is fully degenerated and depends linearly on $\lambda_1$ with slope $\sqrt{2 n_1}$ where $n_1 = 2$ in the present case. Similar plots can be obtained for any value $n_1$. These plots can be compared to the ones obtained in the four cases numerically explored.

\smallskip
All the numerical computations have been performed using orthonormal basis. We have noticed that the mass matrices, which have been computed in terms of an orthonormal basis $E^{1}_{\algB, \beta}$ (here we have only $j=1$) constructed as in Sect.~\ref{sec construction block orthogonal basis}, are almost diagonal, up to terms of order $10^{-6}$ (these small values could be considered as numerical artifacts). This motivates the introduction of the following nomenclature for the labels appearing in the plots of the mass spectra for the gauge bosons. These labels refer to directions defined for specific subsets of matrices $E^{1}_{\!\algB, \beta} \in M_{m_1} = \algB$. The labels $a^i$ will refer to the inherited directions in $\phi(M_{n_i}) \subset \algB$. Let $\widetilde{\phi(\algA)}$ be the smallest square matrix block in $M_{m_1}$ that contains all the $\phi(M_{n_i})$. The label $b$ (resp. $d$) will refer to non diagonal (resp. diagonal) directions in $\widetilde{\phi(\algA)}$ that are not labeled by the $a^i$. When $\widetilde{\phi(\algA)} \neq \algB$, the labels $c^i$ will refer to non diagonal directions in $\algB \backslash \widetilde{\phi(\algA)}$ that commute with all the $\phi(M_{n_{i'}})$ for $i' \neq i$: this means that these directions are matrices in $M_{m_1}$ with non zero entries only in $\algB \backslash \widetilde{\phi(\algA)}$ at the same rows and columns occupied by $\phi(M_{n_i})$. In the same situation, the label $e$ will refer to diagonal directions with non zero entries in $\algB \backslash \widetilde{\phi(\algA)}$ and which commute with $\widetilde{\phi(\algA)}$.  

\smallskip
For the first case $\algA = M_2$ and $\algB = M_3$, see Fig.~\ref{fig-1Vd-1Md}, there is only one real parameter $\lambda_1$. On the plots, this parameter is restricted to $\lambda_1 \in [-1,3]$.\footnote{A numerical exploration on the interval $[-100, 100]$ shows that the lines presented on Fig.~\ref{fig-1Md} are extended linearly.} On Fig.~\ref{fig-1Vd}, one sees that minimum values for the configurations of the $B^{1}_{\!\algB, \beta}$ fields in the Higgs potential, taking into account the fixed values of the inherited fields $B^{1}_{\!\!\algA, \alpha} = \lambda_1 E^{1}_{\!\!\algA, \alpha}$, are zero only at $\lambda_1 = 0, 1$. These values show a maximum near $\lambda_1 \simeq 0.563$ and grow rapidly outside of $\lambda_1 \in [0,1]$. This plot shows a similar global conformation as the one in Fig.~\ref{fig-5Vd}. 
On Fig.~\ref{fig-1Md}, the induced masses for the gauge bosons $A^{1}_{\algB, \mu}$ are presented. This mass spectrum is richer than the one in Fig.~\ref{fig-5Md}. It is not continuous, and one of its discontinuities coincides with the maximum of the values of the Higgs potential minima near $\lambda_1 \simeq 0.563$. The second discontinuity, near $\lambda_1 \simeq 2.376$, corresponds to a discontinuity in the Higgs potential minima plot that is visible at larger scale, as shown in the zoom effect circle. This mass spectrum is organized as follows: the $a^1$-lines have degeneracy $3$, the $c^1$-lines have degeneracy $4$, and the $e$-lines have degeneracy $1$, which amounts to the $8$ fields $A^{1, \beta}_{\algB, \mu}$ on $\halgB$. Notice that the $a^1_1$ and $a^1_3$ (resp. $c^1_1$ and $c^1_3$) straight lines are part of the same straight line (as shown by the dotted lines) with slope $2 = \sqrt{2 n_1}$ (resp. $\sqrt{3/2}$). Up to these small off-diagonal  values in the mass matrix, the $3$ fields $A^{1, \beta}_{\algB, \mu}$ belonging to the $a^1$-lines are the inherited $3$ fields $A^{1, \alpha}_{\!\!\algA,\mu}$ on $M_2$. The slope of the $a^1_1$ and $a^1_3$ lines shows that the inherited fields retain their masses when they are induced by $\phi : \algA \to \algB$. The $a^1_2$-line reveals that there is a slight breaking of this invariance for a specific range in $\lambda_1$. The $e$-lines correspond to the diagonal direction $\diag(1,1,-2)/\sqrt{6} \in M_3$.

\smallskip
For the next three cases, there are two parameters $\lambda_i$, $i=1,2$ and the plots explore the square  $(\lambda_1, \lambda_2) \in [0,1]^2$. Concerning the minimum values for the Higgs potential, all the points in the square can be displayed. But concerning the mass spectra, all the points in the square would give a cloud of points impossible to interpret. This is why we have chosen to display what happens along 7 specific lines in the square $[0,1]^2$, which are given in Fig.~\ref{fig-lambda-square}. Since the resulting plots in $3d$ may be quite difficult to read nevertheless, we have displayed two specific directions: the first one is the diagonal in the plane $(\lambda_1, \lambda_2)$, for $\lambda_1 = \lambda_2 \in [-1, 3]$; the second one is the anti-diagonal $\lambda_1 + \lambda_2 = 0.5$ in the square $[0,1]^2$. The diagonal plots can be directly compared to the first case, and they display a comparable rich structure: a restricted number of degenerated masses and some discontinuities. The anti-diagonal plots can be used to better understand how the inherited and new degrees of freedom behave in relation to each other (as encoded in the nomenclature for the labels). The choice of the parameter $0.5$ for the anti-diagonal line $\lambda_1 + \lambda_2 = 0.5$ is justified by the fact that we then explore in the 3 cases a region without discontinuity.

\smallskip
Let us consider the second case $\algA = M_2 \oplus M_2$ and $\algB = M_4$. The minimum values for the Higgs potential in Fig.~\ref{fig-2Vn} show a line of discontinuity that have a counterpart in the mass spectrum in Fig.~\ref{fig-2Mlc}. Exploring $\lambda_1 = \lambda_2 \in [-1, 3]$ in Fig.~\ref{fig-2Vd-2Md} shows that there are other discontinuities (at least one in the range considered) as in Fig.~\ref{fig-1Vd-1Md}. In the mass spectrum in Fig.~\ref{fig-2Md}, the $a^1$ and $a^2$-lines have degeneracy $3$, the $b$-lines have degeneracy $8$, and the $d$-lines have degeneracy $1$. For $i=1,2$, the slope of the $a^i_1$ and $a^i_3$ (resp. $b_1$ and $b_3$) straight lines is $2 = \sqrt{2 n_1} = \sqrt{2 n_2}$ (resp. $\sqrt{3}$). As in the previous case, modulo very small off-diagonal values in the mass matrix, the $6 = 2 \times 3$ fields in the $a^i$-lines are inherited from the $2 \times 3$ fields $A^{1, \alpha}_{\!\!\algA,\mu}$ and $A^{2, \alpha}_{\!\!\algA,\mu}$ from the two copies of $M_2$ and the $d$-lines correspond to the diagonal direction $\diag(1,1,-1,-1)/\sqrt{4} \in M_4$. The plot in Fig.~\ref{fig-2Mad_0.5} shows how the distribution of these gauge fields change along the anti-diagonal $\lambda_1 + \lambda_2 = 0.5$. The perfect symmetry around the diagonal at $\lambda_1 = 0.25$ in Fig.~\ref{fig-2Mad_0.5} shows that the two $M_2$ blocks play equal role, as expected. At $\lambda_1 = 0.5$, we end up on the side $\lambda_2 = 0$ in Fig.~\ref{fig-2Mlc}, where the top line (end of the $a^1$-line) has a slope $2 = \sqrt{2 n_1}$ with respect to $\lambda_1 \in [0,1]$; the middle line has a slope $\sqrt{3/2}$ (the $b$-line); and the lower line has a slope $0$ (it corresponds to the ends of the $a^2$-line and the $e$-line). Here again, at least in the region $\lambda_1 + \lambda_2 \leq 0.5$ (before the first discontinuity), we have checked numerically that the masses of the inherited fields are preserved by the map $\phi : \algA \to \algB$.

\smallskip
The third case $\algA = M_2 \oplus M_2$ and $\algB = M_5$, illustrated in Figs.~\ref{fig-3Vn-3Mlc}, \ref{fig-3Vd-3Md}, and \ref{fig-3Mad_0.5}, differs from the previous one by the greater number of new degrees of freedom in $\algB$. The discontinuity in Figs.~\ref{fig-3Vn} and \ref{fig-3Mlc} is larger. Its position has also moved, as can be seen also in Fig.~\ref{fig-3Md}. In this latter plot, the $a^1$ and $a^2$-lines have degeneracy $3$, the $b$ have degeneracy $8$, the $c^1$ and $c^2$-lines have degeneracy $4$, the $d$-lines have degeneracy $1$, and the $e$-lines have degeneracy $1$. Notice that the $d$ and $e$ lines are almost always merged in the plot, except for $d_2$ and $e_2$ which are close but clearly separated. For $i=1,2$, the slope of the $a^i_1$ and $a^i_3$ (resp. $b_1$ and $b_3$,  resp. $c^i_1$ and $c^i_3$) straight lines is $2 = \sqrt{2 n_1} = \sqrt{2 n_2}$ (resp. $\sqrt{3}$, resp. $\sqrt{3/2}$). Modulo very small off-diagonal values in the mass matrix, the $6= 2 \times 3$ fields in the $a^i$-lines are inherited from the $2 \times 3$ fields $A^{1, \alpha}_{\!\!\algA,\mu}$ and $A^{2, \alpha}_{\!\!\algA,\mu}$ from the two copies of $M_2$. In accordance with the nomenclature of the labels, the $8$ fields in the $b$-lines are new degrees of freedom along directions $E^{1}_{\!\algB, \beta}$ that are contained in $M_4 \subset M_5$, where $M_4$ contains the two copies of $M_2$. The $8= 2 \times 4$ fields in the $c^i$-lines are new degrees of freedom along directions $E^{1}_{\!\algB, \beta}$ that are defined with components outside of this $M_4 \subset M_5$: the $c^1$-line (resp. $c^2$-line) corresponds to fields in the directions $E^{1}_{\!\algB, \beta}$ with non zero entries outside of $M_4 \subset M_5$ and in the same rows and same columns as the ones in $M_{n_1}$ (resp. $M_{n_2}$). In other words, the $E^{1}_{\!\algB, \beta}$ for the $c^1$-line do not commute with $\phi(M_{n_1})$ while they commute with $\phi(M_{n_2})$, and vice versa for the $c^2$-line. The $d$-lines correspond to the diagonal direction $\diag(1,1,-1,-1,0)/\sqrt{4} \in M_5$ and the $e$-lines correspond to the diagonal direction $\diag(1,1,1,1,-4)/\sqrt{20} \in M_5$. The anti-diagonal plot in Fig.~\ref{fig-3Mad_0.5} brings us more information concerning the relationship between the $a^i$ and $c^i$-lines: it seems that there is a correlation between the $a^1$-line (resp. $a^2$-line) and the $c^1$-line (resp. $c^2$-line) due to the fact that their associated directions $E^{1}_{\!\algB, \beta}$ do not commute. This non commutativity could also explain the curved $b$-line which is “constrained” by the directions in the $a^1$ and $a^2$-lines.

\smallskip
Finally, the fourth case $\algA = M_2 \oplus M_3$ and $\algB = M_5$, illustrated in Figs.~\ref{fig-4Vn-4Mlc}, \ref{fig-4Vd-4Md}, and \ref{fig-4Mad_0.5}, is closer to the second case than to the third case. We conjecture that this is due to the fact that the diagonal in $\algB$ is filled by $\phi$ in the second and fourth cases, while there is a remaining $0$ in the third case (which permits the existence of the directions for the $c^1$ and $c^2$-lines in Fig.~\ref{fig-3Mad_0.5}). In the mass spectrum in Fig.~\ref{fig-4Md}, the $a^1$-lines have degeneracy $3$, the $a^2$-lines have degeneracy $8$, the $b$-lines have degeneracy $12$, and the $d$-lines have degeneracy $1$. The slope of the $a^1_1$ and $a^1_3$ (resp. $a^2_1$ and $a^2_3$,  resp. $b_1$ and $b_3$) straight lines is $2 = \sqrt{2 n_1}$ (resp. $\sqrt{6} = \sqrt{2 n_2}$, resp. $\sqrt{25/6}$). Modulo very small off-diagonal values in the mass matrix, the $3$ fields in the $a^1$-lines are inherited from the $3$ fields $A^{1, \alpha}_{\!\!\algA,\mu}$ from $M_{n_1} = M_2$ and the $8$ fields in the $a^2$-lines are inherited from the $8$ fields $A^{2, \alpha}_{\!\!\algA,\mu}$ from $M_{n_2} = M_3$. The $d$-lines correspond to the diagonal direction $\diag(1,1,1,1,-4)/\sqrt{20} \in M_5$. As showed in Fig.~\ref{fig-4Mad_0.5}, the mass spectrum along the anti-diagonal is no more symmetric, as can also be seen in Fig.~\ref{fig-4Mlc} (look for instance at the singular line in the mass spectrum): this distinguishes this case from the second one and illustrates how a change in the algebra $\algA$ affects the mass spectrum.

\smallskip
Let us make comments on these results. The exploration of the space of configurations for the fields $B^{i, \alpha'}_{\!\!\algA, \alpha}$'s along paths parametrized by the $\lambda_i$'s already shows a rich typology concerning the possible masses for the gauge bosons  $A^{j, \beta}_{\algB,\mu}$.

As seen in Fig.~\ref{fig-1Vd-1Md} for instance, the minimum for a conflictual situation $\lambda_1 = 1$ and $\lambda_2 = 0$ (conflict between the two minimal configurations for the $B^{1}_{\!\!\algA, \alpha}$ and $B^{2}_{\!\!\algA, \alpha}$ in $M_2$) is non zero and produces a global configuration for the fields $B^{1}_{\!\algB, \beta}$ that is neither the null-configuration nor the basis-configuration. The induced masses shows 3 possible values with degeneracies. Inserting this configuration as a initial data for another step into a sequence of NCGFT constructed on the sequence $\{ (\algA_n, \phi_{n,m}) \, / \,  0 \leq n < m \}$, may propagate this in-between result and produce more subtle configurations with richer possibilities for the masses of the gauge bosons.

Since the exploration of the space of configurations for the fields $B^{i, \alpha'}_{\!\!\algA, \alpha}$'s is reduced to paths parametrized by the $\lambda_i$'s, our results do not offer a general and systematic view of what could happen in our kind of models. Nevertheless, the results presented above already displays a rich phenomenology from which some information can be drawn.  The first noticeable feature is that the mass spectra, constrained by the $\phi$-compatibility, reveal that the masses are grouped in specific directions, so that we have neither a full degeneracy (as in Fig.~\ref{fig-5Md}) nor a complete list of independent masses (as many masses as degrees of freedom): these specific directions are grouped according to the inherited degrees of freedom (the $a^i$-lines), according to the way the new degrees of freedom commute or not with the inherited ones (the $b$ and $c^i$-lines), and according to the possible new diagonal degrees of freedom one can introduce (the $d$ and $e$-lines). Masses for inherited gauge bosons are preserved by the $\phi$-compatibility condition quite systematically near the null-configuration. Concerning the first discontinuity on the diagonal plots before the basis-configuration, the position of this discontinuity seems to be related, by an approximate linear relationship, to the ratio of the number of new degrees of freedom over the number of inherited degrees of freedom, see Table~\ref{table discontinuity ratio}. For the second discontinuity, a trend can be detected but without such a similar relationship. More advanced and time-consuming computations will be carried out as part of the thesis work of one of us (G.~N.) to further analyze the possible phenomenology of the models based on our approach. For instance, a computation will explore the behavior of mass matrices under successive embeddings $\phi_{n, n+1} : \algA_n \to \algA_{n+1}$, starting with different configurations. 

\newcommand{\centercell}[1]{\multicolumn{1}{c}{#1}}
\newcolumntype{d}[1]{D{.}{.}{#1}}
\begin{table}[t]
\centering
\begin{tabular}{rd{2.0}d{2.0}d{1.3}d{1.3}d{1.3}}
\toprule
\centercell{Case} & \centercell{$n_{\text{ndof}}$} & \centercell{$n_{\text{idof}}$} & \centercell{$r_{\text{dof}}$} & \centercell{$\lambda_{1, \text{first}}$} & \centercell{$\lambda_{1, \text{second}}$}\\
\midrule
$M_2 \oplus M_3 \to M_5$ & 13 & 11 & 1.182 & 0.475 & 2.526 \\
$M_2 \oplus M_2 \to M_4$ & 9   & 6    & 1.5      & 0.542 & 2.456 \\
$M_2 \to M_3$                         & 5   & 3    & 1.667 & 0.563 & 2.376 \\
$M_2 \oplus M_2 \to M_5$ & 18 & 6   & 3          & 0.734 & 2.263 \\
\bottomrule
\end{tabular}
\caption{Relationship between the positions of the first and second discontinuities and the ratio of the number of new degrees of freedom over the number of inherited degrees of freedom in the diagonal plots Figs.~\ref{fig-4Md}, \ref{fig-2Md}, \ref{fig-1Md}, and \ref{fig-3Md}: $n_{\text{ndof}}$ (resp. $n_{\text{idof}}$) is the number of new (resp. inherited) degrees of freedom, $r_{\text{dof}} = n_{\text{ndof}}/n_{\text{idof}}$ is the ratio of these degrees of freedom, $\lambda_{1, \text{first}}$ (resp. $\lambda_{1, \text{second}}$) is the value of $\lambda_1$ at the first (resp. second) discontinuity.}
\label{table discontinuity ratio}
\end{table}

It is out of the scope of this paper to elaborate on more realistic models and to try to analytically prove some of the mathematical conjectures that the numerical simulations suggest.

\section{Conclusion}

In this paper we have presented a mathematical framework in NCG to lift, in a very natural way, a defining sequence of an $AF$ $C^*$-algebra to a sequence of NCGFT of Yang-Mills-Higgs types. We have restricted the analysis to the derivation-based NCG, and we have exhibited and studied the essential ingredients which permit this lifting: derivation-based differential calculus, modules, connections, metrics and Hodge $\hstar$-operators, Lagrangians… In order to illustrate some possible applications in physics, we have numerically studied some specific toy models, with a focus concerning the typologies of the mass spectra one can obtained by the SSBM naturally present in these NCGFT of Yang-Mills-Higgs types.

A large part of the mathematical study presented here, in particular on the derivation-based NCG of an $AF$ $C^*$-algebra, could be used outside of the context of NCGFT. 

Since, in the literature, more realistic NCGFT have been constructed using spectral triples, we are interested to explore our new NCGFT approach on $AF$ $C^*$-algebras using spectral triples. This study is postponed to a forthcoming paper.

\appendix
\renewcommand\thesection{\Alph{section}}

\section{Appendices}
\label{sec appendices}

\subsection{\texorpdfstring{A technical result on Hodge $\hstar$-operators}{A technical result on Hodge *-operators}}
\label{sec technical result on Hodge * operators}

Let $V$ be a finite dimensional vector space, with $\dim V = n$ and let $V^*$ its dual. Denote by $\{ e_k \}_{k}$ and $\{ \theta^k \}_{k}$ a basis of $V$ and its dual basis in $V^*$. Let $g$ be a metric on $V$ and let $g_{k\ell} \defeq g(e_k, e_\ell)$. Denote by $(g^{k\ell})$ its inverse matrix and by $\abs{g}$ the determinant of $(g_{k\ell})$. 

Let $\algA$ be an associative unital algebra equipped with a linear form $\tau : \algA \to \bbC$. We define forms on $V$ with values in $\algA$ as elements in $\algA \otimes \exter^\grast V^*$. There is then a natural multiplication: for any $\omega \in \algA \otimes \exter^p V^*$ and $\eta \in \algA \otimes \exter^q V^*$, $\omega \wedge \eta \in \algA \otimes \exter^{p+q} V^*$ is defined by
\begin{align*}
(\omega \wedge \eta)&(e_1, \dots, e_{p+q})
\\
& \defeq
\frac{1}{p!q!} \sum_{\sigma\in \kS_{p+q}} (-1)^{\abs{\sigma}} \omega(e_{\sigma(1)}, \dots, e_{\sigma(p)}) \eta(e_{\sigma(p+1)}, \dots, e_{\sigma(p+q)})
\end{align*}
for any $e_I, \dots, e_{p+q} \in V$ where $\kS_{n}$ is the group of permutations of $n$ elements.

Given an orientation $\theta^1 \wedge \cdots \wedge \theta^n$ of the basis $\{ \theta^k \}_{k}$, the metric $g$ and the linear form $\tau$ define an “integration” $\int_V : \algA \otimes \exter^\grast V^* \to \bbC$ which is non zero only on $\algA \otimes \exter^n V^*$ where it is defined, for any $\omega$ uniquely written as $\omega = \sqrt{\abs{g}} a \otimes \theta^1 \wedge \cdots \wedge \theta^n$, by $\int_V \omega \defeq \tau(a)$. This definition does not depend on the basis $\{ \theta^k \}_{k}$ (only on its orientation up to a sign). The $n$-form $\omega_{\vol} \defeq \sqrt{\abs{g}} \bbbone \otimes \theta^1 \wedge \cdots \wedge \theta^n$ is called the volume form.

The metric $g$ defines also a Hodge $\hstar$-operator on $\algA \otimes \exter^\grast V^*$ defined on $\omega = \tfrac{1}{p!} \omega_{\ell_1, \dots, \ell_p} \otimes \theta^{\ell_1} \wedge \cdots \wedge \theta^{\ell_p}$ by the usual formula
\begin{align}
\hstar ( \tfrac{1}{p!} \omega_{\ell_1, \dots, \ell_p} & \otimes \theta^{\ell_1} \wedge \cdots \wedge \theta^{\ell_p} )
\nonumber
\\
&= \tfrac{1}{(n-p)!} \tfrac{1}{p!} \sqrt{\abs{g}} \omega_{\ell'_1, \dots, \ell'_p} g^{\ell_1 \ell'_1} \cdots g^{\ell_p \ell'_p} \epsilon_{\ell_1, \dots, \ell_n} \otimes \theta^{\ell_{p+1}} \wedge \cdots \wedge \theta^{\ell_n}
\label{eq hodge on V}
\end{align}
where $\epsilon_{\ell_1, \dots, \ell_n}$ is the completely antisymmetric tensor such that $\epsilon_{1, \dots, n} = 1$.  For any $\omega, \omega' \in \algA \otimes \exter^p V^*$, a standard computation gives (see for instance~\cite[Sect.~2.4]{Bert96a})
\begin{align}
\label{eq hodge star omega omega'}
\omega \wedge \hstar \omega'
&= \tfrac{1}{p!} \sqrt{\abs{g}} \omega_{\ell_1, \dots, \ell_p} \omega'^{\ell_1, \dots, \ell_p} \otimes \theta^{1} \wedge \cdots \wedge \theta^{n}
\end{align}
with $\omega'^{\ell_1, \dots, \ell_p} \defeq g^{\ell_1 \ell'_1} \cdots g^{\ell_p \ell'_p} \omega'_{\ell'_1, \dots, \ell'_p}$.

\begin{lemma}
\label{lemma *omega decomposed}
Suppose that $V = \toplus_{i=1}^{r} V_i$ is an orthogonal decomposition for $g$. Denote by $g_i$ the restriction of $g$ to $V_i$, denote by $\hstar_i$ the corresponding Hodge star operator on $\algA \otimes \exter^\grast V_i^*$, and denote by $\int_{V_i}$ the corresponding integration with volume form $\omega_{\vol, i}$ such that $\omega_{\vol} =  \omega_{\vol, 1} \wedge \cdots \wedge \omega_{\vol, r}$. 

Let $\omega_i, \omega'_i \in \algA \otimes \exter^{p_i} V_i^*$ and $\omega = \tsum_{i=1}^{r} \omega_i, \omega' = \tsum_{i=1}^{r} \omega'_i \in \algA \otimes \exter^\grast V^*$. Then
\begin{align*}
\int_{V} \omega \wedge \hstar \omega'
&=
\tsum_{i=1}^{r} \int_{V_i} \omega_i \wedge \hstar_i \omega'_i
\end{align*}
\end{lemma}

\begin{proof}
Let us introduce some notations. Let $n_i = \dim V_i$ (so that $n = \tsum_{i=1}^{r} n_i$); let $\{ e^i_k \}_{k=1, \dots, n_i}$ be a basis of $V_i$; for $\ell = 1, \dots, n$ written as $\ell = n_1 + \cdots + n_{i-1} + k$ with $k=1, \dots, n_i$, let $e_\ell \defeq e^i_k$ be the elements of a basis of $V$; let $\{ \theta_i ^k \}_{k=1, \dots, n_i}$ and $\{ \theta^\ell \}_{\ell=1, \dots, n}$ be the corresponding dual basis. Let $I_i$ be the set of indices $\ell = n_1 + \cdots + n_{i-1} + k$ with $k=1, \dots, n_i$, so that $e_\ell \in V_i$ for $\ell \in I_i$, and let $I_i^c$ be its complement in $\{1, \dots, n\}$.

The matrix $(g_{k\ell})$ is block diagonal, and so is its inverse $(g^{k\ell})$ with blocks $(g_i^{k\ell})$ and $\sqrt{\abs{g}} = \prod_{i=1}^{r} \sqrt{\abs{g_i}}$. The orientation of the basis $\{ \theta^\ell \}_{\ell=1, \dots, n}$ is chosen such that $\omega_{\vol} =  \omega_{\vol, 1} \wedge \cdots \wedge \omega_{\vol, r}$ with $\omega_{\vol, i} \defeq \sqrt{\abs{g_i}} \bbbone \otimes \theta_i^1 \wedge \cdots \wedge \theta_i^{n_i}$.

By linearity and the fact that $\omega_i \wedge \hstar \omega'_j = 0$ for $i \neq j$, one has $\omega \wedge \hstar \omega' = \tsum_{i=1}^{r} \omega_i \wedge \hstar \omega'_i$. For fixed $i$, to compute $\hstar \omega'_i$ we use \eqref{eq hodge on V}: in the RHS, only components of $(g^{k\ell})$ belonging to the block $(g_i^{k\ell})$ can appear, so that $\epsilon_{\ell_1, \dots, \ell_n}$ must contain all the indices $\ell \in I_i^c$. This implies that the RHS contains all the volume forms $\omega_{\vol, j} =\sqrt{\abs{g_j}} \bbbone \otimes \theta_j ^1 \wedge \cdots \wedge \theta_j ^{n_j}$ for $j \neq i$.

Let us consider a fixed $n$-uplet $(\ell_1, \dots, \ell_n)$ in the sum in the RHS of \eqref{eq hodge on V}. Our strategy is to collect all the $\ell_r \in I_i$ in front of the $n$-uplet. Since for $r=1, \dots, p_i$ one has $\ell_r \in I_i$, we need only to use a permutation of the remaining indices to collect the other $n_i - p_i$ indices which belong to $I_i$. As a permutation, we use the unique $(n_i - p_i, n - n_i)$-shuffle\footnote{A $(r,s)$-shuffle is a permutation $\sigma \in \kS_{r+s}$ such that $\sigma(1) < \sigma(2) < \cdots < \sigma(r)$ and $\sigma(r+1) < \sigma(r+2) < \cdots < \sigma(r+s)$. It is well-known that there are $\tfrac{(r+s)!}{r! s!}$ $(r,s)$-shuffles in $\kS_{r+s}$.} that maps $(\ell_{p_i + 1}, \dots, \ell_n)$ into $(\ell'_{p_i + 1}, \dots, \ell'_{n_i}, \ell''_{1}, \dots, \ell''_{n - n_i})$ such that $\ell'_r \in I_i$ for  $r = p_i + 1, \dots, n_i$ and $\ell''_{r'} \in I_i^c$ for $r' = n_i + 1, \dots, n$. This shuffle transforms the term (no summation) $\epsilon_{\ell_1, \dots, \ell_n} \theta^{\ell_{p_i+1}} \wedge \cdots \wedge \theta^{\ell_n}$ into $\epsilon_{\ell_1, \dots, \ell_{p_i}, \ell'_{p_i + 1}, \dots, \ell'_{n_i}, \ell''_{1}, \dots, \ell''_{n - n_i}} \theta^{\ell'_{p_i + 1}} \wedge \cdots \wedge \theta^{\ell'_{n_i}} \wedge \theta^{\ell''_{1}} \wedge \cdots \wedge \theta^{\ell''_{n - n_i}}$ (since the shuffle acts on both the $\epsilon$ indices and the $\theta^\ell$, there is no sign).

The summation on all the $n - p_i$-uplets $(\ell_{p_i + 1}, \dots, \ell_n)$ can now be performed in 3 steps: first using the $(n_i - p_i, n - n_i)$-shuffle which separates the indices belonging to $I_i$ and $I_i^c$ (there are $\tfrac{(n-p)!}{(n_1-p)! n_2!}$ such shuffles to use); then using a permutation on the $n - n_i$ indices $\ell''_{r'} \in I_i^c$ for $r' = n_i + 1, \dots, n$ to order them in increasing order (there are $(n - n_i)!$ such permutations) so that we can make appear the $\omega_{\vol, j}$ for $j \neq i$; finally managing the summation on the indices $\ell'_r \in I_i$ for  $r = p_i + 1, \dots, n_i$. This gives the series of equalities (with the previous notations and convention for the indices):
\begin{align*}
&\tfrac{1}{(n - p_i)!} 
\sqrt{\abs{g}}
\epsilon_{\ell_1, \dots, \ell_n} 
\theta^{\ell_{p_i + 1}} \wedge \cdots \wedge \theta^{\ell_n}
\\
&= \begin{multlined}[t]
\tfrac{(n - p_i)!}{(n - p_i)! (n_i - p_i)! (n - n_i)!} 
\sqrt{\abs{g}}
\epsilon_{\ell_1, \dots, \ell_{p_i}, \ell'_{p_i + 1}, \dots, \ell'_{n_i}, \ell''_{1}, \dots, \ell''_{n - n_i}} 
\\
\theta^{\ell'_{p_i + 1}} \wedge \cdots \wedge \theta^{\ell'_{n_i}} 
\wedge 
\theta^{\ell''_{1}} \wedge \cdots \wedge \theta^{\ell''_{n - n_i}}
\end{multlined}
\\
&= \begin{multlined}[t]
\tfrac{(n - n_i)!}{(n_i - p_i)! (n - n_i)!}
\sqrt{\abs{g}}
\epsilon_{\ell_1, \dots, \ell_{p_i}, \ell_{p_i + 1}, \dots, \ell_{n_i}, 1, \dots, n_1 + \cdots + n_{i-1}, n_1 + \cdots + n_{i} + 1, \dots, n}
\\
 \theta^{\ell_{p_i + 1}} \wedge \cdots \wedge \theta^{\ell_{n_i}}
 \wedge 
(\wedge_{\ell' \in I_i^c} \theta^{\ell'})
\end{multlined}
\\
&= \begin{multlined}[t]
\tfrac{1}{(n_i - p_i)!} 
\sqrt{\abs{g_i}}
\epsilon_{\ell_1, \dots, \ell_{p_i}, \ell_{p_i + 1}, \dots, \ell_{n_i}, 1, \dots, n_1 + \cdots + n_{i-1}, n_1 + \cdots + n_{i} + 1, \dots, n}
\\
 \theta^{\ell_{p_i + 1}} \wedge \cdots \wedge \theta^{\ell_{n_i}}
 \wedge 
(\wedge_{j=1, \dots, r ; j\neq i} \omega_{\vol, j})
\end{multlined}
\\
&= \begin{multlined}[t]
(-1)^{p_i(n_1 + \cdots + n_{i-1})}
\tfrac{1}{(n_i - p_i)!} 
\sqrt{\abs{g_i}}
\epsilon_{1, \dots, n_1 + \cdots + n_{i-1}, \ell_1, \dots, \ell_{n_i}, n_1 + \cdots + n_{i} + 1, \dots, n}
\\
\omega_{\vol, 1} \wedge \cdots \wedge \omega_{\vol, i-1}
\wedge
\theta^{\ell_{p_i + 1}} \wedge \cdots \wedge \theta^{\ell_{n_i}}
\wedge 
\omega_{\vol, i+1} \wedge \cdots \wedge \omega_{\vol, r}
\end{multlined}
\end{align*}
To deal with the summation on the indices $\ell_r \in I_i$, we shift their values by $-(n_1 + \cdots + n_{i-1})$ to get indices $k_r = 1, \dots, n_i$. Then we can replace the $\epsilon$ tensor by the tensor $\epsilon_{k_1, \dots, k_{n_i}}$ and at the same time replacing the $\theta^{\ell_r}$'s by the $\theta_i^{k_r}$'s. Using the factor $\tfrac{1}{(n_i - p_i)!} \sqrt{\abs{g_i}}$ in front, this makes appears $\hstar_i$:
\begin{align*}
\hstar \omega'_i
&=
(-1)^{p_i(n_1 + \cdots + n_{i-1})}
\omega_{\vol, 1} \wedge \cdots \wedge \omega_{\vol, i-1}
\wedge
(\hstar_i \omega'_i)
\wedge 
\omega_{\vol, i+1} \wedge \cdots \wedge \omega_{\vol, r}
\end{align*}
so that, with $\omega_i \wedge \hstar_i \omega'_i = a_i \omega_{\vol,i}$, where $a_i = \tfrac{1}{p_i !} \omega_{i, k_1, \dots, k_{p_i}} {\omega'_i}^{k_1, \dots, k_{p_i}}$,
\begin{align*}
\omega_i \wedge \hstar \omega'_i
&=
\omega_{\vol, 1} \wedge \cdots \wedge \omega_{\vol, i-1}
\wedge
(\omega_i \wedge \hstar_i \omega'_i)
\wedge 
\omega_{\vol, i+1} \wedge \cdots \wedge \omega_{\vol, r}
\\
&=
a_i
\omega_{\vol, 1} \wedge \cdots \wedge \omega_{\vol, r}
=
a_i \omega_{\vol}
\end{align*}
Notice that moving $\omega_i$ inside the product of the volume forms $\omega_{\vol, j}$ and then putting $a_i$ in front of this product is possible since volume forms have commutative values in $\algA$. Since $\int_{V_i} \omega_i \wedge \hstar_i \omega'_i = \tau(a_i)$, one gets $\int_{V} \omega \wedge \hstar \omega' = \tsum_{i=1}^{r} \int_{V} \omega_i \wedge \hstar \omega'_i = \tsum_{i=1}^{r} \tau(a_i) = \tsum_{i=1}^{r} \int_{V_i} \omega_i \wedge \hstar_i \omega'_i$.
\end{proof}

\subsection{Construction of a block orthogonal basis}
\label{sec construction block orthogonal basis}

Let $\{ \partial^{i}_{\!\!\algA, \alpha} \defeq \ad_{E^{i}_{\!\!\algA, \alpha}} \}_{\alpha \in I_{i}}$ be a basis in $\Der(\algA_i)$, for $i=1,\dots, r$. We have constructed in Sect.~\ref{sec one step in the sequence} a free family of elements $\partial^{j}_{\algB, \beta} \defeq \ad_{E^{j}_{\algB, \beta}}$ in $\Der(\algB_j)$, where $E^{j}_{\algB, \beta} \defeq \phi_{i}^{j, \ell}(E^{i}_{\!\!\algA, \alpha}) \in \ksl_{m_j}$ for $\beta \defeq (i, \ell, \alpha) \in J^\phi_{j}$ and we have shown that the derivations $\partial^{j}_{\algB, \beta}$ are orthogonal if the derivations $\partial^{i}_{\!\!\algA, \alpha}$ are orthogonal.

In this appendix, we will complete this free family with elements $\partial^{j}_{\algB, \beta} = \ad_{E^{j}_{\algB, \beta}}$ for $\beta \in J^c_{j}$ such that
\begin{align*}
g_{\algB} ( \partial^{j}_{\algB, \beta}, \partial^{j}_{\algB, \beta'} ) = 0 
\text{ for any $\beta \in J^\phi_{j}$ and $\beta' \in J^c_{j}$.}
\end{align*}
This will decompose $\Der(\algB_j)$ into two orthogonal summands. One knows that such a procedure is always possible, but for practical applications (for instance to construct gauge field theories, as in Sect.~\ref{sec numerical exploration of the SSBM}), such a concrete basis could be useful, in particular since the basis that we construct is adapted to the map $\phi : \algA \to \algB$.

Recall that the maps $\phi_{i}^{j, \ell}$ send the matrix algebras $M_{n_i}$, for $i=1,\dots, r$ and $\ell = 1, \dots, \alpha_{ji}$, on the diagonal of $M_{m_j}$, with a possible remaining block $\bbbzero_{n_0}$ on this diagonal. In order to manage this last block in the same way as the others, let us add the value $i=0$ to refer to this block, with $\alpha_{j0}=1$. In the following, we will use the notation $\phi_{i}^{j, \ell}(M_{n_i})$ with $i=0$ (and $\ell = 1$) to refer to this block.

Consider any matrix $E \in M_{m_j}$ which have only non zero entries outside the blocks $\phi_{i}^{j, \ell}(M_{n_i})$ for $i=0, \ldots, r$. Then a straightforward computation using block matrices shows that $\tr( E^{j}_{\algB, \beta} E ) = \tr( E E^{j}_{\algB, \beta} ) = 0$ for any $\beta \in J^\phi_{j}$ (the product in the trace is off diagonal). In the same way, for any $E$ in the block $\phi_{0}^{j, 1}(M_{n_0})$, one has $E^{j}_{\algB, \beta} E = E E^{j}_{\algB, \beta} = 0$.  This implies that $\ad_E$ for any $E$ outside of the blocks $\phi_{i}^{j, \ell}(M_{n_i})$ for $i=1, \ldots, r$ is orthogonal (for the metric induced by the trace) to $\ad_{E^{j}_{\algB, \beta}}$ for any $\beta \in J^\phi_{j}$. 

To complete the free family $\{ \ad_{E^{j}_{\algB, \beta}} \}_{\beta \in J^\phi_{j}}$ in $\Der(\algB_j)$, it is then sufficient to describe the elements $E \in \ksl_{m_j} \simeq \Der(\algB_j)$ “outside” of the blocks $\phi_{i}^{j, \ell}(M_{n_i})$ for $i=1, \ldots, r$. We will do that using the block decomposition induced by the $\phi_{i}^{j, \ell}(M_{n_i})$.

We can introduce a first family of matrices $E^{j}_{\algB, \beta}$ for new indices $\beta$ (in a set $J^c_{j}$) as (traceless) matrices with non zero entries in $\phi_{0}^{j, 1}(M_{n_0})$. There are $n_0^2 -1$ such elements.

Then, let us notice that for fixed $i=0, \dots,r$, the $\phi_{i}^{j, \ell}(M_{n_i})$ for $\ell =1, \dots, \alpha_{ji}$ are in diagonal blocks $\alpha_{ji} n_i \times \alpha_{ji} n_i$. We call these blocks the \emph{enveloping blocks} of $M_{n_i}$ inside $M_{m_j}$. We introduce a second family of matrices $E^{j}_{\algB, \beta}$ as matrices with non zero entries outside the enveloping blocks of the $M_{n_i}$'s. For every $i=0, \dots, r$, the row containing the enveloping block of $M_{n_i}$ contains $\tsum_{i,i', i\neq i'} (\alpha_{ji} n_i)(\alpha_{ji'} n_{i'})$ entries (non zero outside the enveloping block).

The next level of blocks embedding we consider is the one inside the enveloping block of $M_{n_i}$, for every $i$. In such a block, $\phi_{i}^{j, \ell}$ maps $M_{n_i}$ into the diagonal. For every $i = 0, \dots, r$, we can then introduce a family of matrices  $E^{j}_{\algB, \beta}$ with non zero entries inside the enveloping blocks of $M_{n_i}$  but outside the blocks $\phi_{i}^{j, \ell}(M_{n_i})$ (for $\ell = 1, \dots, \alpha_{ji}$). For a fixed $i$, there are $\alpha_{ji}(\alpha_{ji}-1)$ blocks of size $n_i \times n_i$, so that one can construct $\tsum_{i} \alpha_{ji}(\alpha_{ji}-1) (n_i)^2$ such matrices in the third family (as expected, for $i=0$ there is no contribution since $\alpha_{j0} = 1$).

Inside the enveloping block of $M_{n_i}$, we can also construct matrices with non zero entries in the blocks $\phi_{i}^{j, \ell}(M_{n_i})$. Indeed, for fixed $i$ and $\ell$, let us consider the matrix $E_i^\ell$ with $\bbbone_{n_i}$ in the block $\phi_{i}^{j, \ell}(M_{n_i})$. For $\beta = (i,\ell, \alpha)$ one has $E^{j}_{\algB, \beta} E_i^\ell = E^{j}_{\algB, \beta}$ and for $\beta' = (i,\ell', \alpha)$ with $\ell \neq \ell'$ one has $E^{j}_{\algB, \beta'} E_i^\ell = 0$. Notice that $E_i^\ell \notin \ksl_{m_j}$, but, for $\ell = 1, \dots, \alpha_{ji}-1$, the matrices $E_{i}^{\ell} - E_{i}^{\ell+1}$ belong to $\ksl_{m_j}$ and, by the previous remark, are orthogonal to the $E^{j}_{\algB, \beta}$ for $\beta \in J^\phi_{j}$. These matrices constitute the fouth family: there are $\tsum_{i} (\alpha_{ji}-1)$ such matrices (once again, there is not contribution for $i=0$) 

The fifth and last family of matrices are constructed as $n_{i+1} E_i^1 - n_i E_{i+1}^1 \in \ksl_{m_j}$ for $i=0, \dots, r-1$. These $r$ matrices have entries in different enveloping blocks.

Let us collect the number of matrices $E^{j}_{\algB, \beta}$ for $\beta \in J^c_{j}$ that we have constructed:
\begin{align*}
\card(J^c_{j})
&= \begin{multlined}[t]
n_0^2 - 1 
+ \tsum_{i \geq 0} \tsum_{i' \geq 0 ; i\neq i'} (\alpha_{ji} n_i)(\alpha_{ji'} n_{i'}) 
\\
+ \tsum_{i \geq 0} \alpha_{ji}(\alpha_{ji}-1) (n_i)^2
+ \tsum_{i \geq 0} (\alpha_{ji}-1)
+ r
\end{multlined}
\\
&= \begin{multlined}[t]
n_0^2
+ 2 \tsum_{i \geq 1} (\alpha_{ji} n_i) n_0 
+ \tsum_{i,i' \geq 1} (\alpha_{ji} n_i)(\alpha_{ji'} n_{i'})
\\
- \tsum_{i \geq 1} \alpha_{ji} [ (n_i)^2 - 1 ]
- 1
\end{multlined}
\\
&= (m_j)^2 - 1 - \card(J^\phi_{j})
\end{align*}
where we have used $m_j = n_0 + \tsum_{i \geq 1} \alpha_{ji} n_i$ and $\card(J^\phi_{j}) = \tsum_{i \geq 1} \alpha_{ji} [ (n_i)^2 - 1 ]$. This shows that $\card(J_{j}) = (m_j)^2 - 1 = \dim \ksl_{m_j}$ for $J_{j} \defeq J^\phi_{j} \cup J^c_{j}$ and that the free family $\{ E^{j}_{\algB, \beta} \}_{\beta \in J_{j}}$ is a basis for $\ksl_{m_j}$. This construction and this computation can be adapted to the case $n_0 = 0$.

Notice that the derivations $\partial^{j}_{\algB, \beta} = \ad_{E^{j}_{\algB, \beta}}$ for $\beta \in J^c_{j}$ are not necessarily orthogonal. One can apply the Gram-Schmidt process to transform this basis into an orthonormal one.

\bibliography{bibliography}

\clearpage

\begin{figure}[t]
\centering
\subfloat[Minimum values for the Higgs potential with details in insert.]{%
\includegraphics[width=0.48\linewidth]{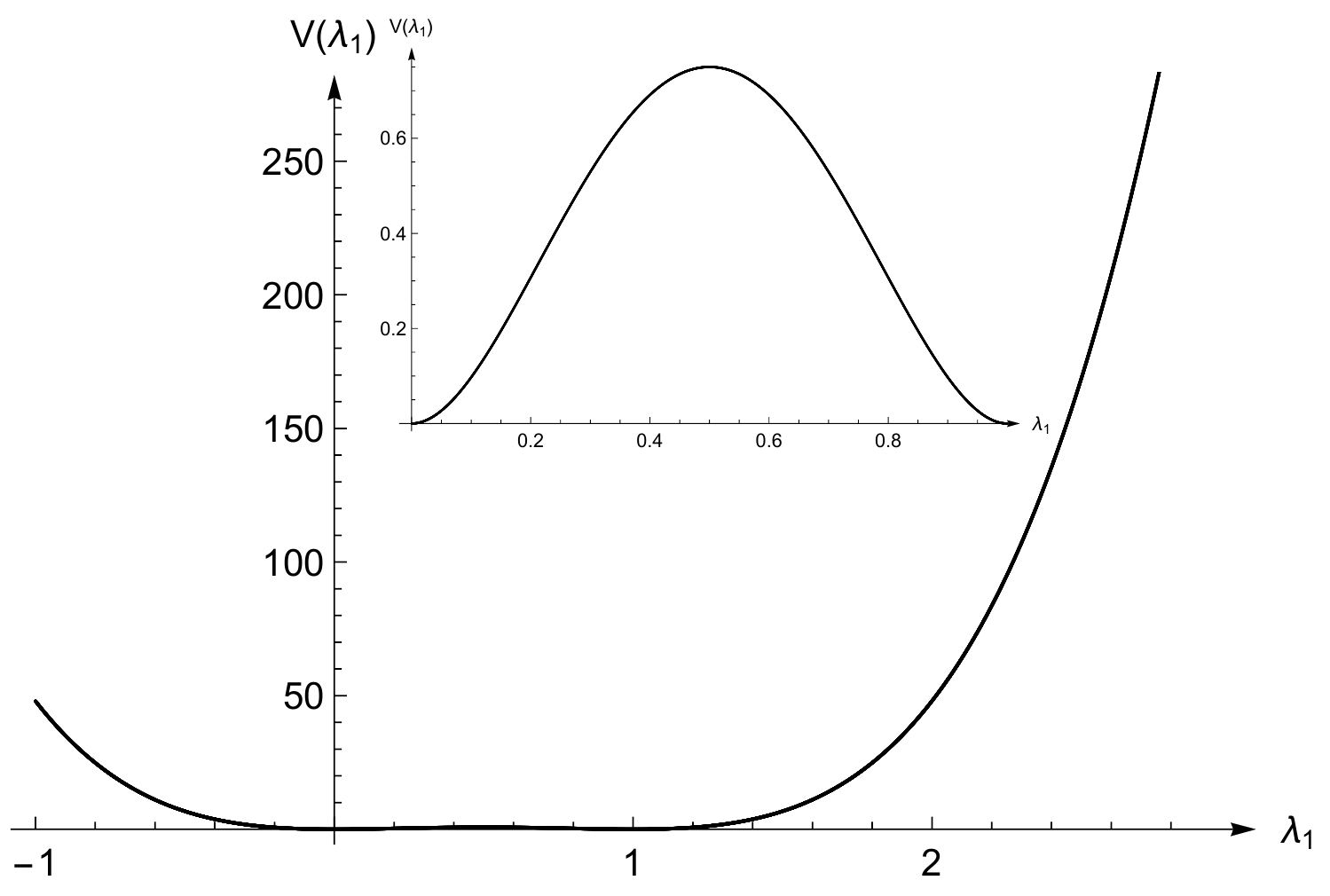}
\label{fig-5Vd}
}
\subfloat[Masses for the gauge fields.]{%
\includegraphics[width=0.48\linewidth]{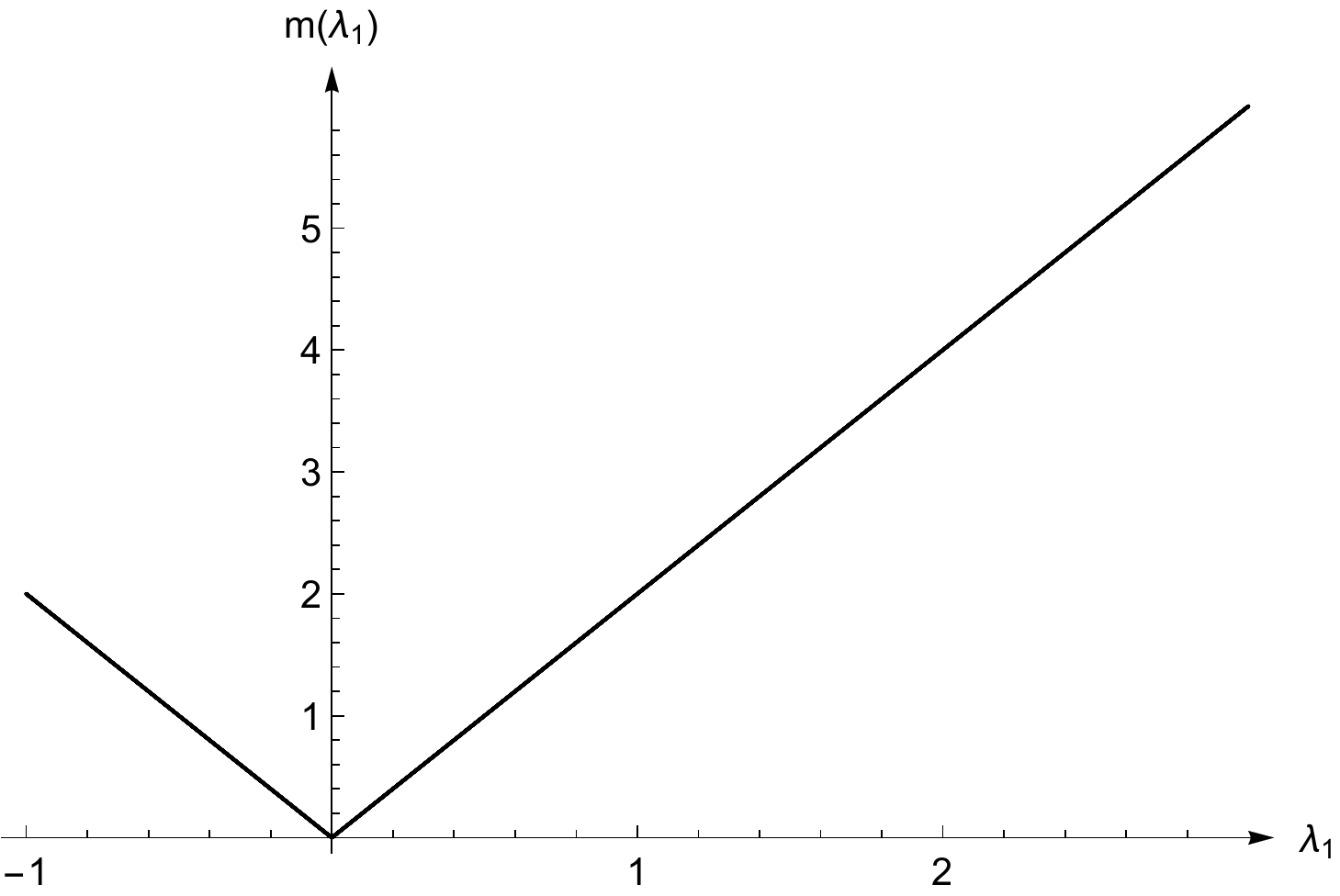}
\label{fig-5Md}
}
\caption{$\algA = M_2$: plots for $\lambda_1 \in [-1, 3]$.}
\label{fig-5Vd-5Md}
\end{figure}

\begin{figure}[t]
\centering
\subfloat[Minimum values for the Higgs potential with details in insert.]{%
\includegraphics[width=0.48\linewidth]{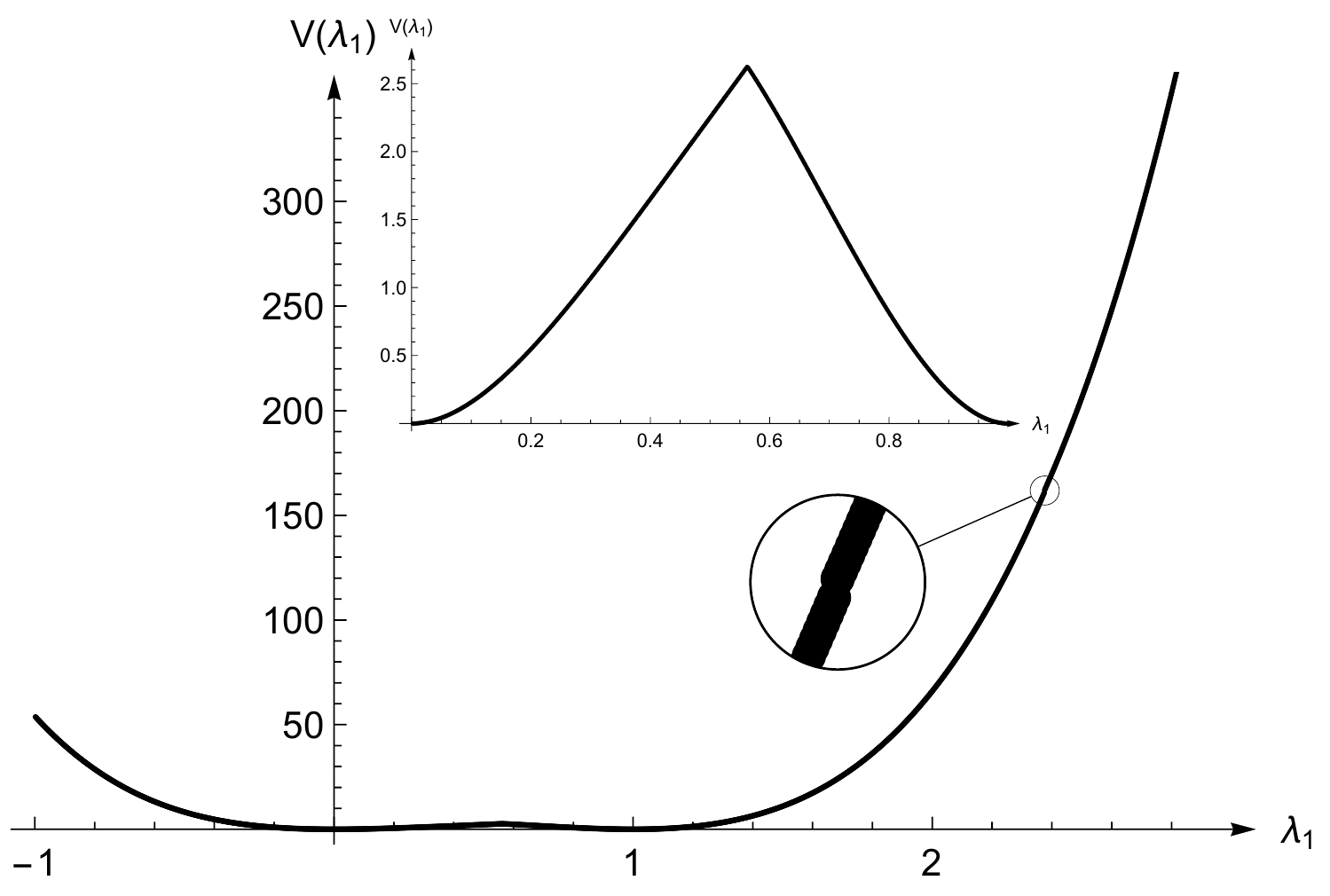}
\label{fig-1Vd}
}
\subfloat[Masses for the gauge fields.]{%
\includegraphics[width=0.48\linewidth]{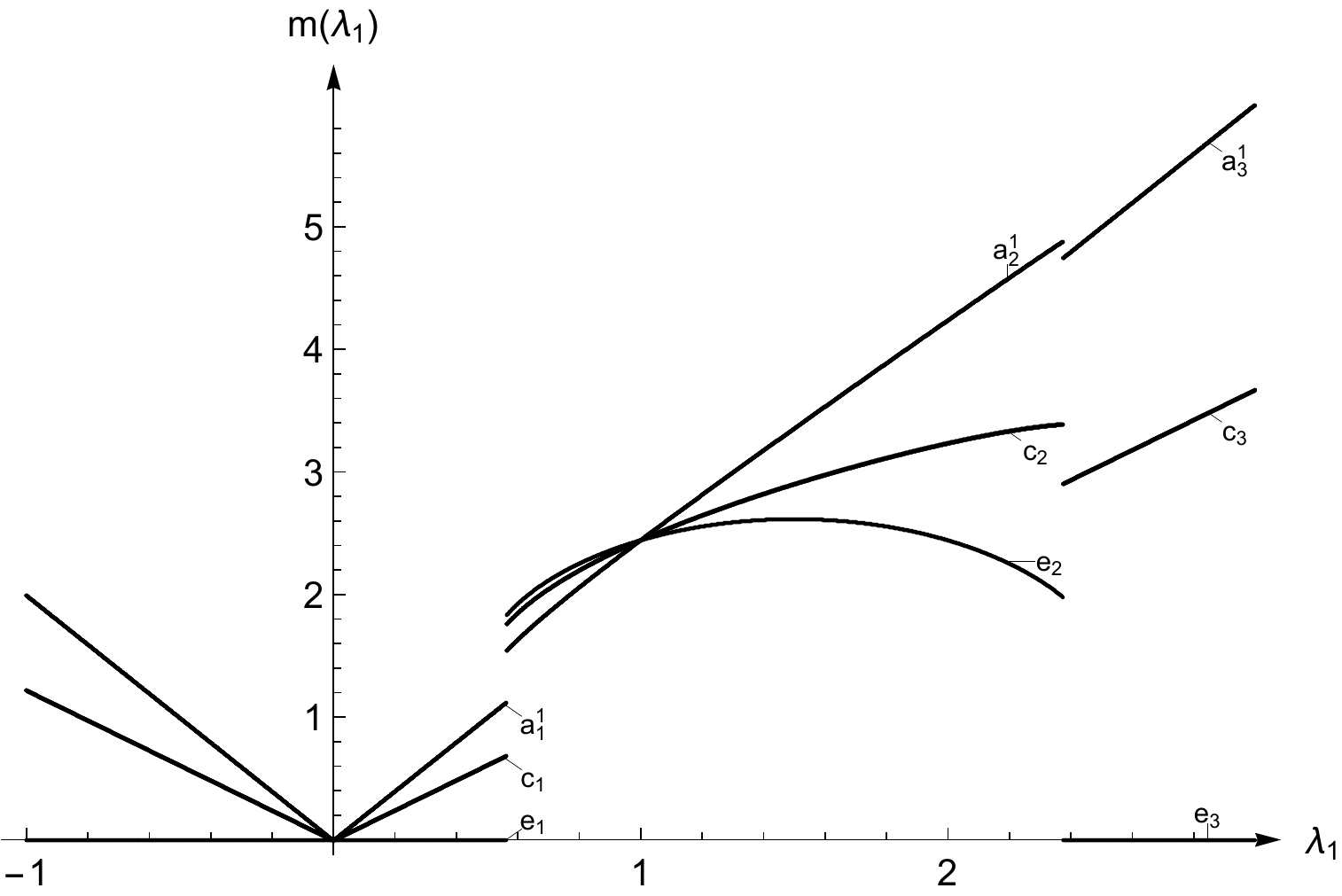}
\label{fig-1Md}
}
\caption{$\algA = M_2 \to \algB = M_3$: plots for $\lambda_1 \in [-1, 3]$.}
\label{fig-1Vd-1Md}
\end{figure}

\begin{figure}
\centering
\includegraphics[width=0.25\linewidth]{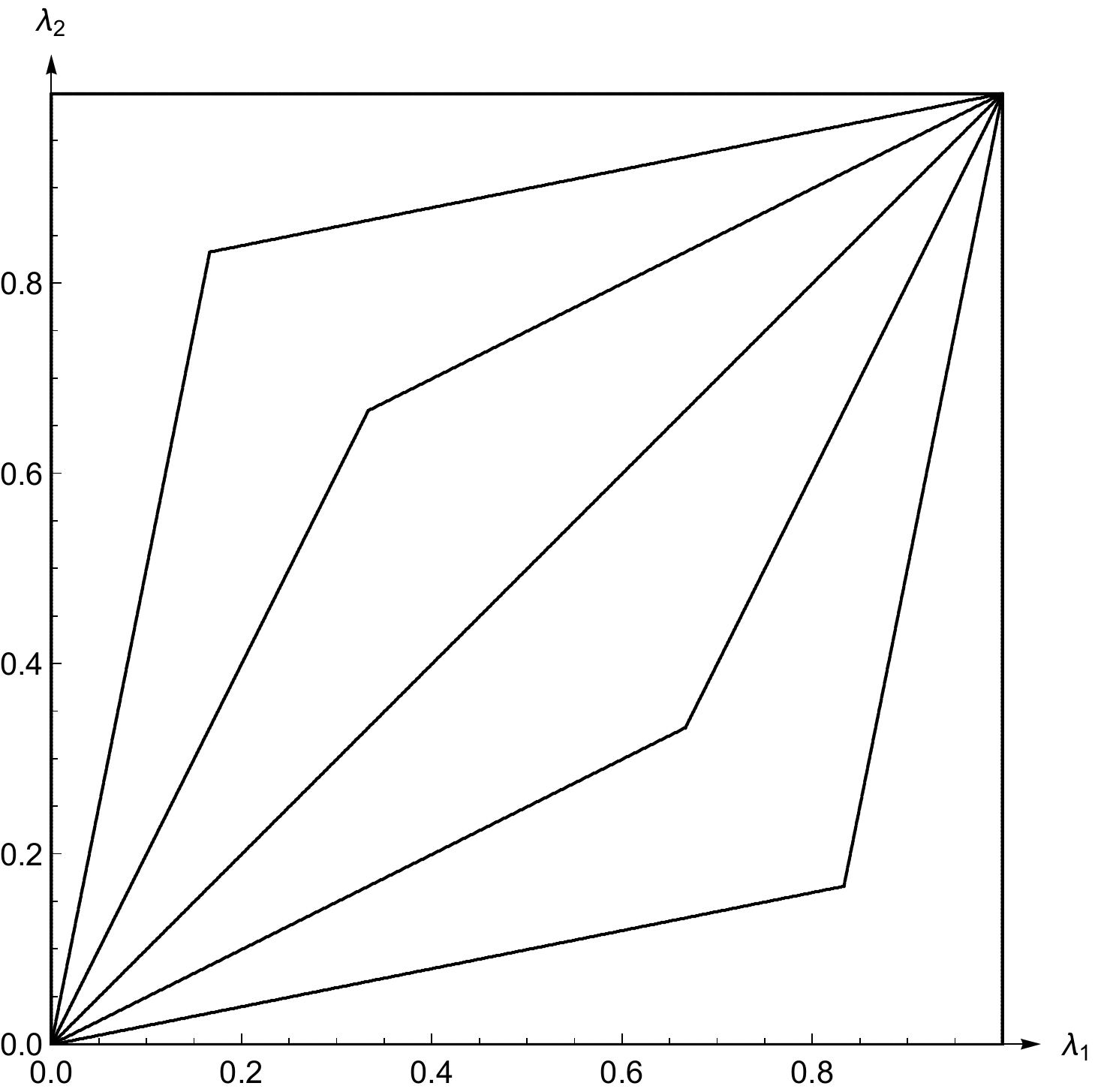}
\caption{The 7 lines for $(\lambda_1, \lambda_1) \in [0,1]^2$ along which computations of masses have been performed.}
\label{fig-lambda-square}
\end{figure}

\begin{figure}
\centering
\subfloat[Minimum values for the Higgs potential.]{%
\includegraphics[width=0.48\linewidth]{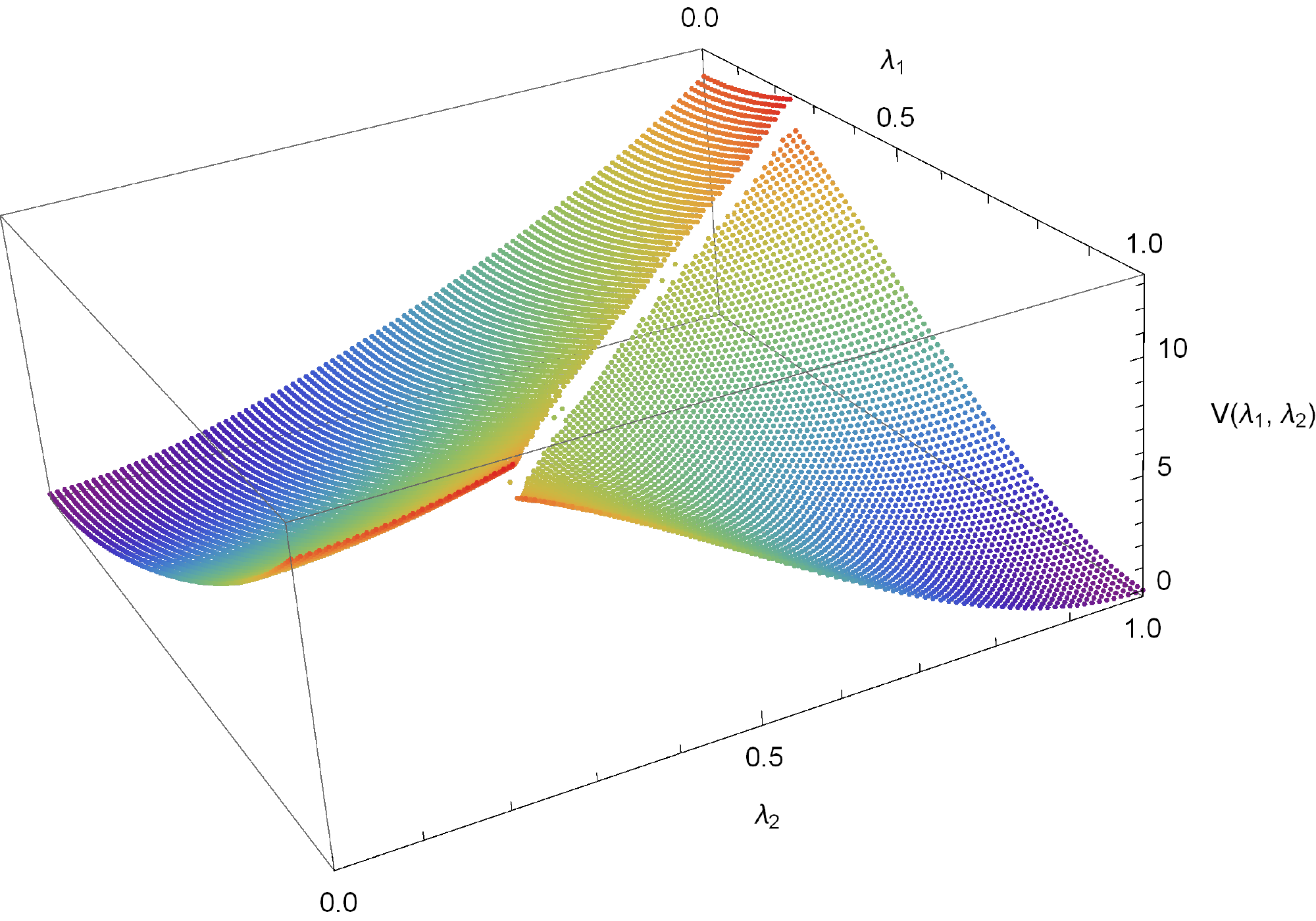}
\label{fig-2Vn}
}
\subfloat[Masses for the gauge fields.]{%
\includegraphics[width=0.48\linewidth]{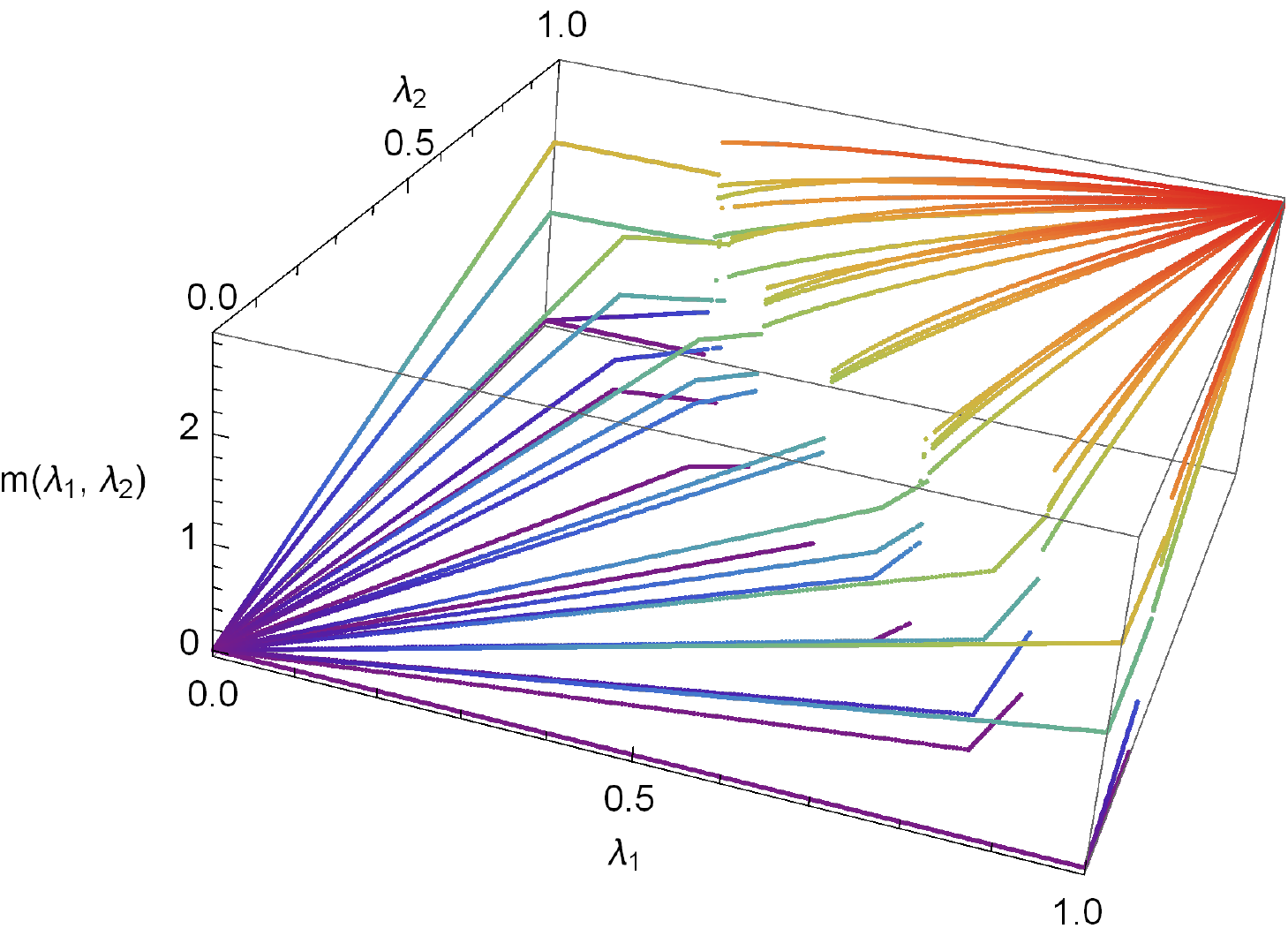}
\label{fig-2Mlc}
}
\caption{$\algA = M_2 \oplus M_2 \to \algB = M_4$: plots for the square $[0,1]^2$ in the plane $(\lambda_1, \lambda_2)$.}
\label{fig-2Vn-2Mlc}
\end{figure}

\begin{figure}
\centering
\subfloat[Minimum values for the Higgs potential.]{%
\includegraphics[width=0.48\linewidth]{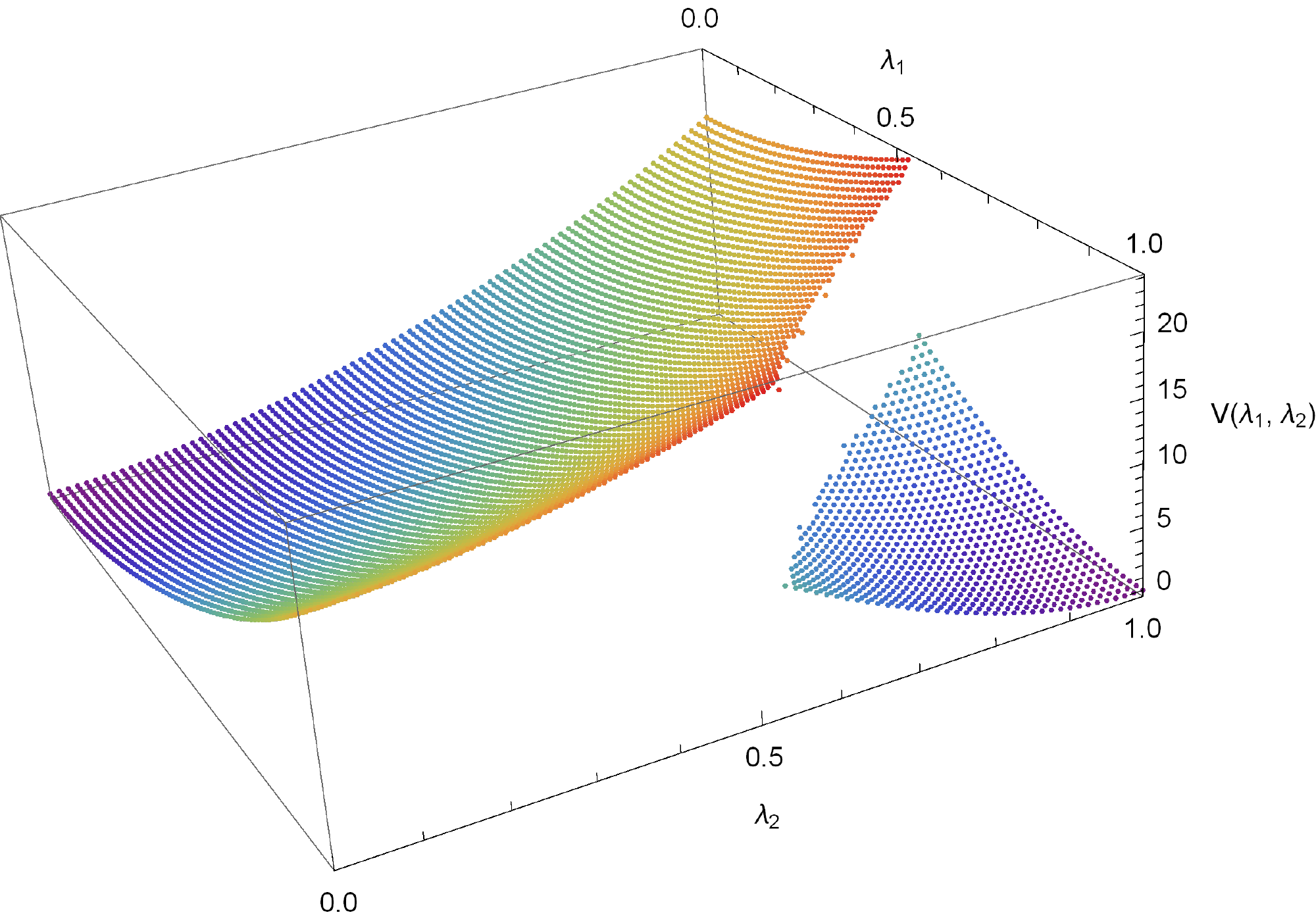}
\label{fig-3Vn}
}
\subfloat[Masses for the gauge fields.]{%
\includegraphics[width=0.48\linewidth]{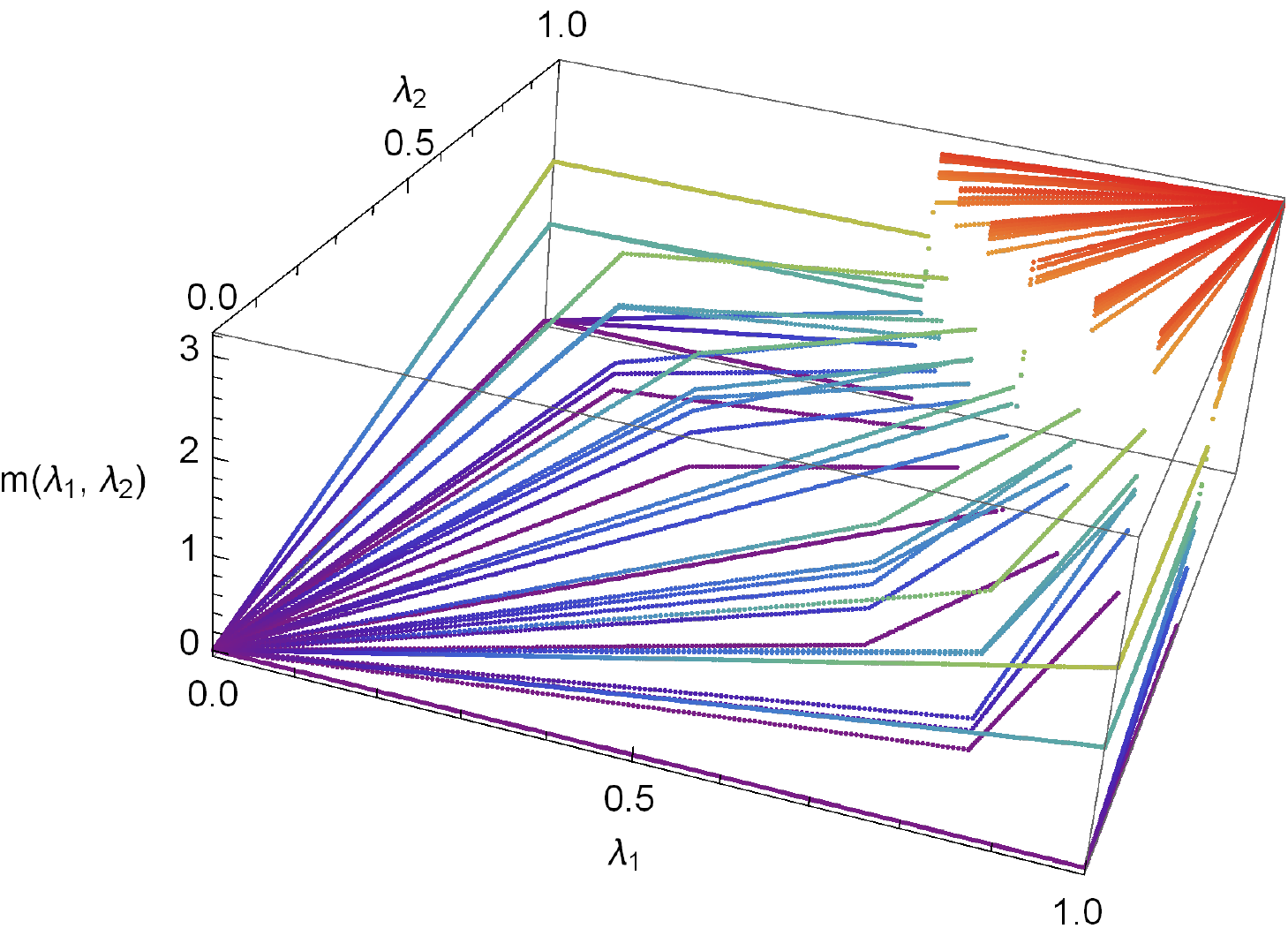}
\label{fig-3Mlc}
}
\caption{$\algA = M_2 \oplus M_2 \to \algB = M_5$: plots for the square $[0,1]^2$ in the plane $(\lambda_1, \lambda_2)$.}
\label{fig-3Vn-3Mlc}
\end{figure}

\begin{figure}
\centering
\subfloat[Minimum values for the Higgs potential.]{%
\includegraphics[width=0.48\linewidth]{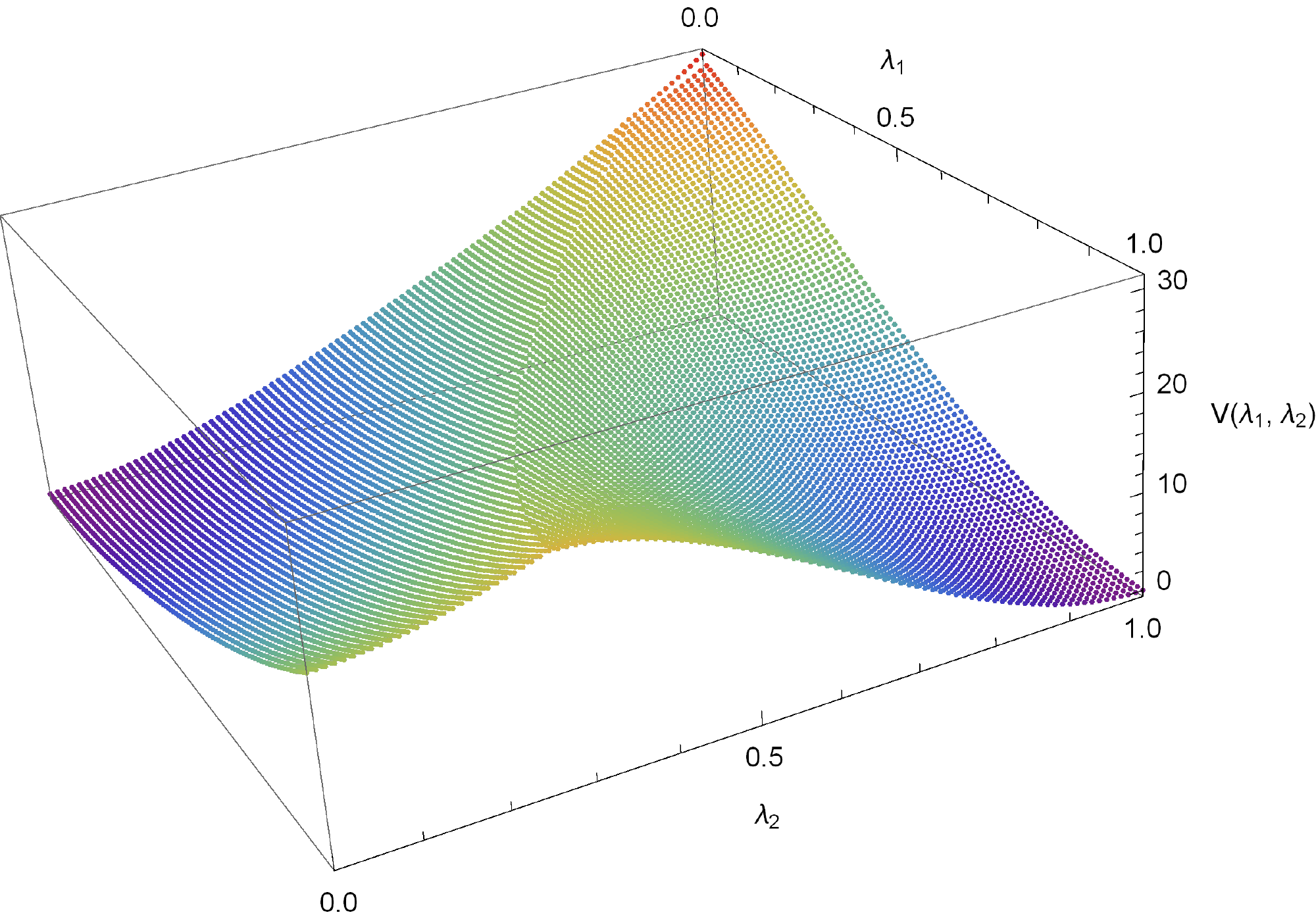}
\label{fig-4Vn}
}
\subfloat[Masses for the gauge fields.]{%
\includegraphics[width=0.48\linewidth]{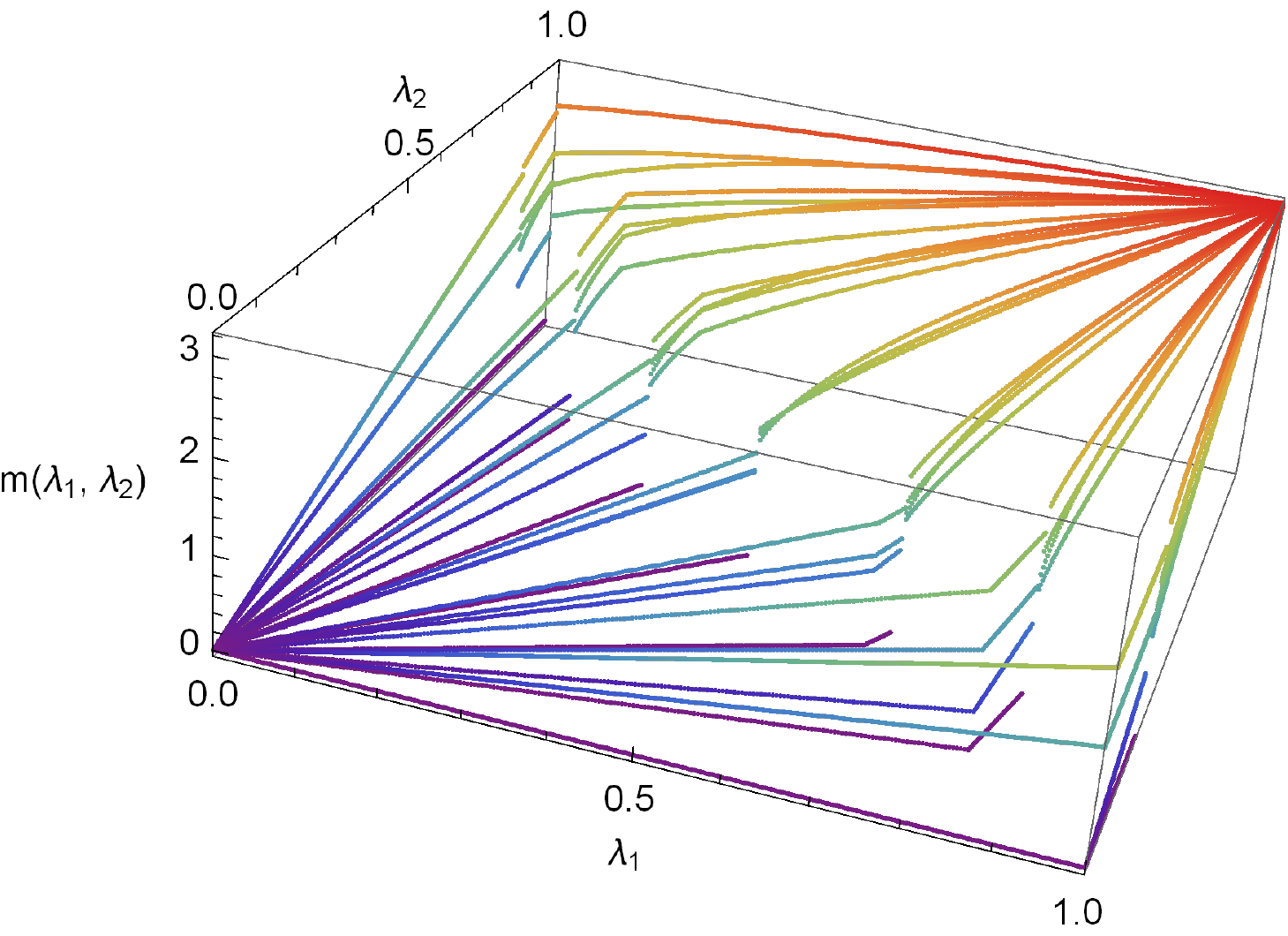}
\label{fig-4Mlc}
}
\caption{$\algA = M_2 \oplus M_3 \to \algB = M_5$: plots for the square $[0,1]^2$ in the plane $(\lambda_1, \lambda_2)$.}
\label{fig-4Vn-4Mlc}
\end{figure}

\begin{figure}
\centering
\subfloat[Minimum values for the Higgs potential with details in insert.]{%
\includegraphics[width=0.48\linewidth]{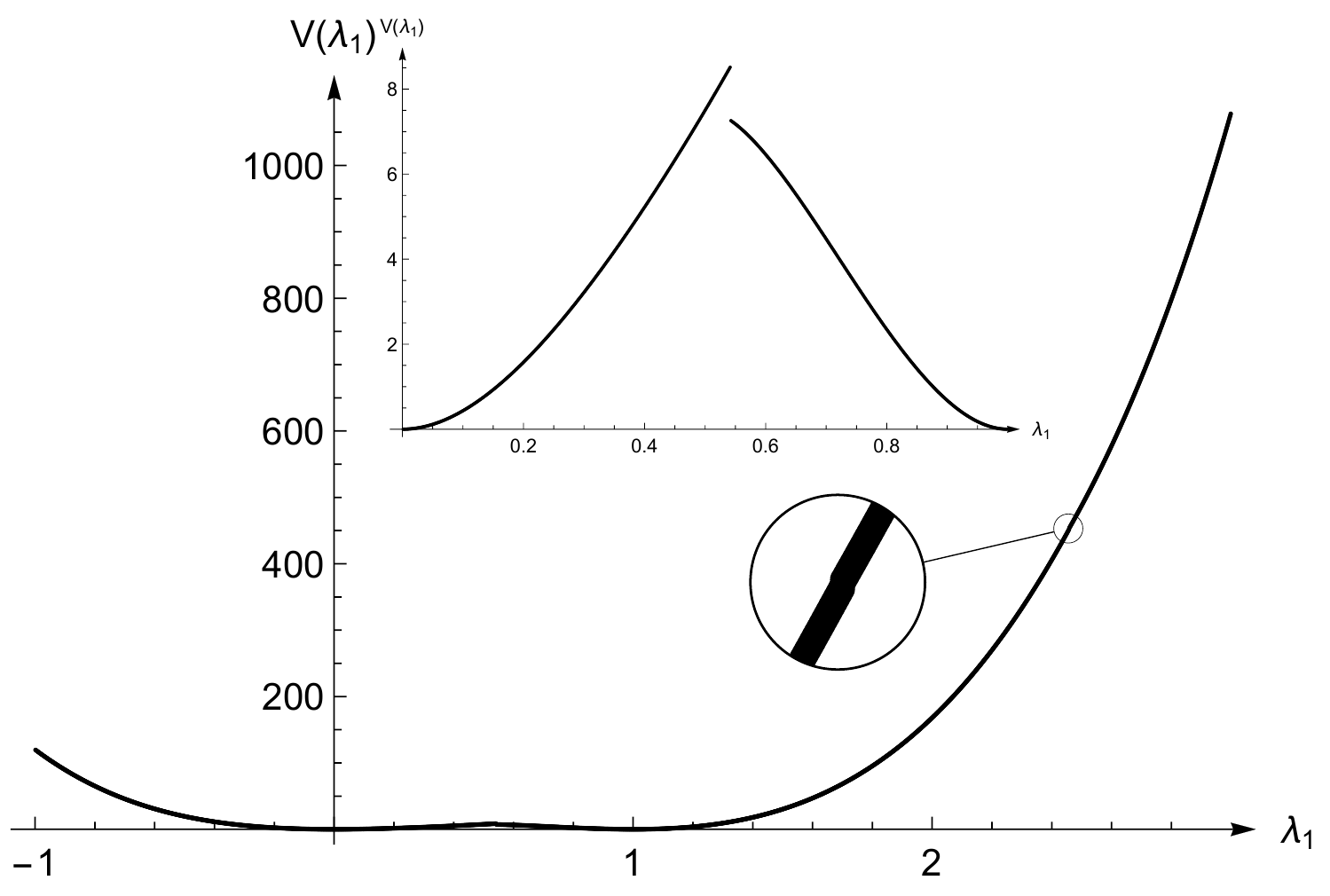}
\label{fig-2Vd}
}
\subfloat[Masses for the gauge fields.]{%
\includegraphics[width=0.48\linewidth]{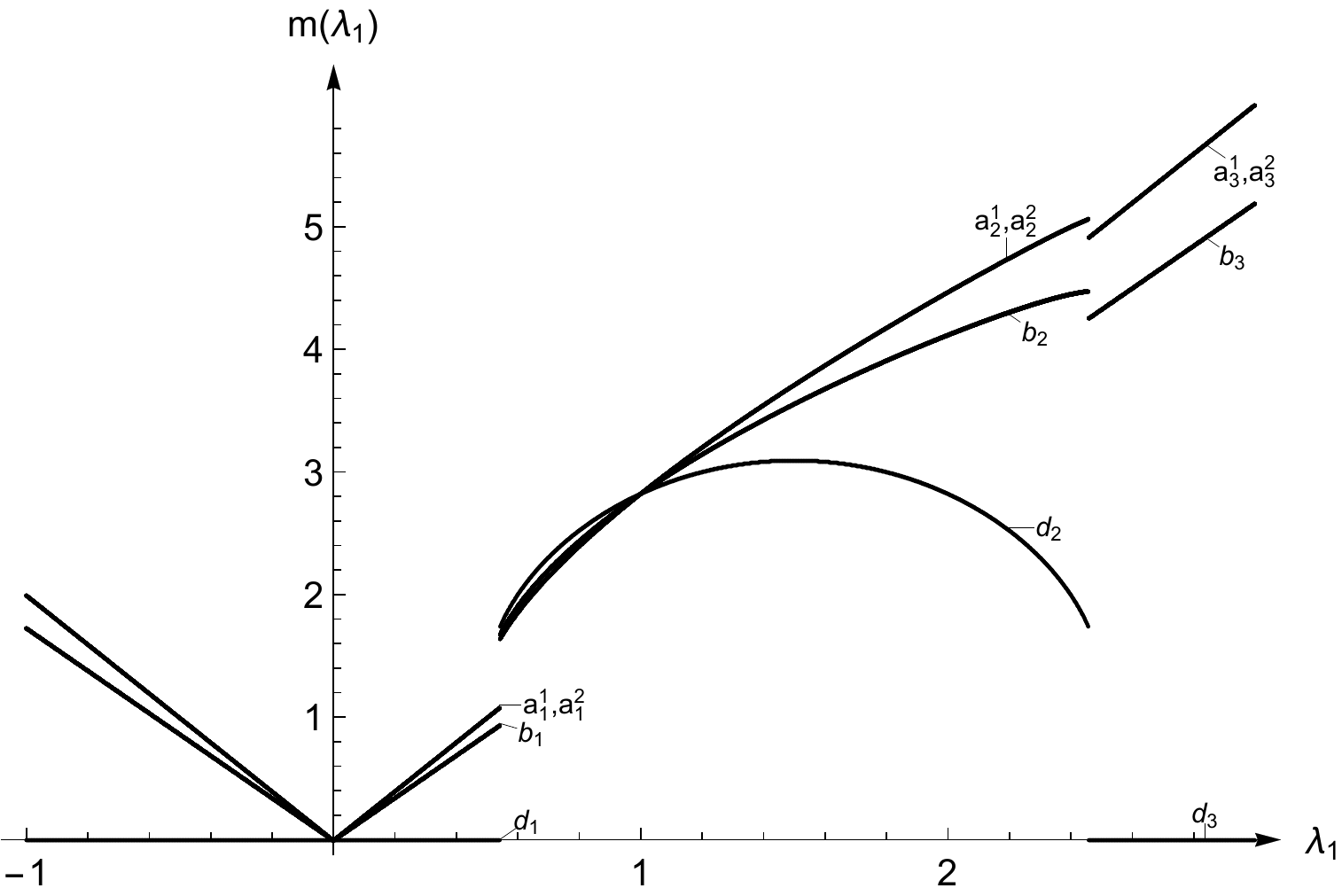}
\label{fig-2Md}
}
\caption{$\algA = M_2 \oplus M_2 \to \algB = M_4$: plots on the diagonal $(\lambda_1, \lambda_1)$ for $\lambda_1 \in [-1, 3]$.}
\label{fig-2Vd-2Md}
\end{figure}

\begin{figure}
\centering
\subfloat[Minimum values for the Higgs potential with details in insert.]{%
\includegraphics[width=0.48\linewidth]{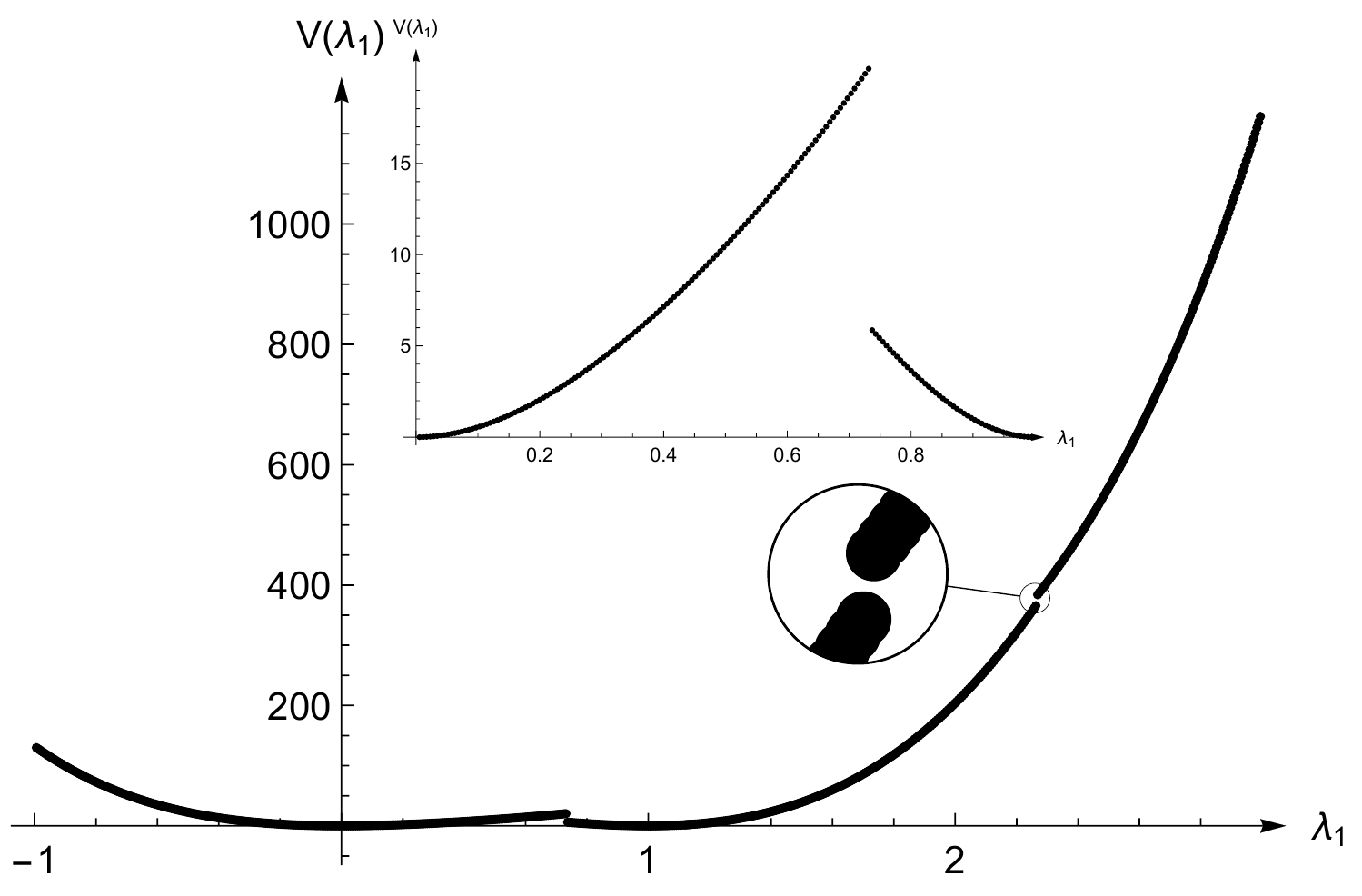}
\label{fig-3Vd}
}
\subfloat[Masses for the gauge fields.]{%
\includegraphics[width=0.48\linewidth]{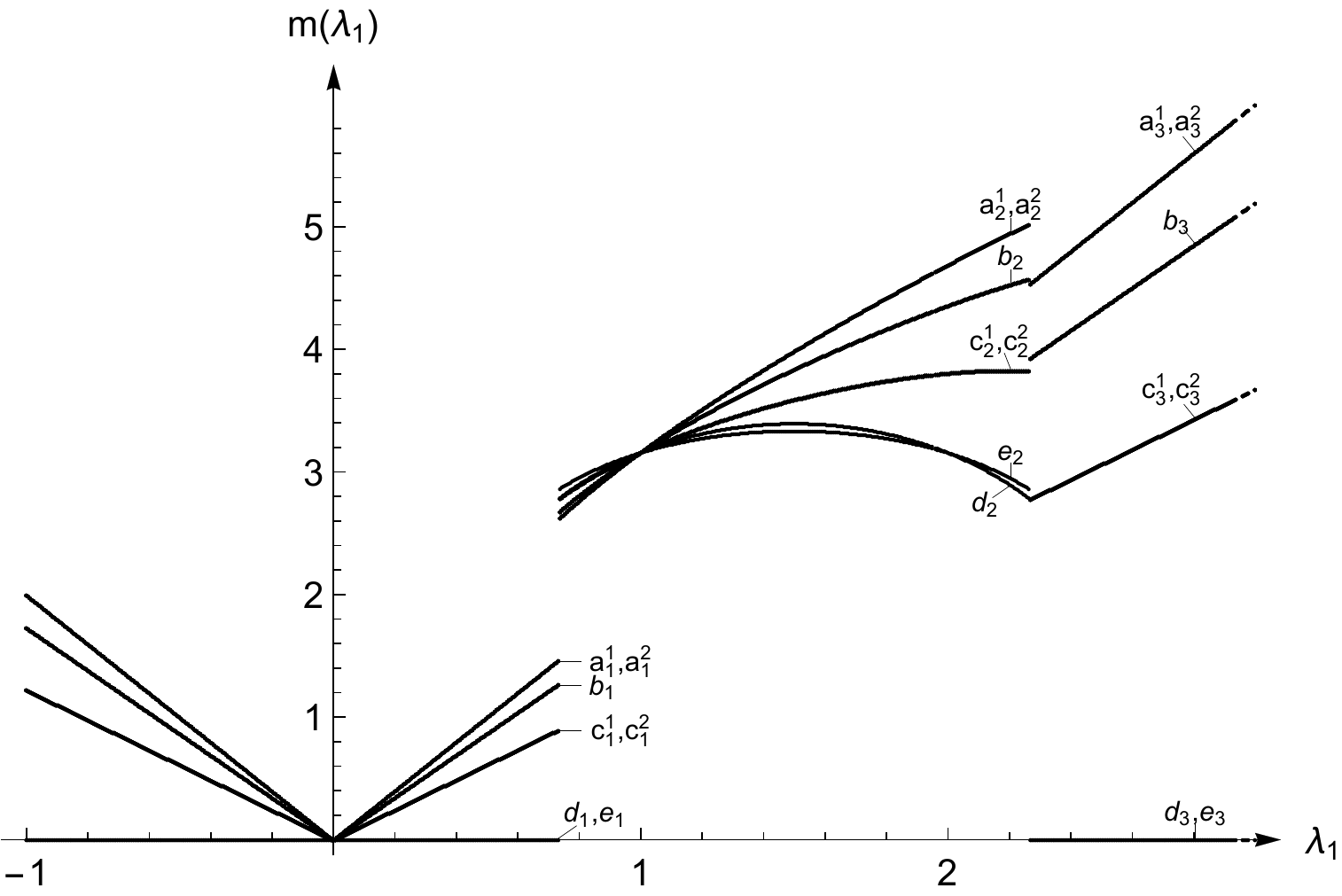}
\label{fig-3Md}
}
\caption{$\algA = M_2 \oplus M_2 \to \algB = M_5$: plots on the diagonal $(\lambda_1, \lambda_1)$ for $\lambda_1 \in [-1, 3]$.}
\label{fig-3Vd-3Md}
\end{figure}

\begin{figure}
\centering
\subfloat[Minimum values for the Higgs potential with details in insert.]{%
\includegraphics[width=0.48\linewidth]{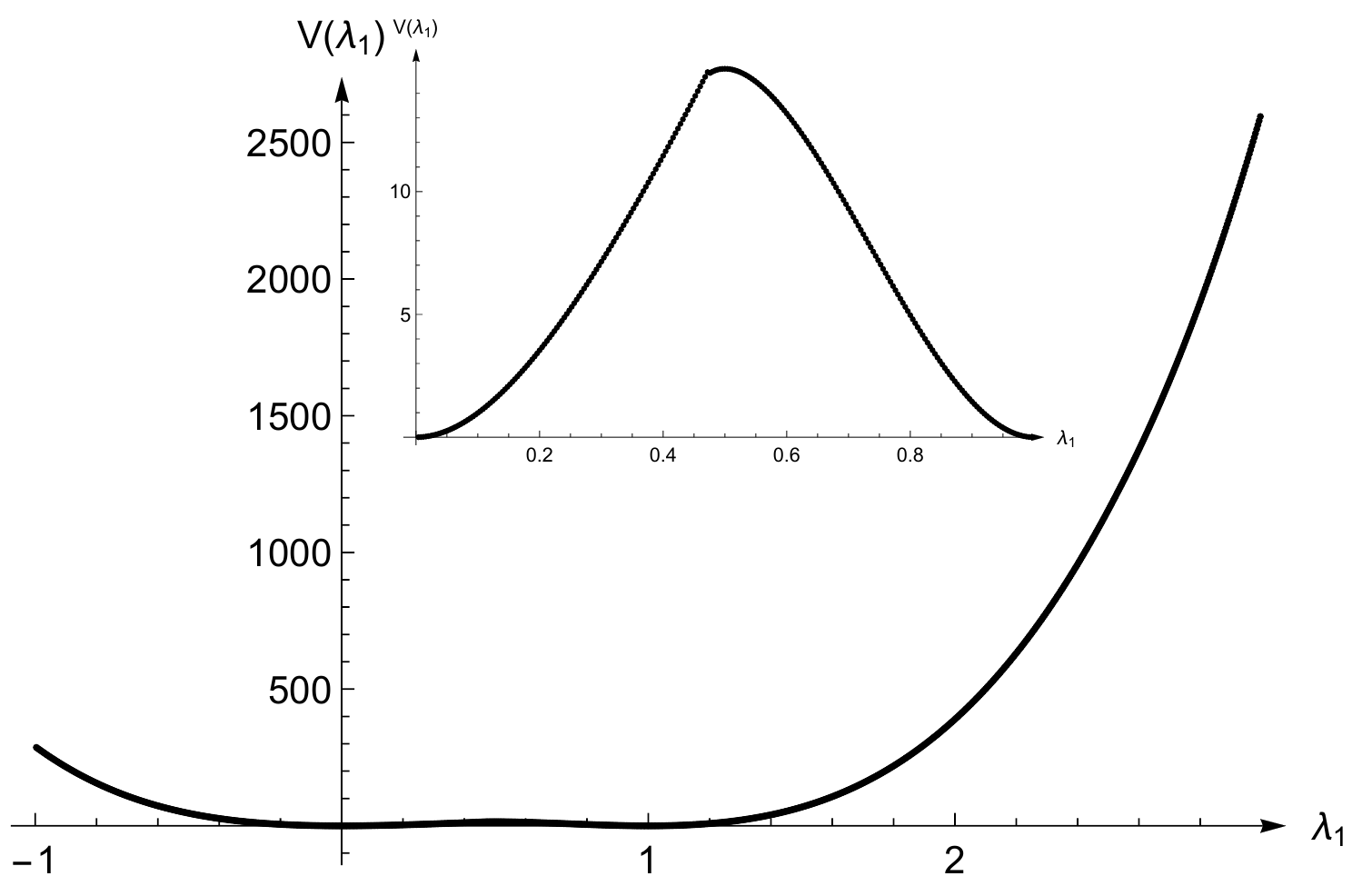}
\label{fig-4Vd}
}
\subfloat[Masses for the gauge fields.]{%
\includegraphics[width=0.48\linewidth]{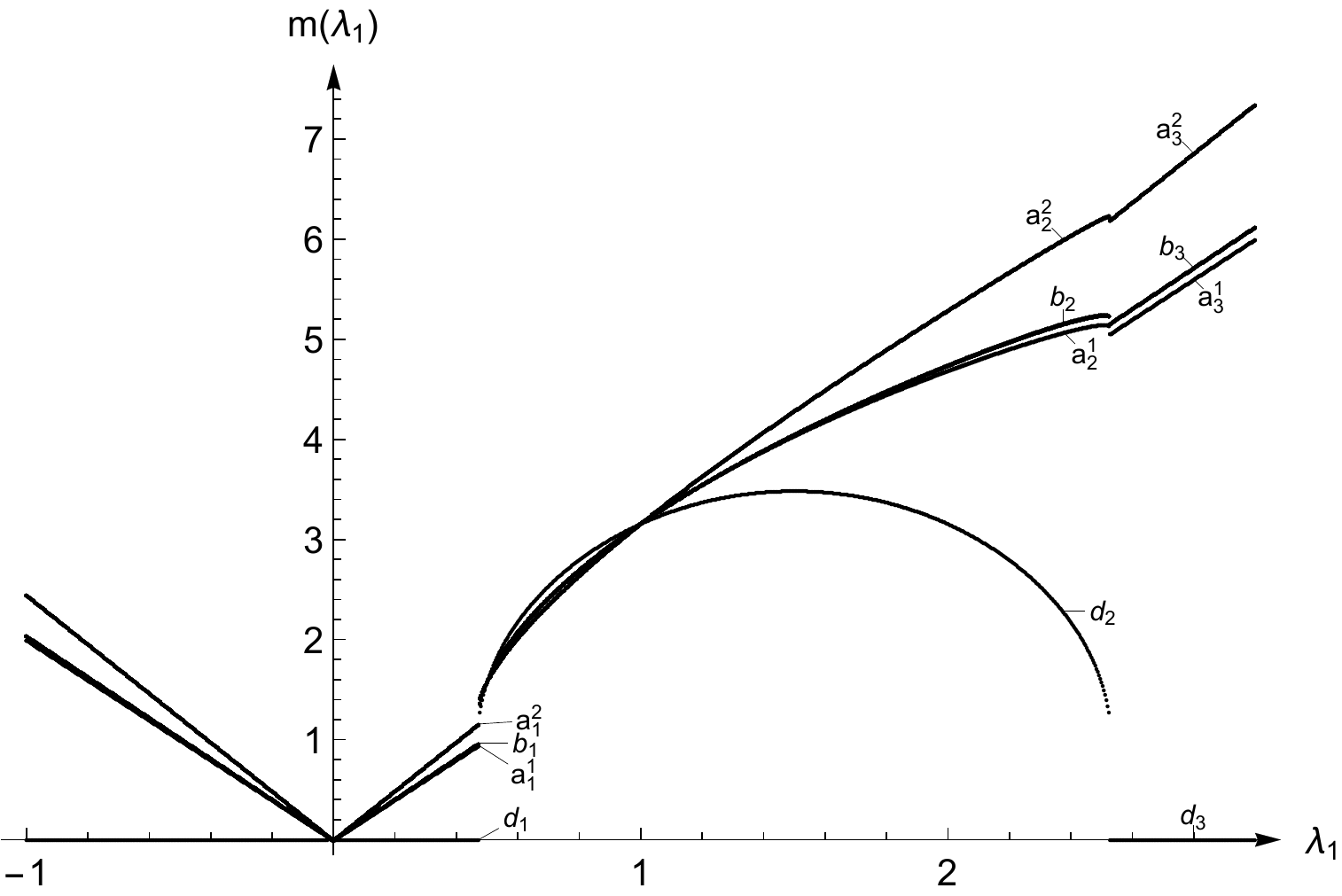}
\label{fig-4Md}
}
\caption{$\algA = M_2 \oplus M_3 \to \algB = M_5$: plots on the diagonal $(\lambda_1, \lambda_1)$ for $\lambda_1 \in [-1, 3]$.}
\label{fig-4Vd-4Md}
\end{figure}

\begin{figure}
\centering
\subfloat[$\algA = M_2 \oplus M_2 \to \algB = M_4$.]{%
\includegraphics[width=0.32\linewidth]{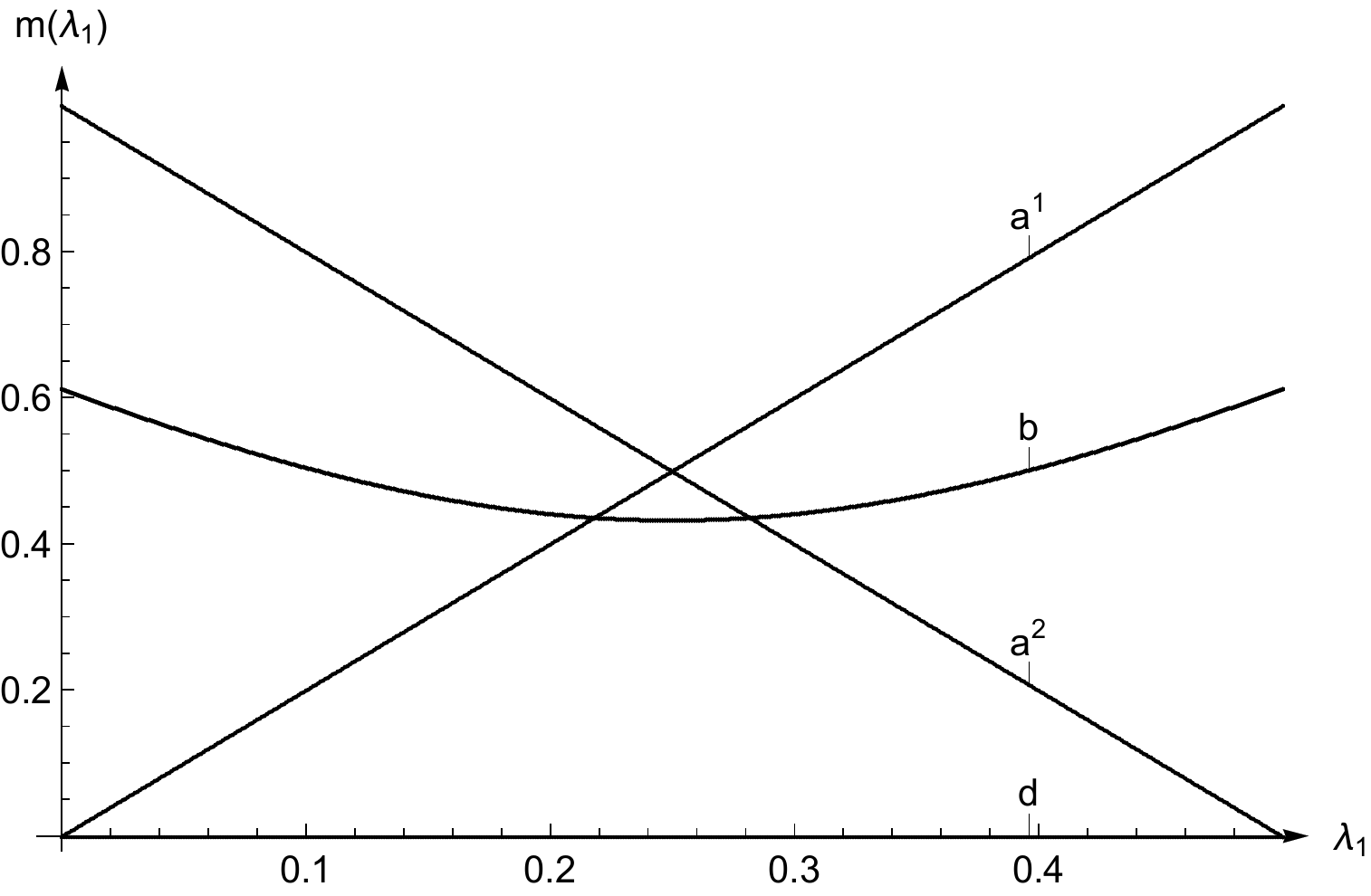}
\label{fig-2Mad_0.5}
}
\subfloat[$\algA = M_2 \oplus M_2 \to \algB = M_5$.]{%
\includegraphics[width=0.32\linewidth]{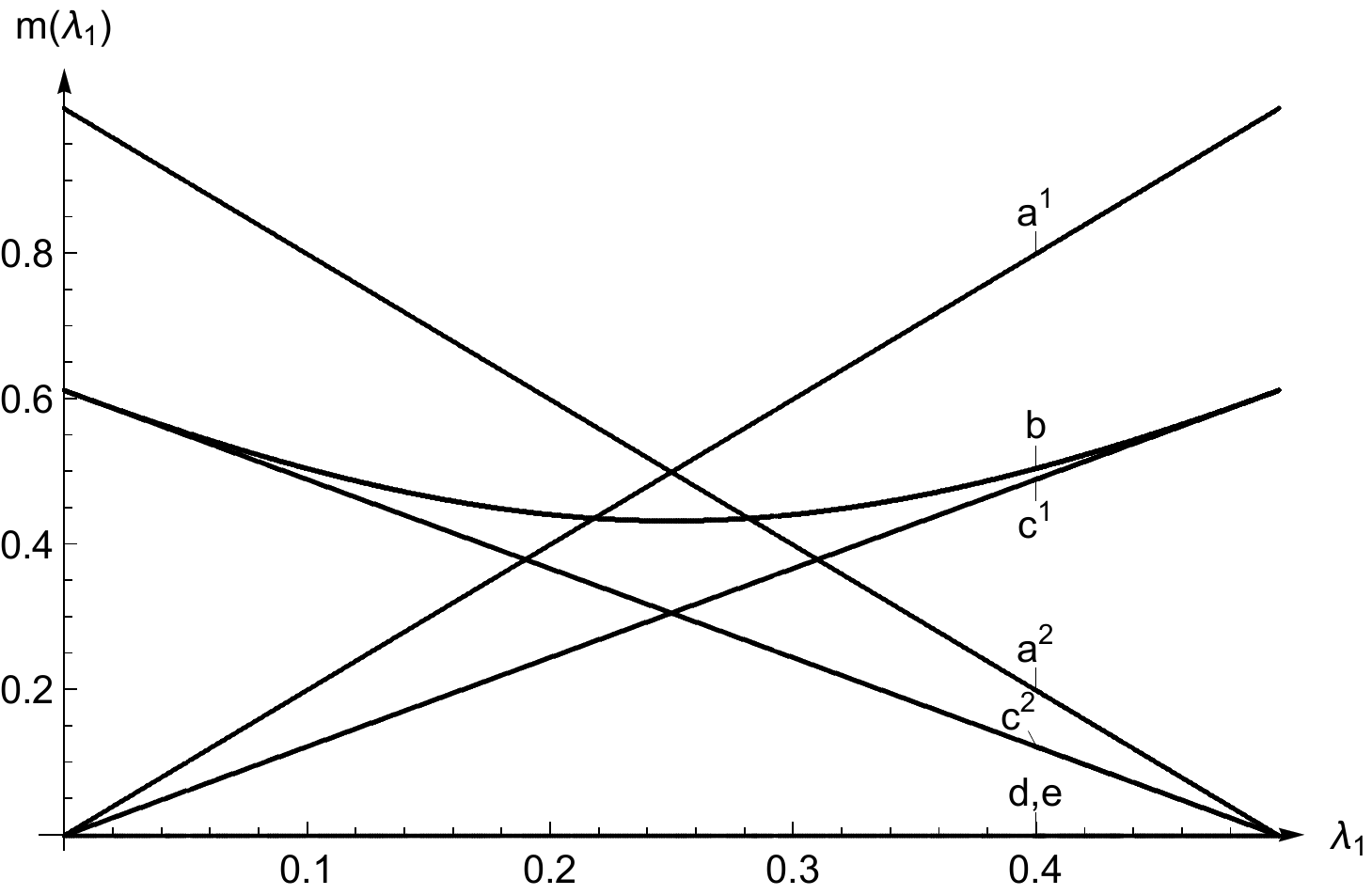}
\label{fig-3Mad_0.5}
}
\subfloat[$\algA = M_2 \oplus M_3 \to \algB = M_5$.]{%
\includegraphics[width=0.32\linewidth]{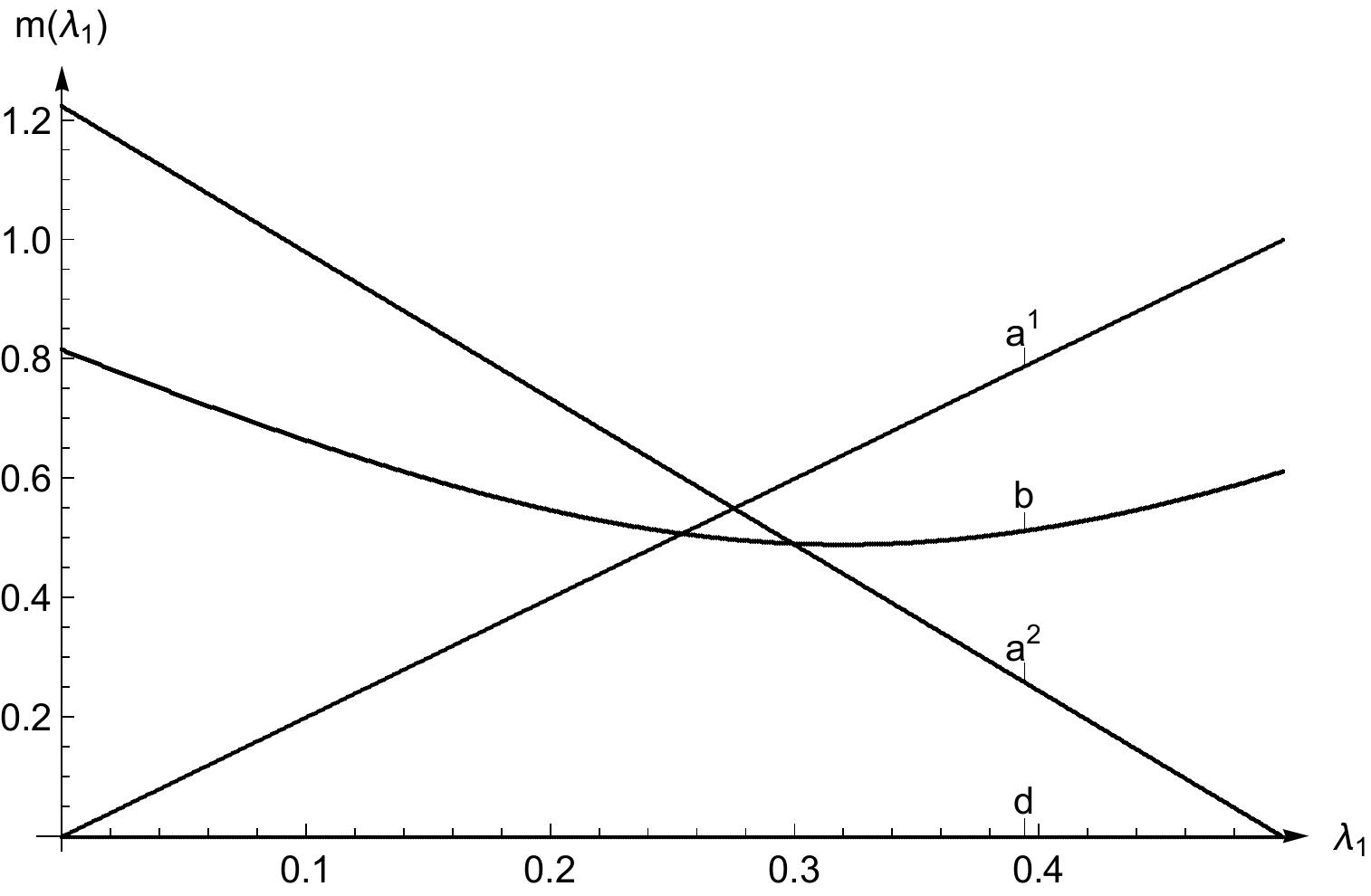}
\label{fig-4Mad_0.5}
}
\caption{Masses for the gauge fields along the line $\lambda_1 + \lambda_2 = 0.5$.}
\label{fig-2Mad-3Mad-4Mad}
\end{figure}

\end{document}